\documentclass[journal,10pt,twocolumn]{IEEEtran}
\usepackage{cite}
\usepackage{graphicx}  
\usepackage{tabulary}
\usepackage{amsfonts} 
\usepackage{amsmath}
\usepackage{chngpage}
\usepackage{makecell}
\usepackage{array}
\usepackage{amsthm}
\usepackage{booktabs}
\usepackage{amssymb}
\usepackage{multirow}
\usepackage{enumerate}
\usepackage{color}
\usepackage[dvipsnames]{xcolor}
\usepackage{mathrsfs}

\usepackage{float}
\usepackage{algorithm}  
\usepackage{algpseudocode}  
\usepackage{amsmath}  
\usepackage{url}

\usepackage{graphicx}
\usepackage{subfigure}
\usepackage[font=small]{caption}

\usepackage{mathtools}
\usepackage{enumitem}

\setlist[enumerate,1]{label=(\roman*)}

\newtheoremstyle{noparens}%
{}{}%
{\itshape}{}%
{\bfseries}{.}%
{ }%
{\thmname{#1}\thmnumber{ #2}\mdseries\thmnote{ #3}}

\theoremstyle{noparens}
\newtheorem{theorem}{Theorem}

\newtheorem{symmetry assumption}{Symmetry Assumption}
\newtheorem{lemma}{Lemma}
\newtheorem{corollary}{Corollary}
\newtheorem{remark}{Remark}
\newtheorem{proposition}{Proposition}
\newtheorem{definition}{Definition}
\newtheorem{property}{Property}

\begin{document}

\title{Indirect Lossy Source Coding with Observed Source Reconstruction: Nonasymptotic Bounds and Second-Order Asymptotics}

\author{
	\IEEEauthorblockN{Huiyuan Yang,~\IEEEmembership{Graduate Student Member,~IEEE}, Yuxuan Shi, Shuo Shao,~\IEEEmembership{Member,~IEEE}, Xiaojun Yuan,~\IEEEmembership{Senior Member,~IEEE}}
	\thanks{Huiyuan Yang and Xiaojun Yuan \emph{(corresponding author)} are with the National Key Laboratory on Wireless Communications, University of Electronic Science and Technology of China, Chengdu 611731, China (e-mail: hyyang@std.uestc.edu.cn; xjyuan@uestc.edu.cn).}
	\thanks{Yuxuan Shi is with the School of Cyber and Engineering, Shanghai Jiao Tong University, Shanghai 200240, China, and also with ZGC Institute of Ubiquitous-X Innovation and Applications, Beijing 100876, China (e-mail: ge49fuy@sjtu.edu.cn).}
	\thanks{Shuo Shao \emph{(corresponding author)} is with the School of Cyber and Engineering, Shanghai Jiao Tong University, Shanghai 200240, China (e-mail: shuoshao@sjtu.edu.cn).}
}

\maketitle
\IEEEpeerreviewmaketitle

\begin{abstract}
	This paper considers the joint compression of a pair of correlated sources, where the encoder is allowed to access only one of the sources. The objective is to recover both sources under separate distortion constraints for each source while minimizing the rate. This problem generalizes the indirect lossy source coding problem by also requiring the recovery of the observed source. In this paper, we aim to study the nonasymptotic and second-order asymptotic properties of this problem. Specifically, we begin by deriving nonasymptotic achievability and converse bounds valid for general sources and distortion measures. The source dispersion (Gaussian approximation) is then determined through asymptotic analysis of the nonasymptotic bounds. We further examine the case of erased fair coin flips (EFCF) and provide its specific nonasymptotic achievability and converse bounds. Numerical results under the EFCF case demonstrate that our second-order asymptotic approximation closely approximates the optimum rate at appropriately large blocklengths.

\end{abstract}


\section{Introduction}
\label{Introduction}
\IEEEPARstart{T}{his} paper considers the problem of jointly compressing two correlated sources in a lossy sense, where the encoder only has access to one of the sources. The objective is to maximize the compression ratio (or minimize the source coding rate) under given distortion constraints for both sources. Note that both the encoder and the decoder are aware of the joint distribution of the two sources. We refer to this setup as indirect lossy source coding with observed source reconstruction, as it involves additional recovery of the observed source compared to the standard indirect lossy source coding setup \cite{Dobrushin1962Information}.

This joint compression problem is of interest, as it arises when the observable data and its inferences (semantics) are both required. For example, in scenarios like autonomous driving, video surveillance, and AR/VR gaming, the raw data (e.g., videos) and its inferences (e.g., features) serve the human vision and machine processing tasks \cite{Gunduz2023Beyond, Yang2023Semantic}, respectively; in contexts such as medical image compression and remote sensing, the compressed images tend to fulfill varied requirements simultaneously: some require direct image recovery, and others rely on inferred information from the images \cite{Pardeep2020Versatile, Zhang2022Artificial}; in the compression of speech signals, the semantic meaning conveyed by text words aids in tasks such as automatic transcription or language translation, while the speech waveform provides additional information, such as the speaking habits of speakers, aiding tasks such as speaker recognition and verification \cite{Liu2022Indirect}. Some other examples can be found in \cite{Liu2022Indirect, Yuxuan2023Rate} and their references.

A block coding formulation of this joint compression problem was previously studied in \cite{Liu2021Rate}, where its first-order asymptotic properties were investigated. In this paper, we further consider a more general one-shot formulation, which allows us to study the nonasymptotic and second-order asymptotic behaviors of this problem. Specifically, we consider general sources, general distortion measures, and finite codebooks. It is important to note that the finite codebook size cannot guarantee a zero joint excess distortion probability, particularly at low rates, where this probability refers to the probability that at least one of the recovery distortions of the two sources exceeds its corresponding threshold.
Consequently, the joint compression aims to efficiently encode the observable source so that the decoder can recover both sources with a tolerable joint excess distortion probability based on the codeword. Our nonasymptotic modeling is more aligned with practical compression systems with delay, complexity, or storage constraints and, therefore, has the potential to yield results that more accurately reflect the fundamental limits of real-world systems.

In the remainder of this section, we provide a brief overview of related works in Subsection \ref{Related_Works_and_Discussions}, followed by an outline of the contributions and organization of this paper in Subsection \ref{Contributions_and_Organization}.

\subsection{Related Works and Discussions}
\label{Related_Works_and_Discussions}

The paradigm of simultaneously recovering observable and hidden sources was first used to analyze semantic communication problems \cite{Liu2021Rate, Liu2022Indirect, Yuxuan2023Rate, Shi2023Excess, Stavrou2023Role}. Specifically, the authors in \cite{Liu2022Indirect} and \cite{Liu2021Rate} proposed a semantic source model comprising an observable source and a correlated hidden source, with each source dedicated to memorylessly generating data and its embedded semantics, respectively. As semantics are typically inferred from data and cannot be directly observed, the authors considered a lossy source coding problem that seeks to simultaneously recover data and its embedded semantics under their respective distortion constraints by encoding only the observable source. The corresponding semantic rate-distortion function was then derived. Note that the problem considered in \cite{Liu2021Rate} and \cite{Liu2022Indirect} is an asymptotic version of our problem and inspires our studies. While \cite{Liu2021Rate} and \cite{Liu2022Indirect} focused on the point-to-point source coding scenario, \cite{Yuxuan2023Rate} and \cite{Shi2023Excess} expanded the scope by applying this semantic modeling approach to the multi-terminal source coding and joint source-channel coding scenarios, respectively. \cite{Stavrou2023Role} proposed a Blahut–Arimoto type algorithm to compute the so-called semantic rate-distortion function. \cite{Li2024Fundamental} further gives some analytical properties of the semantic rate-distortion function and introduces a neural network designed for estimating the semantic rate-distortion function from samples. However, the above works focus on first-order asymptotics, in contrast to the nonasymptotic and second-order asymptotics we will investigate.

From another line of research, finite blocklength and second-order analysis of source coding traces back to the seminal work of Strassen \cite{Volker1962Asymptotische} in 1962, in which he investigated the scenario of almost lossless source coding. The second-order analysis of lossy source coding can date back to Kontoyiannis \cite{Kontoyiannis2000Pointwise} in 2000 (see also \cite{Kontoyiannis2002Arbitrary}) in the context of variable-rate coding, where the distortion is required to remain below a given threshold. This is where the dispersion of lossy source coding first appears, under the name ``minimal coding variance''. In the early 2010s, \cite{Ingber2011Dispersion} and \cite{Kostina2012Fixed} revisited this research field under the excess-distortion probability constraint. Specifically, using type-based approaches, the authors of \cite{Ingber2011Dispersion} rediscovered the dispersion of lossy source coding with fixed excess-distortion probabilities for finite-alphabet sources. The case of i.i.d. quadratic Gaussian source was also treated in \cite{Ingber2011Dispersion}. In \cite{Kostina2012Fixed}, the authors proposed several general nonasymptotic achievability and converse bounds for lossy source coding. Moreover, they derived second-order asymptotics that are valid for sources with abstract alphabets through asymptotic analysis of the nonasymptotic bounds. Following \cite{Ingber2011Dispersion} and \cite{Kostina2012Fixed}, various lossy source coding scenarios have been studied in the finite blocklength regime, such as in \cite{Kostina2013LossyJoint, Kostina2016Nonasymptotic, No2016Strong, Zhou2017Second, Zhou2019Refined, Zhou2019NonAsymptotic, Li2018Strong}. Specifically, lossy joint source-channel coding and noisy lossy source coding in the finite blocklength regime were considered in \cite{Kostina2013LossyJoint} and \cite{Kostina2016Nonasymptotic}, respectively, where their nonasymptotic bounds and second-order asymptotics were studied. In \cite{No2016Strong} and \cite{Zhou2017Second}, the second-order asymptotics of the successive refinement problem were investigated under separate and joint excess-distortion probability constraints, respectively. \cite{Zhou2019Refined} established the second-order, moderate, and large deviation asymptotics of the mismatched code in \cite{Lapidoth1997On}. \cite{Zhou2019NonAsymptotic} provided non-asymptotic converse bounds (based on distortion-tilted information) and refined asymptotics for the Kaspi problem and the Fu-Yeung problem, respectively. \cite{Li2018Strong} derived the strong functional representation lemma and established achievability bounds for nonasymptotic variable-length lossy source coding, multiple description coding, and Gray-Wyner system based on this lemma. We recommend \cite{Dembo2002Source}, \cite{Vincent2014Asymptotic}, and \cite{Zhou2023Finite} for a comprehensive survey.

It is worth emphasizing that the trade-off between distortion constraints, formalized by the joint excess-distortion probability constraint in our formulation, also appears in \cite{Zhou2017Second} and \cite{Zhou2019NonAsymptotic}. What distinguishes our distortion trade-off from theirs is that ours is coupled with the irreducible randomness introduced by the unobservable hidden source, which implies that we have to characterize the impact of different trade-off patterns on the rate (or the joint excess-distortion probability) in the presence of irreducible randomness. This coupling also differentiates our work from \cite{Kostina2016Nonasymptotic}, where the distortion trade-off does not exist. As will be seen, we will generalize the methods in \cite{Kostina2012Fixed} and \cite{Kostina2016Nonasymptotic} to handle this coupling of dual challenges.




\subsection{Contributions and Organization}
\label{Contributions_and_Organization}

The main contributions of this paper are listed as follows:
\begin{itemize}
	
	\item We derive nonasymptotic achievability and converse bounds for the problem of indirect lossy source coding with observed source reconstruction. These bounds are applicable to general sources and distortion measures.
	
	\item For stationary memoryless sources and separable distortion measures, we find the dispersion (second-order asymptotics) of indirect lossy source coding with observed source reconstruction.
	
	\item We derive the rate-distortion function for the case of erased fair coin flips and establish nonasymptotic achievability and converse bounds specifically tailored to this scenario.
	
\end{itemize}
Numerical results demonstrate that our second-order asymptotic results effectively approximate the optimum rate at given blocklengths.

The remainder of this paper is organized as follows. In Section \ref{Preliminaries}, we introduce basic definitions and properties. In Sections \ref{Nonasymptotic_Achievability_Bounds} and \ref{Nonasymptotic_Converse_Bounds}, we derive the general nonasymptotic achievability and converse bounds, respectively. The second-order asymptotics are introduced in Section \ref{Asymptotic_Analysis}. The analysis results under the case of erased fair coin flips are derived and depicted in Section \ref{Case_Study}. Finally, Section \ref{Conclusion} concludes the paper.

\section{Preliminaries}
\label{Preliminaries}

\subsection{The General and Block Settings}


\begin{figure*}
	\centering
	\subfigure[An $(M, d_s, d_x, \epsilon)$ code.]{
		\begin{minipage}{0.6\textwidth}
			\includegraphics[width=\textwidth]{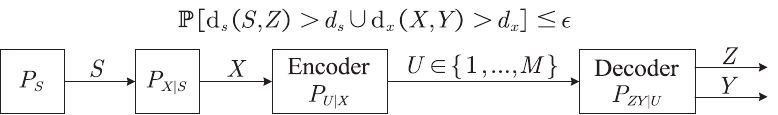} \\
	\end{minipage}}
	\subfigure[A $(k, M, d_s, d_x, \epsilon)$ code.]{
		\begin{minipage}{0.6\textwidth}
			\includegraphics[width=\textwidth]{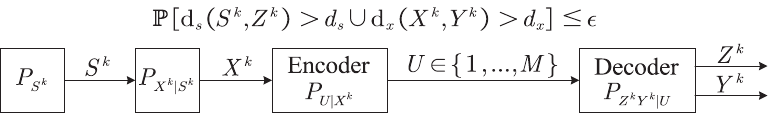} \\
	\end{minipage}}
	\caption{Indirect lossy source coding with observed source reconstruction in the nonasymptotic regime.} 
	\label{JDSLC_figs}
\end{figure*}

Consider correlated sources $(S, X)\sim P_{SX}=P_{S}P_{X|S}$, where we are given the distribution $P_S$ on alphabet $\mathcal{S}$ and the conditional distribution $P_{X|S}: \mathcal{S} \to \mathcal{X}$. We call $S$ the hidden source and $X$ the observable source. Two distortion measures, $\textsf{d}_s: \mathcal{S} \times \widehat{\mathcal{S}} \mapsto [0, +\infty)$ and $\textsf{d}_x: \mathcal{X} \times \widehat{\mathcal{X}} \mapsto [0, +\infty)$, have been provided for quantifying the recovery distortion of $S$ and $X$, respectively. For the sake of rigor, we restrict our discussion to the case where $\mathcal{S}$,  $\widehat{\mathcal{S}}$, $\mathcal{X}$, and $\widehat{\mathcal{X}}$ are all finite sets. Under the general setting above, we define the $(M, d_s, d_x, \epsilon)$ code as follows.
\begin{definition}
	(General):
	An $(M, d_s, d_x, \epsilon)$ code 
	for $\{\mathcal{S}, \mathcal{X}, \widehat{\mathcal{S}}, \widehat{\mathcal{X}}, P_{SX}, \textsf{d}_s: \mathcal{S} \times \widehat{\mathcal{S}} \mapsto [0, +\infty), \textsf{d}_x: \mathcal{X} \times \widehat{\mathcal{X}} \mapsto [0, +\infty)\}$ 
	is a pair of random mappings $P_{U|X}: \mathcal{X} \to \{1, \dots, M\}$ and $P_{ZY|U}: \{1,\dots, M\} \to \widehat{\mathcal{S}}\times \widehat{\mathcal{X}}$ such that the joint excess distortion probability satisfies $\mathbb{P}[\textsf{d}_s(S,Z)>d_s \cup \textsf{d}_x(X,Y)>d_x] \leq \epsilon$.
\end{definition}


The block setting is a special case of the general setting, induced by the specialization that the alphabets $\mathcal{S}$, $\mathcal{X}$, $\widehat{\mathcal{S}}$, and $\widehat{\mathcal{X}}$ are all $k$-fold Cartesian products. Note that $k$ is referred to as the blocklength. In block setting with $\mathcal{S} = \mathcal{M}^k$, $\mathcal{X}=\mathcal{A}^k$, $\widehat{\mathcal{S}} = \hat{\mathcal{M}}^k$, and $\widehat{\mathcal{X}}=\hat{\mathcal{A}}^k$, the $(k, M, d_s, d_x, \epsilon)$ code is defined as follows.
\begin{definition}
	(Block):
	A $(k, M, d_s, d_x, \epsilon)$ code is an $(M, d_s, d_x, \epsilon)$ code for $\{\mathcal{M}^k, \mathcal{A}^k, \hat{\mathcal{M}}^k, \hat{\mathcal{A}}^k, P_{S^kX^k}, \textsf{d}_s: \mathcal{M}^k \times \hat{\mathcal{M}}^k \mapsto [0, +\infty), \textsf{d}_x: \mathcal{A}^k \times \hat{\mathcal{A}}^k \mapsto [0, +\infty)\}$.
\end{definition}
The problem of nonasymptotic indirect lossy source coding with observed source reconstruction in the general and block settings are illustrated in Fig. \ref{JDSLC_figs}.

\subsection{Tilted Information}
The noisy rate-distortion function is defined as \cite{Liu2021Rate}
\begin{subequations}
	\label{rate_distortion_function}
	\begin{align}
	R_{S,X}(d_s,d_x) \triangleq \min_{P_{ZY|X}}\ &I(X;Z,Y) \label{rate_distortion_function_obj} \\
	\mathrm{s.t.}\
	&\mathbb{E}\left[\bar{\textsf{d}}_s(X,Z)\right] \leq d_s, \\
	&\mathbb{E}\left[\textsf{d}_x(X,Y)\right] \leq d_x, 
	\end{align}
\end{subequations}
where $\bar{\textsf{d}}_s: \mathcal{X} \times \widehat{\mathcal{S}} \mapsto [0, +\infty)$ is given by
\begin{equation}
\bar{\textsf{d}}_s(x,z)\triangleq \mathbb{E}[\textsf{d}_s(S, z)| X=x].
\end{equation}
Note that the minimum in \eqref{rate_distortion_function} can be achieved by some $P_{ZY|X}$ only for maximum admissible distortions $(d_s,d_x)$ satisfying $d_s \geq d_{s,\min}\triangleq \mathbb{E}[\min_{z \in \widehat{\mathcal{S}}} \bar{\textsf{d}}_s(X,z)]$ and $d_x \geq d_{x,\min}\triangleq \mathbb{E}[\min_{y \in \widehat{\mathcal{X}}} \textsf{d}_x(X,y)]$. Thus $R_{S,X}(d_s,d_x)$ is non-empty if and only if $(d_s,d_x) \in \mathcal{D}_{\mathrm{adm}} \triangleq \{(d_s,d_x): d_s \geq d_{s,\min}, d_x \geq d_{x,\min}\}$. We further define
\begin{align}
R_{\tilde{X}}(d_s,d_x) \triangleq R_{S,\tilde{X}}(d_s,d_x),
\end{align}
where $P_{S\tilde{X}} = P_{\tilde{X}}P_{S|X}$. Clearly, $R_{X}(d_s,d_x) = R_{S,X}(d_s,d_x)$.

Denote by $\mathcal{P}^{\star}(d_s,d_x)$ the set of optimal solutions of problem \eqref{rate_distortion_function} associated with $(d_s,d_x) \in \mathcal{D}_{\mathrm{adm}}$.  Define sets
\begin{align}
\mathcal{D}_{sx} \triangleq \{&(d_s,d_x) \in \mathcal{D}_{\mathrm{adm}}: \forall P_{Z^\star Y^\star|X} \in \mathcal{P}^{\star}(d_s,d_x),\nonumber \\ &\mathbb{E}\left[\bar{\textsf{d}}_s(X,Z^\star)\right] = d_s \textrm{ and }  \mathbb{E}\left[\textsf{d}_x(X,Y^\star)\right] = d_x\},
\end{align}
\begin{align}
\mathcal{D}_{\bar{s}x} \triangleq \{&(d_s,d_x)\in \mathcal{D}_{\mathrm{adm}}: \exists P_{Z^\star Y^\star|X} \in \mathcal{P}^{\star}(d_s,d_x) \textrm{ such that} \nonumber \\ &\mathbb{E}\left[\bar{\textsf{d}}_s(X,Z^\star)\right] < d_s;\ \forall P_{Z^\star Y^\star|X} \in \mathcal{P}^{\star}(d_s,d_x), \nonumber \\ &\mathbb{E}\left[\textsf{d}_x(X,Y^\star)\right] = d_x\},
\end{align}
\begin{align}
\mathcal{D}_{s\bar{x}} \triangleq \{&(d_s,d_x)\in \mathcal{D}_{\mathrm{adm}}: \forall P_{Z^\star Y^\star|X} \in \mathcal{P}^{\star}(d_s,d_x), \nonumber \\ &\mathbb{E}\left[\bar{\textsf{d}}_s(X,Z^\star)\right] = d_s;\ \exists P_{Z^\star Y^\star|X} \in \mathcal{P}^{\star}(d_s,d_x) \nonumber \\ &\textrm{ such that } \mathbb{E}\left[\textsf{d}_x(X,Y^\star)\right] < d_x\},
\end{align}
\begin{align}
\mathcal{D}_{\bar{s}\bar{x}} \triangleq \{&(d_s,d_x)\in \mathcal{D}_{\mathrm{adm}}: \exists P_{Z^\star Y^\star|X} \in \mathcal{P}^{\star}(d_s,d_x) \textrm{ such that} \nonumber \\ &\mathbb{E}\left[\bar{\textsf{d}}_s(X,Z^\star)\right] < d_s;\ \exists P_{Z^\star Y^\star|X} \in \mathcal{P}^{\star}(d_s,d_x) \nonumber \\ &\textrm{ such that } \mathbb{E}\left[\textsf{d}_x(X,Y^\star)\right] < d_x\}.
\end{align}
Clearly, $\{\mathcal{D}_{sx}, \mathcal{D}_{\bar{s}x}, \mathcal{D}_{s\bar{x}}, \mathcal{D}_{\bar{s}\bar{x}}\}$ is a partition of $\mathcal{D}_{\mathrm{adm}}$, resulting from the tightness or looseness of the two constraints. By the methodologies for proving \cite[Theorem 2.5.1]{Berger1971Rate} and \cite[Theorem 2.5.5]{Berger1971Rate}, we conclude that 	$R_{S,X}(d_s,d_x)$ is differentiable on $\mathcal{D}_{\mathrm{in}} \triangleq \textsf{int}(\mathcal{D}_{sx}) \cup \textsf{int}(\mathcal{D}_{\bar{s}x}) \cup \textsf{int}(\mathcal{D}_{s\bar{x}}) \cup \textsf{int}(\mathcal{D}_{\bar{s}\bar{x}})$, where $\textsf{int}(\cdot)$ denotes the interior of the input set.

For $(d_s,d_x) \in \mathcal{D}_{\mathrm{in}}$, the noisy $(\textsf{d}_s,\textsf{d}_x)$-tilted information in $(s,x) \in \mathcal{S} \times \mathcal{X}$ given representations $z \in \widehat{\mathcal{S}}$ and $y \in \widehat{\mathcal{X}}$ is defined as
\begin{equation}
	\label{tilted_information}
	\begin{aligned}
	\tilde{\jmath}_{S,X}(s,x,z,y,d_s,d_x) \triangleq& \imath_{X;Z^\star Y^\star}(x;z,y) + \lambda_s^\star \textsf{d}_s(s,z)\\ 
	&+ \lambda_x^\star \textsf{d}_x(x,y) - \lambda_s^\star d_s - \lambda_x^\star d_x,
	\end{aligned}
\end{equation}
where $P_{Z^\star Y^\star| X} \in \mathcal{P}^{\star}(d_s,d_x)$, $P_{Z^\star Y^\star}=\sum_{x \in \mathcal{X}}P_XP_{Z^\star Y^\star| X}$, and
\begin{equation}
\label{information_density}
\imath_{X;Z^\star Y^\star}(x;z,y) \triangleq \log \frac{\mathrm{d} P_{Z^\star Y^\star| X=x}}{\mathrm{d} P_{Z^\star Y^\star}}(z,y),
\end{equation}
	\begin{equation}
	\label{lambda_s}
	\lambda_s^\star \triangleq -\frac{\partial R_{S,X}(d_s,d_x)}{\partial d_s},
	\end{equation}
	\begin{equation}
	\label{lambda_x}
	\lambda_x^\star \triangleq -\frac{\partial R_{S,X}(d_s,d_x)}{\partial d_x}.
	\end{equation}
We further define the $(\bar{\textsf{d}}_s,\textsf{d}_x)$-tilted information in $x$ for the surrogate noiseless two-constraint source coding problem \cite[Section VI]{Blahut1972Computation}, \cite[Problem 10.19]{Cover2006Elements}, $\jmath_{X}(x,d_s,d_x)$, as
\begin{equation}
	\begin{aligned}
	&\jmath_{X}(x,d_s,d_x)\\ 
	\triangleq &\log\! \frac{1}{\mathbb{E}[\exp\{\lambda_s^\star d_s \! + \! \lambda_x^\star d_x \!-\! \lambda_s^\star \bar{\textsf{d}}_s(x,Z^\star) \!-\! \lambda_x^\star \textsf{d}_x(x,Y^\star)\}]},
	\end{aligned}
\end{equation}
where the expectation is with respect to the unconditional distribution $P_{Z^\star Y^\star}$.
We give some useful properties of $\jmath_{X}(x,d_s,d_x)$ as follows.
\begin{property} 
	\label{tilted_information_property}
	Fix $(d_s,d_x) \in \mathcal{D}_{\mathrm{in}}$. For $P_{Z^\star Y^\star}$-a.e. $(z,y)$, it holds that
	\begin{equation}
	\label{d_tilted_information_of_surrogate}
	\begin{aligned}
	\jmath_{X}(x,d_s,d_x) = &\imath_{X;Z^\star Y^\star}(x;z,y) + \lambda_s^\star \bar{\textsf{d}}_s(x,z)\\ 
	&+ \lambda_x^\star \textsf{d}_x(x,y) - \lambda_s^\star d_s - \lambda_x^\star d_x,
	\end{aligned}
	\end{equation}
	where $P_{XZ^\star Y^\star} = P_X P_{Z^\star Y^\star | X}$. Moreover,
	\begin{align}
	R_{S,X}(d_s,d_x) =& \min_{P_{ZY|X}} \mathbb{E}\big[\imath_{X;ZY}(X;Z,Y) + \lambda_s^\star \bar{\textsf{d}}_s(X,Z) \nonumber \\ 
	&\quad \quad \ \, + \lambda_x^\star \textsf{d}_x(X,Y)\big] - \lambda_s^\star d_s - \lambda_x^\star d_x \label{property1} \\
	=& \mathbb{E}\big[\jmath_{X}(X,d_s,d_x)\big], \label{property3}
	\end{align}
	and for all $z \in \widehat{\mathcal{S}}$ and $y \in \widehat{\mathcal{X}}$,
	\begin{equation}
	\label{property4}
		\begin{aligned}
		&\mathbb{E}\big[\exp\big\{\lambda_s^\star d_s + \lambda_x^\star d_x - \lambda_s^\star \bar{\textsf{d}}_s(X,z) - \lambda_x^\star \textsf{d}_x(X,y)\\ 
		&+ \jmath_{X}(X,d_s,d_x) \big\}\big] \leq 1
		\end{aligned}
	\end{equation}
	with equality for $P_{Z^\star Y^\star}$-a.e. $(z,y)$.
\end{property}
\begin{proof}[Proof]
This property can be obtained by extending the methodology in \cite[Appendix B.1]{Kostina2013Lossy} to accommodate our two-constraint case. We give the detailed proof in Appendix \ref{proof_tilted_information_property}.
\end{proof}
\begin{remark} 
	\label{remark_core_properties}
	The optimal solutions to problem \eqref{rate_distortion_function}, which in turn achieve $R_{S,X}(d_s,d_x)$, consistently adhere to Property \ref{tilted_information_property}. Nevertheless, the conditional distributions $P_{Z^\star Y^\star|X}$ that satisfy Property \ref{tilted_information_property} are not necessarily feasible to problem \eqref{rate_distortion_function}.
\end{remark}

Using Property \ref{tilted_information_property}, for $P_{Z^\star Y^\star}$-a.e. $(z,y)$,
\begin{equation}
\label{relation_of_tilted_information}
\tilde{\jmath}_{S,X}\!(s,x,z,y,d_s,d_x) \!=\! \jmath_{X}\!(x,d_s,d_x) + \lambda_s^\star \textsf{d}_s(s,z) - \lambda_s^\star \bar{\textsf{d}}_s(x,z).
\end{equation}
By \eqref{tilted_information} and \eqref{relation_of_tilted_information}, we have
\begin{align}
R_{S,X}(d_s,d_x) =& \mathbb{E}[\tilde{\jmath}_{S,X}(S,X,Z^\star,Y^\star,d_s,d_x)] \nonumber \\ 
=& \mathbb{E}[\jmath_{X}(X,d_s,d_x)].
\end{align}
Define the noisy rate-dispersion function, $\tilde{\mathcal{V}}(d_s,d_x)$, as
\begin{equation}
\label{noisy_rate_dispersion_func}
\tilde{\mathcal{V}}(d_s,d_x) \triangleq \textrm{Var}\left[\tilde{\jmath}_{S,X}(S,X,Z^\star,Y^\star,d_s,d_x)\right].
\end{equation}
Similarly, define the rate-dispersion function of the surrogate noiseless problem, $\mathcal{V}(d_s,d_x)$, as
\begin{equation}
\mathcal{V}(d_s,d_x) \triangleq \textrm{Var}\left[\jmath_{X}(X,d_s,d_x)\right].
\end{equation}
The following proposition reveals the relationship between $\tilde{\mathcal{V}}(d_s,d_x)$ and $\mathcal{V}(d_s,d_x)$.
\begin{proposition} 
	\label{Prop_relationship_V_tildeV}
	$\tilde{\mathcal{V}}(d_s,d_x)$ can be written as
	\begin{equation}
	\label{relationship_V_tildeV}
	\tilde{\mathcal{V}}(d_s,d_x) = \mathcal{V}(d_s,d_x) + \lambda_s^{\star 2}\textrm{Var}\left[\textsf{d}_s(S,Z^\star)|X,Z^\star\right],
	\end{equation}
	where $\textrm{Var}\left[U|V\right] \triangleq \mathbb{E}\left[(U - \mathbb{E}\left[U|V\right])^2\right]$.
\end{proposition}
\begin{proof}[Proof]
It is a consequence of \eqref{relation_of_tilted_information} and the law of total variance.
\end{proof}

As will be seen, our investigation in the following sections encounters two coupled challenges: 1) how to control the randomness introduced by the hidden source effectively and 2) how to handle the mutual influence between the two distortion constraints. While the first challenge alone has been elegantly solved in \cite{Kostina2016Nonasymptotic}, the coupling of the two challenges arises uniquely in our scenario and will be addressed in the following sections.



\section{Nonasymptotic Achievability Bounds}
\label{Nonasymptotic_Achievability_Bounds}
In this section, we provide two nonasymptotic achievability bounds that are valid for general sources and distortion measures.

\begin{theorem} 
	\label{random_coding_achievability}
	(Achievability): For any $P_{\bar{Z}\bar{Y}}$ defined on $\widehat{\mathcal{S}} \times \widehat{\mathcal{X}}$, there exists a deterministic $(M, d_s,d_x, \epsilon)$ code with 
	\begin{equation}
		\label{achievability_1}
		\epsilon \leq \int_{0}^1 \mathbb{E}\Big[\mathbb{P}^M\left[\pi(X, \bar{Z}, \bar{Y}) > t | X \right]\Big] \mathrm{d} t,
	\end{equation}
	where $P_{X\bar{Z} \bar{Y}} = P_{X} P_{\bar{Z} \bar{Y}}$ and
	\begin{equation}
	\begin{aligned}
	\pi(x,z,y) =& \mathbb{P}[\textsf{d}_s(S, Z) > d_s\\ 
	&\cup \textsf{d}_x(X, Y) > d_x | X=x, Z=z, Y=y].
	\end{aligned}
	\end{equation}
\end{theorem}
\begin{proof}[Proof]
	This proof is established by giving an achievability scheme inspired by random coding. Provided $M$ codewords $\{(c_s^{(i)}, c_x^{(i)})\}_{i=1}^M$, the encoder $\textsf{f}$ and decoder $\textsf{c}$ that minimize the excess distortion probability operate as follows. Having observed $x \in \mathcal{X}$, the optimal encoder chooses
	\begin{equation}
		i^\star \in \arg\min_{i} \pi\left(x,c_s^{(i)}, c_x^{(i)}\right),
	\end{equation}
	that is, $\textsf{f}(x) = i^\star$, and the corresponding decoder outputs $\textsf{c}(\textsf{f}(x)) = (c_s^{(i^\star)}, c_x^{(i^\star)})$.
	
	Denote $\textsf{c}_s(\textsf{f}(x)) = c_s^{(i^\star)}$ and $\textsf{c}_x(\textsf{f}(x)) = c_x^{(i^\star)}$. Following the argument in \cite[Theorem 3]{Kostina2016Nonasymptotic}, the excess distortion probability achieved by the above scheme can be computed by
	\begin{align}
	&\mathbb{P}\left[\textsf{d}_s(S, \textsf{c}_s(\textsf{f}(X))) > d_s \cup \textsf{d}_x(X, \textsf{c}_x(\textsf{f}(X))) > d_x \right] \label{ensemble_error_prob} \\
	=& \mathbb{E}\left[\pi \left(X,  \textsf{c}_s(\textsf{f}(X)), \textsf{c}_x(\textsf{f}(X))\right)\right]\\
	=& \int_{0}^1 \mathbb{P}\left[\pi \left(X,  \textsf{c}_s(\textsf{f}(X)), \textsf{c}_x(\textsf{f}(X))\right) > t\right] \mathrm{d} t\\
	=& \int_{0}^1 \mathbb{E} \left[ \mathbb{P}\left[\pi \left(X,  \textsf{c}_s(\textsf{f}(X)), \textsf{c}_x(\textsf{f}(X))\right) > t | X \right]\right] \mathrm{d} t\\
	=&\int_{0}^1 \mathbb{E} \bigg[ \prod_{i=1}^M \mathbb{P}\left[\pi \left(X,  c_s^{(i)}, c_x^{(i)}\right) > t | X \right]\bigg] \mathrm{d} t, \label{achievability_middle1}
	\end{align}
	where \eqref{achievability_middle1} holds since an error event occurs if and only if all the $M$ codewords fail to meet the distortion requirements.
	
	We then average the right side of \eqref{achievability_middle1} with respect to the codewords $\{(Z_i, Y_i)\}_{i=1}^M$ drawn i.i.d. from $P_{\bar{Z}\bar{Y}}$, independently of any other random variable, so that $P_{X Z_1 Y_1 \dots Z_M Y_M} = P_X \times P_{\bar{Z}\bar{Y}} \times \dots \times P_{\bar{Z}\bar{Y}} $, to obtain
	\begin{align}
		&\mathbb{E} \!\bigg[ \int_{0}^1 \! \mathbb{E} \!\bigg[ \prod_{i=1}^M \mathbb{P}\left[\pi \left(X,  Z_i,  Y_i\right) \!>\! t | X , Z_i, Y_i \right] \bigg| \{(Z_i, Y_i)\}_{i=1}^M\! \bigg] \!\!\mathrm{d} t \!\bigg]\\
		=& \int_{0}^1 \!\mathbb{E} \!\bigg[ \mathbb{E}\! \bigg[   \prod_{i=1}^M \mathbb{P}\left[\pi \left(X,  Z_i,  Y_i\right) > t | X , Z_i, Y_i \right] \bigg| X \bigg]  \bigg] \mathrm{d} t\\
		=& \int_{0}^1 \! \mathbb{E}\! \bigg[    \prod_{i=1}^M \mathbb{E} \!\bigg[ \mathbb{P}\left[\pi \left(X,  Z_i,  Y_i\right) > t | X , Z_i, Y_i \right] \bigg| X \bigg]  \bigg] \mathrm{d} t\\
		=& \int_{0}^1 \mathbb{E} \big[  \mathbb{E}^M \big[ \mathbb{P}\left[\pi \left(X,  \bar{Z},  \bar{Y}\right) > t | X , \bar{Z}, \bar{Y} \right] \big| X \big]  \big] \mathrm{d} t\\
		=& \int_{0}^1 \mathbb{E} \left[ \mathbb{P}^M\left[\pi \left(X,  \bar{Z},  \bar{Y}\right) > t | X\right] \right] \mathrm{d} t.
	\end{align}
	Since there exists a codebook achieving excess distortion probability no greater than the average over codebooks, \eqref{achievability_1} follows.
\end{proof}
	

	The following weakening of Theorem \ref{random_coding_achievability} is appropriate for the achievability part of our second-order analysis.
	
	\begin{theorem} 
		\label{weakening_random_coding_achievability}
		(Achievability): For any $P_{\bar{Z}\bar{Y}}$ defined on $\widehat{\mathcal{S}} \times \widehat{\mathcal{X}}$ and $\gamma > 0$, there exists an $(M, d_s,d_x, \epsilon)$ code with 
		\begin{equation}
		\label{achievability_2}
		\epsilon \leq  \mathbb{P}\left[g_{\bar{Z}\bar{Y}}(X,U) \geq \log\gamma \right] + e^{-\frac{M}{\gamma}},
		\end{equation}
		where $U$ is uniform on $[0, 1]$, independent of $X$, and
		\begin{equation}
			\label{achi_inner_problem}
			g_{\bar{Z}\bar{Y}}(x,t) \triangleq \inf_{P_{ZY}:\ \pi(x, Z, Y) \leq t \ \textrm{a.s.}} D(P_{ZY} \Vert P_{\bar{Z}\bar{Y}}).
		\end{equation}
	\end{theorem}
	\begin{proof}[Proof]
	Following the proof of \cite[Theorem 4]{Kostina2016Nonasymptotic}, by Theorem \ref{random_coding_achievability}, there exists an $(M, d_s,d_x, \epsilon)$ code with 
	\begin{align}
	\epsilon \leq& \int_{0}^1 \mathbb{E}\Big[\mathbb{P}^M\left[\pi(X, \bar{Z}, \bar{Y}) > t | X \right]\Big] \mathrm{d} t\\
	\leq& e^{-\frac{M}{\gamma}}\mathbb{E}\left[\min\left\{1, \gamma\int_{0}^{1} \mathbb{P}\left[\pi(X, \bar{Z}, \bar{Y}) \leq t | X \right] \mathrm{d} t \right\}\right] \nonumber \\
	&+\int_0^1 \mathbb{E}\left[\left|1 - \gamma \mathbb{P}\left[\pi(X, \bar{Z}, \bar{Y}) \leq t | X \right] \right|^+\right] \mathrm{d} t \label{apply_ine} \\
	\leq& e^{-\frac{M}{\gamma}} + \int_0^1 \mathbb{E}\left[\left|1 - \gamma \mathbb{P}\left[\pi(X, \bar{Z}, \bar{Y}) \leq t | X \right] \right|^+\right] \mathrm{d} t,
	\end{align}
	where $|a|^+ \triangleq \max\{0,a\}$, and to obtain \eqref{apply_ine} we applied
	\begin{equation}
	(1 -p)^M \leq e^{-Mp} \leq e^{-\frac{M}{\gamma}} \min\{1, \gamma p\} + |1 - \gamma p|^+.
	\end{equation}
	To bound $\left|1 - \gamma \mathbb{P}\left[\pi(X, \bar{Z}, \bar{Y}) \leq t | X = x\right] \right|^+$, let $P_{ZY}$ be some distribution (chosen individually for each $x$ and $t$) such that $\pi(x, Z, Y) \leq t$ a.s., and write
	\begin{align}
	&\left|1 - \gamma \mathbb{P}\left[\pi(X, \bar{Z}, \bar{Y}) \leq t | X = x\right] \right|^+\\
	\leq&\left|1 \!-\! \gamma \mathbb{E}\left[\exp\left(- \imath_{ZY\Vert \bar{Z}\bar{Y}}(Z,Y) \right) 1\!\left\{ \pi(X, Z, Y) \!\leq\! t\right\} | X \!=\! x\right] \right|^+ \label{achievability_proof_1} \\
	=&\left|1 - \gamma \mathbb{E}\left[\exp\left(- \imath_{ZY\Vert \bar{Z}\bar{Y}}(Z,Y) \right)\right] \right|^+\label{achievability_proof_2}\\
	\leq&\left|1 - \gamma \exp\left(-D\left(P_{ZY} \Vert P_{\bar{Z}\bar{Y}}\right)\right) \right|^+ \label{achievability_proof_3}\\
	\leq& 1\left\{D\left(P_{ZY} \Vert P_{\bar{Z}\bar{Y}}\right) > \log \gamma \right\}, \label{achievability_proof_4}
	\end{align}
	where \eqref{achievability_proof_1} is by the change of measure argument,
	\begin{equation}
	\imath_{ZY\Vert \bar{Z}\bar{Y}}(z,y) = \log \frac{\mathrm{d} P_{ZY}}{\mathrm{d} P_{\bar{Z}\bar{Y}}}(z,y),
	\end{equation}
	\eqref{achievability_proof_2} is by the choice of $P_{ZY}$, \eqref{achievability_proof_3} is by Jensen's inequality, and \eqref{achievability_proof_4} follows from
	\begin{equation}
	\gamma \exp\left(-D\left(P_{ZY} \Vert P_{\bar{Z}\bar{Y}}\right)\right) \geq 
	\left\{
	\begin{aligned}
	1& \ \ \ \textrm{if} \ D\left(P_{ZY} \Vert P_{\bar{Z}\bar{Y}}\right) \leq \log \gamma,\\
	0& \ \ \ \textrm{otherwise}. 
	\end{aligned}
	\right.
	\end{equation}
\end{proof}
	
\section{Nonasymptotic Converse Bounds}
\label{Nonasymptotic_Converse_Bounds}

In this section, we provide the general nonasymptotic converse bounds valid for general sources and distortion measures.
For fixed $P_X$ and auxiliary conditional distribution $P_{\bar{X}|\bar{Z}\bar{Y}}$, denote
\begin{align}
f_{\bar{X}|\bar{Z}\bar{Y}}(s,x,z,y) \triangleq& \imath_{\bar{X}|\bar{Z}\bar{Y}\Vert X}(x;z,y) + \sup_{\lambda_s \geq 0} \lambda_s(\textsf{d}_s(s,z) - d_s) \nonumber \\ 
&+\sup_{\lambda_x \geq 0} \lambda_x(\textsf{d}_x(x,y) - d_x) \!-\! \log M,
\end{align}
where
\begin{equation}
\imath_{\bar{X}|\bar{Z}\bar{Y}\Vert X}(x;z,y) \triangleq \log \frac{\mathrm{d} P_{\bar{X}|\bar{Z}=z,\bar{Y}=y}}{\mathrm{d} P_X}(x).
\end{equation}
The nonasymptotic converse results can now be stated as follows.
\begin{theorem} 
	\label{general_converse_theorem}
	(Converse): If an $(M, d_s, d_x, \epsilon)$ code exists, then the following inequality must be hold:
	\begin{align}
	\label{general_converse}
	\epsilon \geq& \inf_{\substack{P_{ZY|X}:\\ \mathcal{X} \to \widehat{\mathcal{S}}\times \widehat{\mathcal{X}}}} \sup_{\substack{P_{\bar{X}|\bar{Z}\bar{Y}}:\\ \widehat{\mathcal{S}}\times \widehat{\mathcal{X}} \to \mathcal{X}}} \sup_{\gamma \geq 0}\big\{\mathbb{P}\left[f_{\bar{X}|\bar{Z}\bar{Y}}(S,X,Z,Y) \geq \gamma\right]\nonumber\\ 
	&- \exp(-\gamma)\big\},
	\end{align}
	where the middle supremum is over those $P_{\bar{X}|\bar{Z}\bar{Y}}$ such that Radon-Nikodym derivative of $P_{\bar{X}|\bar{Z}=z,\bar{Y}=y}$ with respect to $P_X$ at $x$ exists for $P_{ZY|X}P_X$-a.e. $(z,y,x)$.
\end{theorem}
\begin{proof}[Proof]
	Following the proof of \cite[Theorem 2]{Kostina2016Nonasymptotic}, let the encoder and the decoder be the random mappings $P_{U|X}$ and $P_{ZY|U}$, respectively, where $U$ takes values in $\{1, \dots, M\}$. Then, for any $\gamma \geq 0$,
	\begin{align}
	&\mathbb{P}\left[f_{\bar{X}|\bar{Z}\bar{Y}}(S,X,Z,Y) \geq \gamma\right]\\ 
	=& \mathbb{P}\left[f_{\bar{X}|\bar{Z}\bar{Y}}(S,X,Z,Y) \!\geq\! \gamma, \textsf{d}_s(S,Z)\! > \! d_s  \cup  \textsf{d}_x(X,Y) \!>\! d_x \right] \nonumber \\
	& + \mathbb{P}\left[f_{\bar{X}|\bar{Z}\bar{Y}}(S,X,Z,Y) \!\geq\! \gamma, \textsf{d}_s(S,Z) \!\leq\! d_s, \textsf{d}_x(X,Y) \!\leq\! d_x\right]\\
	\leq& \epsilon + \mathbb{P}\big[\imath_{\bar{X}|\bar{Z}\bar{Y}\Vert X}(X;Z,Y) \geq \gamma +\log M, \nonumber \\
	&\textsf{d}_s(S,Z) \leq d_s, \textsf{d}_x(X,Y) \leq d_x\big] \\ 
	\leq& \epsilon + \mathbb{P}\left[\imath_{\bar{X}|\bar{Z}\bar{Y}\Vert X}(X;Z,Y) \geq \gamma +\log M\right]\\
	\leq& \epsilon + \frac{\exp(-\gamma)}{M}\mathbb{E}\left[\exp(\imath_{\bar{X}|\bar{Z}\bar{Y}\Vert X}(X;Z,Y))\right] \label{Markov_Inequality} \\
	\leq& \epsilon + \frac{\exp(-\gamma)}{M}\sum_{u=1}^M \bigg( \int_{\widehat{\mathcal{S}} \times \widehat{\mathcal{X}}} \mathrm{d} P_{ZY|U}(z,y|u)\nonumber \\ 
	&\cdot\int_{\mathcal{X}} \mathrm{d} P_{\bar{X}|\bar{Z}\bar{Y}}(x|z,y)\bigg) \label{Markov_chain_plus_less_than_1} \\
	=& \epsilon + \exp(-\gamma).
	\end{align}
	where \eqref{Markov_Inequality} is by  Markov's inequality, and \eqref{Markov_chain_plus_less_than_1} holds since $X-U-(Z,Y)$ forms a Markov chain in this order and $P_{U|X}(u|x)\leq 1$ for all $(x,u) \in \mathcal{X}\times \{1,\dots, M\}$.
\end{proof}
The following corollary of Theorem \ref{general_converse_theorem} forms the basis of the converse part of our second-order analysis.
\begin{corollary} 
	\label{converse_bound_corollary}
	(Converse):
	Any $(M, d_s, d_x, \epsilon)$ code must satisfy
	\begin{align}
	\label{converse_bound2}
	\epsilon \geq& \sup_{\substack{P_{\bar{X}|\bar{Z}\bar{Y}}:\\ \widehat{\mathcal{S}}\times \widehat{\mathcal{X}} \to \mathcal{X}}} \sup_{\gamma \geq 0} \Big\{\mathbb{E}\Big[\inf_{z \in \widehat{\mathcal{S}}, y \in \widehat{\mathcal{X}}}\mathbb{P}\left[f_{\bar{X}|\bar{Z}\bar{Y}}(S,X,z,y) \geq \gamma | X\right]\Big] \nonumber \\ 
	&- \exp(-\gamma)\Big\},
	\end{align}
	where the first supremum is over those $P_{\bar{X}|\bar{Z}\bar{Y}}$ such that Radon-Nikodym derivative of $P_{\bar{X}|\bar{Z}=z,\bar{Y}=y}$ with respect to $P_X$ at $x$ exists for every $z \in \widehat{\mathcal{S}}$, $y \in \widehat{\mathcal{X}}$ and $P_X$-a.e. $x$.
\end{corollary}
\begin{proof}[Proof]
	We weaken \eqref{general_converse} using the max–min inequality and
	\begin{align}
	&\inf_{P_{ZY|X}}\mathbb{P}\left[f_{\bar{X}|\bar{Z}\bar{Y}}(S,X,Z,Y) \geq \gamma\right]\\
	=& \inf_{P_{ZY|X}}\mathbb{E}\left[\mathbb{P}\left[f_{\bar{X}|\bar{Z}\bar{Y}}(S,X,Z,Y) \geq \gamma | X\right]\right]\\
	=&\mathbb{E}\Big[\inf_{z \in \widehat{\mathcal{S}}, y \in \widehat{\mathcal{X}}}\mathbb{P}\left[f_{\bar{X}|\bar{Z}\bar{Y}}(S,X,z,y) \geq \gamma | X\right]\Big].
	\end{align}
\end{proof}

\section{Asymptotic Analysis}
\label{Asymptotic_Analysis}

In this section, we consider the block coding setting by letting the alphabets be Cartesian products $\mathcal{S} = \mathcal{M}^k$, $\mathcal{X}=\mathcal{A}^k$, $\widehat{\mathcal{S}} = \hat{\mathcal{M}}^k$, $\widehat{\mathcal{X}}=\hat{\mathcal{A}}^k$, and study the second-order asymptotics in blocklength $k$. We make the following assumptions.
\begin{enumerate}
	\item 
	\label{memoryless_sources}
	(Stationary Memoryless Sources):
	$P_{S^k X^k} = P_S P_{X|S} \times \dots \times P_S P_{X|S}$. 
	
	\item 
	\label{separable_distortion}
	(Separable Distortion Measures):
	\begin{equation}
	\textsf{d}_s(s^k,z^k) = \frac{1}{k}\sum_{i=1}^k \textsf{d}_s(s_i,z_i),
	\end{equation}
	\begin{equation}
	\textsf{d}_x(x^k,y^k) = \frac{1}{k}\sum_{i=1}^k \textsf{d}_x(x_i,y_i).
	\end{equation}

	\item
	\label{Finite_alphabets_for_asymptotic}
	(Finite Alphabets):
	The alphabets $\mathcal{M}$, $\mathcal{A}$, $\hat{\mathcal{M}}$, $\hat{\mathcal{A}}$ are finite sets. 
	
	\item
	\label{differentiability}
	(Differentiability): 
	For all $P_{\bar{X}}$ in some neighborhood of $P_X$, $\textsf{supp}(P_{\bar{Z}^\star \bar{Y}^\star}) = \textsf{supp}(P_{Z^\star Y^\star})$, where $P_{\bar{Z}^\star \bar{Y}^\star}$ achieves $R_{\bar{X}}(d_s,d_x)$; $R_{\bar{X}}(d_s,d_x)$ is twice continuously differentiable with respect to $P_{\bar{X}}$.
	
	\item
	\label{Assumption_positive_definite}
	(Non-Degenerate Single-Letter Distortion Measures):
	There exist possibly identical elements $x_1,x_2 \in \mathcal{A}$ such that functions $\bar{\textsf{d}}_s(x_1,\cdot)$ and $\textsf{d}_x(x_2,\cdot)$ vary non-constantly across the elements in $\hat{\mathcal{M}}$ and $\hat{\mathcal{A}}$, respectively. 
\end{enumerate}

Define the minimum achievable codebook size at blocklength $k$, maximum admissible distortions $(d_s,d_x)$, and joint excess distortion probability $\epsilon$ as
\begin{equation}
M^\star(k, d_s, d_x, \epsilon) \triangleq \min \{M: \exists (k, M, d_s, d_x, \epsilon)\ \textrm{code}\}.
\end{equation}
Then, we have the following dispersion theorem for the problem of indirect lossy source coding with observed source reconstruction.
\begin{theorem} 
	\label{theorem_dispersion}
	(Second-Order Asymptotics): Under assumptions \ref{memoryless_sources}-\ref{Assumption_positive_definite}, fixing $(d_s, d_x) \in \mathcal{D}_{\mathrm{in}}$ and $0 < \epsilon < 1$, the minimum achievable codebook size $M^\star(k, d_s, d_x, \epsilon)$ satisfies
	\begin{align}
	\log M^\star(k, d_s, d_x, \epsilon) =& k R_{S,X}(d_s,d_x)+ \sqrt{k \tilde{\mathcal{V}}(d_s,d_x)} Q^{-1}(\epsilon) \nonumber \\ 
	&+ O(\log k),
	\end{align}
	where $Q^{-1}(\cdot)$ denotes the inverse of the complementary standard Gaussian cumulative distribution function, and $f(k)=O(g(k))$ means $\lim\sup_{k \to \infty}\left|f(k)/g(k)\right| < \infty$.
	
\end{theorem}
\begin{proof}[Proof]
	The achievability and converse parts of the proof can be found in Appendices \ref{proof_achievability_theorem_Gaussian_approximation} and \ref{proof_converse_theorem_Gaussian_approximation}, respectively.
\end{proof}

Theorem \ref{theorem_dispersion} provides a closed-form second-order approximation of the optimum finite blocklength coding rate $\log M^\star(k, d_s, d_x, \epsilon) / k$, i.e., for $(d_s, d_x) \in \textsf{int}(\mathcal{D}_{sx})$,
\begin{align}
\label{second_order_cong}
\frac{\log M^\star(k, d_s, d_x, \epsilon)}{k} \cong R_{S,X}(d_s,d_x)+ \sqrt{\frac{\tilde{\mathcal{V}}(d_s,d_x)}{k}} Q^{-1}(\epsilon),
\end{align}
where the notation $\cong$ denotes that the equality holds up to a term of $O\left(\log k/k\right)$. In the next section, we illustrate the approximation performance of \eqref{second_order_cong} in the case of erased fair coin flips.

\section{Case Study: Erased Fair Coin Flips}
\label{Case_Study}

In the erased fair coin flips (EFCF) case, we examine two correlated sources, $S$ and $X$. Source $S$ is a binary equiprobable source taking values in $\{0,1\}$, and source $X$, taking values in $\{0,1,e\}$, corresponds to the source obtained by observing source $S$ through a binary erasure channel with an erasure rate of $\delta$. We consider Hamming distortion measures, i.e., $\textsf{d}_s(s^k,z^k) = \frac{1}{k}\sum_{i=1}^k \textrm{1}\{s_i \neq z_i\}$ and $\textsf{d}_x(x^k,y^k) = \frac{1}{k}\sum_{i=1}^k \textrm{1}\{x_i \neq y_i\}$.

\subsection{Rate-Distortion Function}

In this section, we abbreviate $R_{S,X}(d_s,d_x)$ as $R(d_s,d_x)$.
For $0 < \delta \leq 1/3$, $d_s \geq \delta/2$ and $d_x \geq 0$, the rate-distortion function $R(d_s,d_x)$ is given by the following theorem.
\begin{theorem} 
	\label{theorem_Rate_Distortion_func}
	Let $0 < \delta < 1/3$, $d_s \geq \delta / 2$, and $d_x \geq 0$. Define sets
	\begin{itemize}
		\item [(i)] $\mathcal{D}_1 \triangleq \{(d_s,d_x): 0 \leq d_x \leq 2 \delta,\ d_s \geq d_x/2 + \delta/2 \}$;
		\item [(ii)] $\mathcal{D}_2 \triangleq \{(d_s,d_x): 2 \delta \leq d_x \leq 1/2 + \delta/2,\ d_s \geq d_x - \delta/2 \}$;
		\item [(iii)] $\mathcal{D}_3 \triangleq \{(d_s,d_x): d_x \geq d_s + \delta/2,\  \delta/2 \leq d_s \leq 1/2 \}$;
		\item [(iv)] $\mathcal{D}_4 \triangleq \{(d_s,d_x): 2d_s - \delta \leq d_x \leq d_s + \delta/2,\  \delta/2 \leq d_s \leq 3\delta/2 \}$;
		\item [(v)] $\mathcal{D}_5 \triangleq \{(d_s,d_x): d_x \geq 1/2 + \delta/2, \ d_s \geq 1 / 2 \}$.
	\end{itemize}
	Then,
	\begin{itemize}
		\item [(i)] if $(d_s, d_x) \in \mathcal{D}_1$, 
		\begin{equation}
		R(d_s,d_x) = h(\delta) + (1 - \delta) \log 2 - h(d_x) - d_x \log 2,
		\end{equation}
		which is achieved by $X-Y^{\star}-Z^{\star}$ with
		\begin{equation}
		P_{Y^\star}(0) \!=\! P_{Y^\star}(1) \!=\! \frac{1 \!-\! \delta \!-\! d_x}{2 \!-\! 3 d_x}, P_{Y^\star}(e) \!=\! \frac{2\delta \!-\! d_x}{2 \!-\! 3 d_x},
		\end{equation}
		\begin{equation}
		P_{X|Y^{\star}}(x|y)=\left\{
		\begin{array}{rcl}
		1 - d_x & & {x=y}\\
		d_x/2 & & {\mbox{otherwise}}
		\end{array} \right.,
		\end{equation}
		and
		\begin{equation}
		P_{Z^{\star}|Y^{\star}}(z|y)=\left\{
		\begin{array}{rcl}
		1 & & {z=y}\\
		\theta & & {z=0,\ y = e}\\
		1 - \theta & & {z=1,\ y = e}\\
		0 & & {\mbox{otherwise}}
		\end{array} \right.,
		\end{equation}
		where $0 \leq \theta \leq 1$. 
		
		\item [(ii)] if $(d_s, d_x) \in \mathcal{D}_2$, 
		\begin{equation}
		R(d_s,d_x) = (1 - \delta) \left[\log 2 - h((d_x - \delta)/(1 - \delta))\right],
		\end{equation}
		which is achieved by $X-Y^{\star}-Z^{\star}$ with
		\begin{equation}
		P_{Y^\star}(0) = P_{Y^\star}(1) = \frac{1}{2},\ P_{Y^\star}(e) = 0,
		\end{equation}
		\begin{equation}
		P_{X|Y^{\star}}(x|y)=\left\{
		\begin{array}{rcl}
		1 - d_x & & {x=y \neq e}\\
		d_x - \delta & & {x\neq y,\ x, y \neq e}\\
		\delta & & {x= e,\ y \neq e}
		\end{array} \right.,
		\end{equation}
		and
		\begin{equation}
		P_{Z^{\star}|Y^{\star}}(z|y)=\left\{
		\begin{array}{rcl}
		1 & & {z=y \neq e}\\
		0 & & {z \neq y,\ y \neq e}
		\end{array} \right.;
		\end{equation}

		\item [(iii)] if $(d_s, d_x) \in \mathcal{D}_3$,
		\begin{equation}
		R(d_s,d_x) \!=\! (1 \!-\! \delta) \left[\log 2 \!-\! h((d_s \!-\! \delta/2)/(1 \!-\! \delta))\right],
		\end{equation}
		which is achieved by $X - Z^{\star} - Y^{\star}$ with
		\begin{equation}
		P_{Z^\star}(0) = P_{Z^\star}(1) = \frac{1}{2},
		\end{equation}
		\begin{equation}
		P_{X|Z^{\star}}(x|z)\!=\!\left\{
		\begin{array}{rcl}\!\!
		1 \!-\! d_s \!-\! \delta/2 & & {x=z}\\
		d_s \!-\! \delta/2 & & {x\neq z,\ x \neq e}\\
		\delta & & {x= e}
		\end{array} \right.,
		\end{equation}
		and
		\begin{equation}
		P_{Y^{\star}|Z^{\star}}(y|z)=\left\{
		\begin{array}{rcl}
		1 & & {y=z}\\
		0 & & {y \neq z}
		\end{array} \right.;
		\end{equation}
		
		\item [(iv)] if $(d_s, d_x) \in \mathcal{D}_4$, 
		\begin{equation}
		\label{R_D_for_D4}
		\begin{aligned}
		R(d_s,d_x) \!= &h(\delta) + (1 - \delta) \log 2\\ 
		&- H\!\left(d_s \!-\! \delta/2,\, d_x \!-\! d_s \!+\! \delta/2,\, 1 \!-\! d_x\right),
		\end{aligned}
		\end{equation}
		which is achieved by
		\begin{equation}
		\label{D4_optimal_edge}
		P_{Z^{\star}Y^{\star}}(z, y) = \left\{
		\begin{array}{rcl}
		\frac{\delta + d_x -1}{\delta + 4 d_x - 2 d_s - 2} & & {z=y}\\
		\frac{d_x - d_s - \delta/2}{\delta + 4d_x - 2 d_s -2} & & {y=e}\\
		0 & & {\mbox{otherwise}}
		\end{array} \right.,
		\end{equation}
		and
			\begin{align}
			\label{optimal_back_condition}
			&P_{X|Z^{\star}Y^{\star}}(x|z, y) \nonumber \\ 
			&= \left\{
			\begin{array}{rcl}
			1-d_x & & {(x|z,y)=\begin{aligned}[t]
				&(0|0,0),(1|1,1),\\
				&(e|0,e),(e|1,e)
				\end{aligned}}\\
			d_s-\frac{\delta}{2} & & {(x|z,y)=\begin{aligned}[t]
				&(1|0,0),(0|1,1),\\
				&(1|0,e),(0|1,e)
				\end{aligned}}\\
			d_x - d_s + \frac{\delta}{2} & & {(x|z,y)=\begin{aligned}[t]
				&(e|0,0),(e|1,1),\\
				&(0|0,e),(1|1,e)
				\end{aligned}}\\
			0 & & {\mbox{otherwise}}
			\end{array} \right..
			\end{align}
		\item [(v)] if $(d_s, d_x) \in \mathcal{D}_5$, $R(d_s,d_x) = 0$,
	\end{itemize}
	where $x,y \in \{0,1,e\}$, $z \in \{0,1\}$, $h(\cdot)$ is the binary entropy function, $H(\cdot,\cdot,\cdot)$ is the discrete entropy function, and $0\log(0)$ is taken to be $0$.
	
\end{theorem}

\begin{proof}[Proof]
	We only give a sketch of the proof. 
	
	The rate-distortion function and the optimal test channel are obtained by solving problem \eqref{rate_distortion_function}. By the chain rule of mutual information, $I(X;Z,Y) \geq I(X;Y)$, with the equality hold if and only if $X-Y-Z$ form a Markov chain in this order. As a consequence, if there exist $P_{Y^{\star}|X}$ and $P_{Z^{\star}|Y^{\star}}$ such that $P_{Y^{\star}|X} = \arg\min_{P_{Y|X}}I(X;Y)\
	\mathrm{s.t.}\ \mathbb{E}\left[\textsf{d}_x(X,Y)\right] \leq d_x$, $X-Y^{\star}-Z^{\star}$ form a Markov chain, and $\mathbb{E}\left[\bar{\textsf{d}}_s(X,Z^{\star})\right] \leq d_s$, we have the optimal test channel $P_{Z^{\star}Y^{\star}|X} =P_{Z^{\star}|Y^{\star}} P_{Y^{\star}|X}$. This corresponds to the case when $(d_s, d_x) \in \mathcal{D}_1 \cup \mathcal{D}_2$. Similarly, $(d_s, d_x) \in \mathcal{D}_3$ implies the case when $X-Z^{\star}-Y^{\star}$ form a Markov chain. Given that $\mathbb{P}[Z^{\star}=0, Y^{\star}=0] = 1$ leads to the rate-distortion function when $(d_s, d_x) \in \mathcal{D}_5$, our focus lies solely on the case when $(d_s, d_x) \in \mathcal{D}_4$.
	
	The optimal test channel for $(d_s, d_x) \in \mathcal{D}_4$ is obtained by solving the Karush-Kuhn-Tucker (KKT) conditions of problem \eqref{rate_distortion_function}. Specifically, the Lagrangian associated with problem \eqref{rate_distortion_function} is written as
	\begin{align}
	&\mathcal{L}\left(\{P_{ZY|X}(z,y|x): (z,y,x) \in \mathcal{Q}\} , s_1, s_2, \lambda(x), \mu(x,z,y)\right) \nonumber \\ 
	=&\sum_{(z,y,x) \in \mathcal{Q}}P_X P_{ZY|X} \log \frac{P_{ZY|X}}{\sum_{x' \in \mathcal{X}}P_X P_{ZY|X}}\nonumber\\ 
	&+ s_1 \bigg(\sum_{(z,y,x) \in \mathcal{Q}}P_X P_{ZY|X}\bar{\textsf{d}}_s(x,z) - d_s \bigg) \nonumber \\
	&+ s_2 \bigg(\sum_{(z,y,x) \in \mathcal{Q}}P_X P_{ZY|X}\textsf{d}_x(x,y) - d_x \bigg)\nonumber\\ 
	&- \sum_{(z,y,x) \in \mathcal{Q}} \mu(x,z,y) P_{ZY|X}\nonumber\\ 
	&+ \sum_{x \in \mathcal{X}} \lambda(x)\bigg(\sum_{z \in \widehat{\mathcal{M}}, y \in \widehat{\mathcal{X}}}P_{ZY|X} - 1\bigg),
	\end{align}
	where $\mathcal{Q} \triangleq \{(z,y,x): P_{ZY|X} > 0 \}$ and the Lagrange multipliers $s_1 \geq 0$, $s_2 \geq 0$, $\mu(x,z,y) \geq 0$. Note that, varying selections of $\mathcal{Q}$ lead to different Lagrangians and, consequently, different KKT conditions. The optimal test channel for $(d_s, d_x) \in \mathcal{D}_4$ is obtained by solving the KKT conditions correspond to $\mathcal{Q} = \{(z,y,x): (z,y) \neq (0, 1),  (z,y) \neq (1, 0)\}$.

\end{proof}

\begin{figure}[htbp]
	\centering
	\includegraphics[width=1\linewidth]{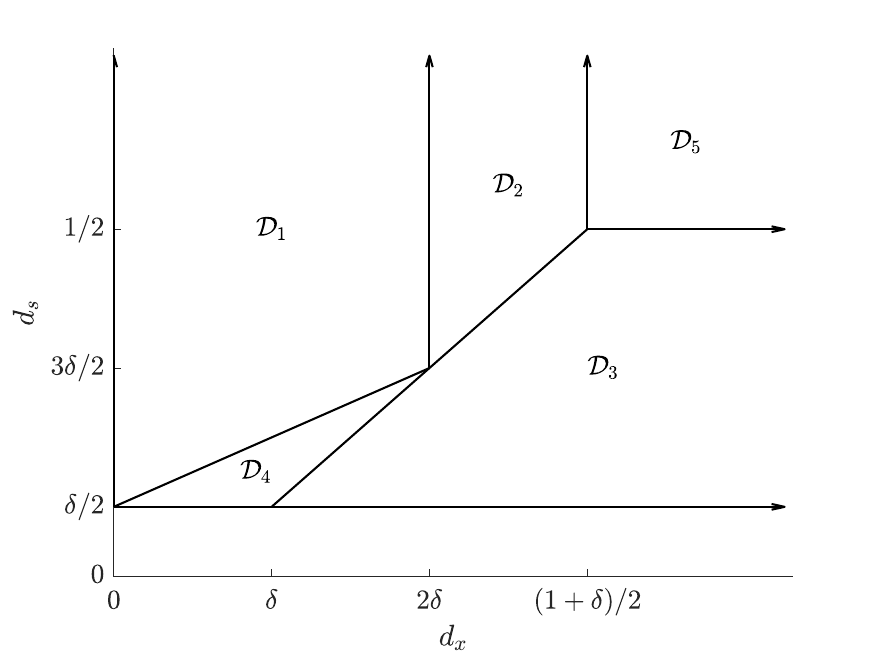}
	\caption{The five regions of the distortion plane with $\delta = 0.2$.}
	\label{five_region}
\end{figure}

\begin{figure}[htbp]
	\centering
	\includegraphics[width=1\linewidth]{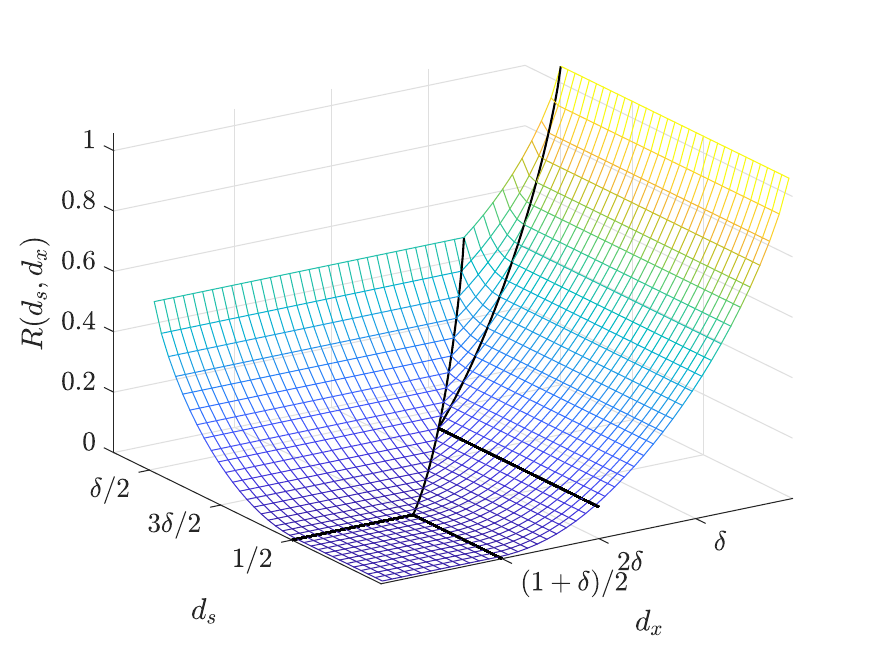}
	\caption{The rate-distortion function of the EFCF case with $\delta = 0.2$.}
	\label{R_D}
\end{figure}

\begin{figure}[htbp]
	\centering
	\includegraphics[width=1\linewidth]{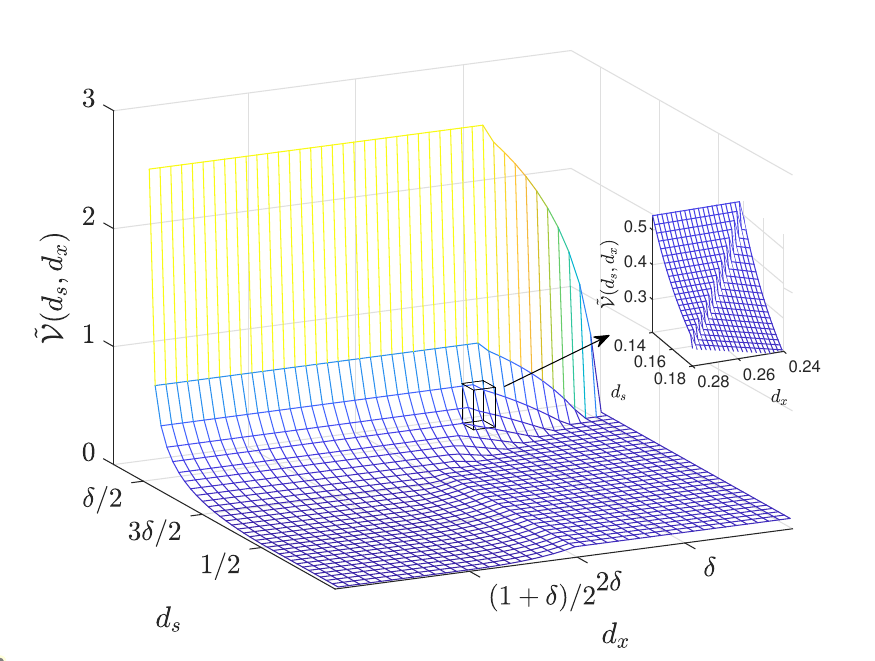}
	\caption{The rate dispersion function under the EFCF case with $\delta = 0.2$.}
	\label{rate_dispersion}
\end{figure}

The five regions and the rate-distortion function in Theorem \ref{theorem_Rate_Distortion_func} with $\delta=0.2$ is plotted in Fig. \ref{five_region} and Fig. \ref{R_D}, respectively. As implied in Theorem \ref{theorem_Rate_Distortion_func}, when $(d_s, d_x) \in \mathcal{D}_1 \cup \mathcal{D}_2$, the distortion constraint for source $X$ is tight and that for source $S$ is loose. Conversely, when $(d_s, d_x) \in \mathcal{D}_3$, the distortion constraint for $S$ is tight and for $X$ is loose. For $(d_s, d_x) \in \mathcal{D}_4$, both distortion constraints are tight, whereas for $(d_s, d_x) \in \mathcal{D}_5$, both constraints are loose. The corresponding noisy rate-dispersion function, computed by \eqref{relationship_V_tildeV}, is plotted in Fig. \ref{rate_dispersion}. As shown in Fig. \ref{rate_dispersion}, the noisy rate-dispersion function is neither convex nor concave and is discontinuous at the boundary between $\mathcal{D}_3$ and $ \mathcal{D}_2 \cup \mathcal{D}_4$ (except for several isolated points).

For $(d_s, d_x) \in \mathcal{D}_4$, some useful quantities are computed as follows:
\begin{equation}
\lambda_s^{\star} = \log\left(\frac{d_x - d_s + \delta / 2}{d_s - \delta/2}\right),
\end{equation}
\begin{equation}
\lambda_x^{\star} = \log\left(\frac{1 - d_x}{d_x - d_s + \delta / 2}\right),
\end{equation}

\begin{align}
\jmath_{X}(0,d_s,d_x) =&\jmath_{X}(1,d_s,d_x) \nonumber \\ 
=& -\lambda_s^{\star}d_s - \lambda_x^{\star} d_x - \log\! \left(\!\frac{1 \!-\! \delta}{2 (1 \!-\! d_x)}\!\right)
\end{align}
\begin{equation}
\jmath_{X}(e,d_s,d_x) =-\lambda_s^{\star}d_s - \lambda_x^{\star} d_x - \log\!\left(\!\sqrt{\frac{d_s \!-\! \delta/2}{d_x \!-\! d_s \!+\! \delta/2}}\!\cdot\! \frac{\delta}{1 \!-\! d_x}\!\right)
\end{equation}
\begin{equation}
\mathcal{V}(d_s,d_x) = \delta  (1 - \delta) \log^2 \left(\sqrt{\frac{d_s \!-\! \delta/2}{d_x \!-\! d_s \!+\! \delta/2}} \cdot \frac{2\delta}{1\!-\!\delta} \right),
\end{equation}
\begin{equation}
\tilde{\mathcal{V}}(d_s,d_x) = \mathcal{V}(d_s,d_x) + \frac{\delta}{4} \left(\lambda_s^{\star}\right)^2.
\end{equation}

\subsection{Nonasymptotic Achievability Bounds}

We first give the nonasymptotic achievability bound for $(d_s, d_x) \in \mathcal{D}_1 \cup \mathcal{D}_4$.
\begin{theorem} 
	\label{case_study_achievability}
	(Achievability, EFCF with $(d_s, d_x) \in \mathcal{D}_1 \cup \mathcal{D}_4$): In erased fair coin flips, for $(d_s, d_x) \in \mathcal{D}_1 \cup \mathcal{D}_4$, there exists an $(k, M, d_s, d_x, \epsilon)$ code such that 
	\begin{align}
	\label{case_study_achievability_bound}
	\epsilon \leq &\sum_{t=0}^k \mathsf{binopmf}(t; k, \delta) \cdot \sum_{i=0}^t \mathsf{binopmf}(i; t, 1/2)\nonumber\\ 
	&\cdot \bigg(1 - \sum_{j=0}^t \mathsf{binopmf}(j; t, \mathbb{P}[Y^{\star} \neq e]) \nonumber \\ 
	&\cdot \sum_{r=0}^{k-t} \mathsf{binopmf}(r; k - t, \mathbb{P}[Y^{\star} = e])\nonumber\\ 
	&\cdot \sum_{v=0}^{\lfloor k d_x \rfloor - j} \mathsf{binopmf}(v - r; k - t - r, 1/2) \nonumber \\ 
	&\cdot \mathsf{binocdf}(\lfloor k d_s \rfloor - i -(v-r); r, 1/2) \bigg)^M,
	\end{align} 
	where $\mathsf{binopmf}(\cdot; n, p)$ and $\mathsf{binocdf}(\cdot; n, p)$ denote the probability mass function (PMF) and the cumulative distribution function (CDF) of the binomial distribution, respectively, with $n$ degrees of freedom and success probability $p$.
	
\end{theorem}

\begin{proof}[Proof]
	
	We adopt the achievability scheme in \cite[Theorem 1]{Yang2024Joint}. Specifically, provided $M$ codewords $\{(\tilde{z}_m^k, \tilde{y}_m^k)\}_{m=1}^M$ and having observed $x^k$, the encoder $\textsf{f}$ chooses 
	\begin{equation}
	\begin{aligned}
	m^\star \in \arg\min_{m} \mathbb{P}[&\textsf{d}_s(S^k, \tilde{z}_m^k) > d_s\\ 
	&\cup \textsf{d}_x(X^k, \tilde{y}_m^k) > d_x | X^k=x^k],
	\end{aligned}
	\end{equation}
	i.e., $\textsf{f}(x^k) = m^\star$. Then, the decoder outputs $\textsf{c}(\textsf{f}(x^k)) = (\tilde{z}_m^k, \tilde{y}_m^k)$. We separate the decoder into $\textsf{c}_s(\textsf{f}(x^k)) = \tilde{z}_m^k$ and $\textsf{c}_x(\textsf{f}(x^k)) = \tilde{y}_m^k$ for convenience.
	
	Consider the ensemble of codes with codewords
	drawn i.i.d. from the optimal edge distribution $P_{Z^\star Y^\star}$ for $(d_s, d_x) \in \mathcal{D}_1 \cup \mathcal{D}_4$. We discuss cases $(d_s, d_x) \in \mathcal{D}_1$ and $(d_s, d_x) \in \mathcal{D}_4$ together because both cases have $P_{Z^\star Y^\star}$ with the same structure. Specifically, given $(d_s, d_x) \in \mathcal{D}_1 \cup \mathcal{D}_4$, let $P_{\tilde{Z}\tilde{Y}}$ be the corresponding optimal edge distribution $P_{Z^\star Y^\star}$ in Theorem \ref{theorem_Rate_Distortion_func}, and define $M$ mutually independent random codewords $\{(\tilde{Z}_m^k, \tilde{Y}_m^k) \sim P_{\tilde{Z}\tilde{Y}} \times \dots \times P_{\tilde{Z}\tilde{Y}}\}_{m=1}^M$. Let $(\tilde{Z}^k, \tilde{Y}^k) \sim P_{\tilde{Z}\tilde{Y}} \times \dots \times P_{\tilde{Z}\tilde{Y}}$. 
	Then, the ensemble error probability can be computed by \eqref{ensemble_error_prob_01}-\eqref{achi_1},
	\begin{figure*}[!t]
		\begin{align}
		&\mathbb{P}\left[\textsf{d}_s(S^k, \textsf{C}_s(\textsf{f}(X^k))) > d_s \cup \textsf{d}_x(X^k,\textsf{C}_x(\textsf{f}(X^k))) > d_x\right] \label{ensemble_error_prob_01} \\
		=&\sum_{t=0}^k \mathbb{P}\left[t\ \text{erasures}\ \text{in}\ X^k \right] \cdot \mathbb{P}\big[ \textsf{d}_s(S^k, \textsf{C}_s(\textsf{f}(X^k))) > d_s\cup \textsf{d}_x(X^k,\textsf{C}_x(\textsf{f}(X^k))) > d_x | t\ \text{erasures}\ \text{in}\ X^k \big]\\
		=&\sum_{t=0}^k \mathbb{P}\left[t\ \text{erasures}\ \text{in}\ X^k \right]\cdot \mathbb{P}\big[ k\textsf{d}_s(S^k, \textsf{C}_s(\textsf{f}(X^k))) \!>\! \lfloor k d_s \rfloor\cup k \textsf{d}_x(X^k,\textsf{C}_x(\textsf{f}(X^k))) \!>\! \lfloor k d_x \rfloor | t\ \text{erasures}\ \text{in}\ X^k \big]\\
		=&\sum_{t=0}^k \mathbb{P}\left[t\ \text{erasures}\ \text{in}\ X^k \right] \cdot \sum_{i=0}^t \bigg \{ \mathbb{P}\left[t \textsf{d}_s(S^t,\textsf{C}_s(\textsf{f}(X^t))) = i | X^t = e,\dots, e \right] \nonumber \\
		&\cdot \mathbb{P}\big[ (k-t)\textsf{d}_s(S^{k-t}, \textsf{C}_s(\textsf{f}(X^{k-t}))) \!>\! \lfloor k d_s \rfloor - i \cup k \textsf{d}_x(X^{k},\textsf{C}_x(\textsf{f}(X^{k}))) \!>\! \lfloor k d_x \rfloor | X^{k-t} \!=\! S^{k-t},t\ \text{erasures}\ \text{in}\ X^k  \big] \bigg \} \label{achi_1}
		\end{align}
		\hrulefill
	\end{figure*}
	where \eqref{achi_1} holds since the distribution of $k\textsf{d}_s(S^t, \textsf{C}_s(\textsf{f}(X^t)))$, given $ X^t = e,\dots, e$, does not dependent on the codebook. 
	
	Next, we compute the probability terms in \eqref{achi_1}. Clearly, since $\mathbb{P}\left[X=e \right] = \delta$,
	\begin{align}
	\mathbb{P}\left[t\ \text{erasures}\ \text{in}\ X^k \right] = C_{k}^t \cdot \delta^t \cdot (1 - \delta)^{k-t};
	\end{align}
	since $P_{S|X}(0|e) = P_{S|X}(1|e) = 1/2$,
	\begin{align}
	\mathbb{P}\left[ t\textsf{d}_s(S^t, \textsf{C}_s(\textsf{f}(X^t))) = i| X^t = e,\dots, e \right] = C_{t}^{i} \cdot \frac{1}{2^t}.
	\end{align}
	For the computation of the last probability term in \eqref{achi_1}, we have \eqref{last_term_01}-\eqref{last_term_04},
	\begin{figure*}[!t]
		\begin{align}
		&\mathbb{P}\left[ (k-t)\textsf{d}_s(S^{k-t}, \textsf{C}_s(\textsf{f}(X^{k-t}))) > \lfloor k d_s \rfloor - i \cup k \textsf{d}_x(X^{k},\textsf{C}_x(\textsf{f}(X^{k}))) > \lfloor k d_x \rfloor | X^{k-t} = S^{k-t},\ t\ \text{erasures}\ \text{in}\ X^k  \right] \label{last_term_01} \\
		=&\prod_{m=1}^M \mathbb{P}\left[ (k-t)\textsf{d}_s(S^{k-t}, \tilde{Z}^{k-t}_m) > \lfloor k d_s \rfloor - i \cup k \textsf{d}_x(X^k, \tilde{Y}^k_m) > \lfloor k d_x \rfloor | X^{k-t} = S^{k-t},\ t\ \text{erasures}\ \text{in}\ X^k  \right]\\
		=&\left( \mathbb{P}\left[ (k-t)\textsf{d}_s(S^{k-t}, \tilde{Z}^{k-t}) > \lfloor k d_s \rfloor - i \cup k \textsf{d}_x(X^{k}, \tilde{Y}^{k}) > \lfloor k d_x \rfloor | X^{k-t} = S^{k-t},\ t\ \text{erasures}\ \text{in}\ X^k  \right] \right)^M \\
		=&\left(1 - \mathbb{P}\left[ (k-t)\textsf{d}_s(S^{k-t}, \tilde{Z}^{k-t}) \leq \lfloor k d_s \rfloor - i \cap k \textsf{d}_x(X^k, \tilde{Y}^k) \leq \lfloor k d_x \rfloor | X^{k-t} = S^{k-t},\ t\ \text{erasures}\ \text{in}\ X^k \right] \right)^M \label{last_term_04}
		\end{align}
		\hrulefill
	\end{figure*}
	where the probability term in \eqref{last_term_04} can be computed by \eqref{prob_term_01}-\eqref{prob_term_04},
	\begin{figure*}[!t]
		\begin{align}
		&\mathbb{P}\left[ (k-t)\textsf{d}_s(S^{k-t}, \tilde{Z}^{k-t}) \leq \lfloor k d_s \rfloor - i \cap k \textsf{d}_x(X^k, \tilde{Y}^k) \leq \lfloor k d_x \rfloor | X^{k-t} = S^{k-t},\ t\ \text{erasures}\ \text{in}\ X^k \right] \label{prob_term_01} \\
		=&\sum_{j=0}^t \mathbb{P}\left[ t \textsf{d}_x(X^t, \tilde{Y}^t) = j | X^t = e,\dots, e \right] \nonumber \\ 
		&\cdot \mathbb{P}\left[ (k-t)\textsf{d}_s(S^{k-t}, \tilde{Z}^{k-t}) \leq \lfloor k d_s \rfloor - i \cap (k - t) \textsf{d}_x(X^{k - t}, \tilde{Y}^{k-t}) \leq \lfloor k d_x \rfloor - j | X^{k-t} = S^{k-t}\right]\\
		=&\sum_{j=0}^t \bigg \{ \mathbb{P}\left[ t \textsf{d}_x(X^t, \tilde{Y}^t) = j | X^t = e,\dots, e \right] \cdot \sum_{r = 0}^{k-t} \mathbb{P}\left[ \text{$r$ erasures in $\tilde{Y}^{k-t}$} \right] \nonumber \\
		&\cdot \mathbb{P}\left[ (k-t)\textsf{d}_s(S^{k-t}, \tilde{Z}^{k-t}) \!\leq\! \lfloor k d_s \rfloor - i \cap (k - t) \textsf{d}_x(X^{k - t}, \tilde{Y}^{k-t}) \!\leq\! \lfloor k d_x \rfloor - j | X^{k-t} \!=\! S^{k-t}, \text{$r$ erasures in $\tilde{Y}^{k-t}$} \right] \bigg \}\\
		=&\sum_{j=0}^t \bigg \{ \mathbb{P}\left[ t \textsf{d}_x(X^t, \tilde{Y}^t) = j | X^t = e,\dots, e \right] \cdot \sum_{r = 0}^{k-t} \mathbb{P}\left[ \text{$r$ erasures in $\tilde{Y}^{k-t}$} \right] \nonumber \\
		&\cdot \sum_{v=0}^{\lfloor k d_x \rfloor - j} \mathbb{P}\left[ (k\!-\!t)\textsf{d}_s(S^{k-t}, \tilde{Z}^{k-t}) \!\leq\! \lfloor k d_s \rfloor - i \cap (k - t) \textsf{d}_x(X^{k - t}, \tilde{Y}^{k-t}) =v | X^{k-t} \!=\! S^{k-t}, \text{$r$ erasures in $\tilde{Y}^{k-t}$} \right] \bigg \} \label{prob_term_04}
		\end{align}
		\hrulefill
	\end{figure*}
	with
	\begin{equation}
	\begin{aligned}
	&\mathbb{P}\left[ t \textsf{d}_x(X^t, \tilde{Y}^t) = j | X^t = e,\dots, e \right]\\ 
	&= C_t^j \cdot (\mathbb{P}[\tilde{Y} \neq e])^j \cdot (\mathbb{P}[\tilde{Y} = e])^{t-j},
	\end{aligned}
	\end{equation}
	\begin{equation}
	\begin{aligned}
	&\mathbb{P}\left[ \text{$r$ erasures in $\tilde{Y}^{k-t}$} \right]\\ 
	&= C_{k-t}^r \cdot (\mathbb{P}[\tilde{Y} = e])^r \cdot (\mathbb{P}[\tilde{Y} \neq e])^{k-t-r},
	\end{aligned}
	\end{equation}
	and the last probability term in \eqref{prob_term_04} being computed by \eqref{Last_Prob_01}-\eqref{Last_Prob_03}.
	\begin{figure*}[!t]
		\begin{align}
		&\mathbb{P}\left[ (k-t)\textsf{d}_s(S^{k-t}, \tilde{Z}^{k-t}) \leq \lfloor k d_s \rfloor - i  \cap (k - t) \textsf{d}_x(X^{k - t}, \tilde{Y}^{k-t}) = v | X^{k-t} = S^{k-t}, \text{$r$ erasures in $\tilde{Y}^{k-t}$} \right] \label{Last_Prob_01} \\
		=& \mathbb{P}\left[ \text{The $(k-t-r)$ non-erasure elements in $\tilde{Y}^{k-t}$ intorduce $(v-r)$ bit errors}  \right] \nonumber \\
		&\cdot \mathbb{P}\Big[ \text{The $r$ elements in $\tilde{Z}^{k-t}$ corresponding to the $r$ erasures in $\tilde{Y}^{k-t}$ intorduce}\nonumber\\ 
		&\text{ less than or equal to $(\lfloor k d_s \rfloor \!-\! i\!-\!(v\!-\!r))$ bit errors}  \Big] \label{language_description} \\
		=&\left\{
		\begin{array}{rcl}
		C_{k-t-r}^{v-r} \cdot \frac{1}{2^{k-t-r}} \cdot \sum_{u=0}^{\lfloor k d_s \rfloor - i -(v-r)} C_{r}^{u} \cdot \frac{1}{2^{r}} & & {
			\begin{array}{l}
			v-r \in \{0, \dots, k-t-r\}\\
			\text{and}\ \lfloor k d_s \rfloor - i - (v-r) \in \{0, \dots, r\}
			\end{array}
		}  \\
		C_{k-t-r}^{v-r} \cdot \frac{1}{2^{k-t-r}} & & {
			\begin{array}{l}
			v-r \in \{0, \dots, k-t-r\}\\
			\text{and}\ \lfloor k d_s \rfloor - i - (v-r) > r
			\end{array}
		}  \\
		0 & & {\text{otherwise}}
		\end{array} \right. \label{Last_Prob_03}
		\end{align}
		\hrulefill
	\end{figure*}
	Note that \eqref{language_description} can be obtained by realizing that $\mathbb{P}[\textsf{d}_x(X, \tilde{Y}) = 1|X=S, \tilde{Y}=e]=1$ and $\mathbb{P}[\textsf{d}_s(S, \tilde{Z}) = \textsf{d}_x(X, \tilde{Y})|X=S, \tilde{Y} \neq e]=1$. 
	
	Since there exists at least one code with an excess probability no greater than the ensemble average, the code satisfying \eqref{case_study_achievability_bound} must exist.
	
\end{proof}

The nonasymptotic achievability bound for $(d_s, d_x) \in \mathcal{D}_2 \cup \mathcal{D}_3$ is given by the following theorem.
\begin{theorem} 
	\label{case_study_achievability2}
	(Achievability, EFCF with $(d_s, d_x) \in \mathcal{D}_2 \cup \mathcal{D}_3$): In erased fair coin flips, for $(d_s, d_x) \in \mathcal{D}_2 \cup \mathcal{D}_3$, there exists an $(k, M, d_s, d_x, \epsilon)$ code such that 
	\begin{align}
	\label{case_study_achievability2_bound}
	\epsilon \! \leq \! &\sum_{t=0}^k \mathsf{binopmf}(t; k, \delta) \cdot \sum_{i=0}^t \mathsf{binopmf}(i; t, 1/2)\nonumber\\ 
	&\cdot \bigg(\!1 - \mathsf{binocdf}(\min\{\lfloor k d_s \rfloor \!-\! i, \lfloor k d_x \rfloor \!-\! t\}; k\!-\!t, 1/2)\!\bigg)^M.
	\end{align} 
	
\end{theorem}
\begin{proof}[Proof]
	This proof closely resembles that of Theorem \ref{case_study_achievability} and is omitted here.
\end{proof}

\subsection{Nonasymptotic Converse Bound}

\begin{theorem} 
	\label{case_study_converse}
	(Converse, EFCF): In erased fair coin flips, any $(k, M, d_s, d_x, \epsilon)$ code must satisfy
	\begin{equation}
	\label{case_study_converse_bound}
	\begin{aligned}
	\epsilon \geq& \sup_{\gamma \geq 0} \bigg\{\mathbb{P}\bigg[\sum_{i=1}^k \jmath_{X}(X_i,d_s,d_x) + \lambda_s^{\star} \big(\textsf{d}_s(S_i, 0)\\ 
	&- \bar{\textsf{d}}_s(X_i, 0)\big) \geq \gamma + \log M\bigg] - \exp(-\gamma)\bigg\},
	\end{aligned}
	\end{equation}
	where 
	\begin{equation}
	\label{EFCF_distortion_measures}
	\textsf{d}_s(s, z)\!=\!\textrm{1}\{s \!\neq\! z\},
	\
	\bar{\textsf{d}}_s(x, z)\!=\!\left\{\!\!
	\begin{array}{rcl}
	\textrm{1}\{x \!\neq\! z\} & & {x \neq e}\\
	1/2 & & {x=e}
	\end{array} \right.\!\!,
	\end{equation}
	and $\{(S_i, X_i)\}_{i=1}^k$ are independently drawn from the EFCF source.
\end{theorem}
\begin{proof}[Proof]
	We weaken the bound in \eqref{converse_bound2} by choosing
	\begin{align}
	P_{\bar{X}^k|\bar{Z}^k=z^k,\bar{Y}^k=y^k}(x^k)=&\prod_{i=1}^k P_{X|Z^{\star}=z_i, Y^{\star}=y_i} (x_i), \\
	\lambda_s=& k \lambda_s^{\star},\\
	\lambda_x=& k \lambda_x^{\star}.
	\end{align}
	By Corollary \ref{converse_bound_corollary}, any $(k, M, d_s, d_x, \epsilon)$ code must satisfy
	\begin{align} 
	\epsilon \geq& \sup_{\gamma \geq 0} \bigg\{\mathbb{E}\bigg[\min_{z^k \in \hat{\mathcal{S}}^k, y^k \in \hat{\mathcal{A}}^k}
	\mathbb{P}\bigg[\sum_{i=1}^k\imath_{X;Z^\star Y^\star}(X_i;z_i,y_i)\nonumber\\
	&+ \lambda_s^{\star}(\textsf{d}_s(S_i,z_i) - d_s) +  \lambda_x^{\star} (\textsf{d}_x(X_i,y_i) - d_x)\nonumber\\
	&\geq \log M + \gamma| X^k \bigg]\bigg] - \exp(-\gamma)\bigg\}\\
	=& \sup_{\gamma \geq 0} \bigg\{\mathbb{E}\bigg[\min_{z^k \in \hat{\mathcal{S}}^k}
	\mathbb{P}\bigg[\sum_{i=1}^k\jmath_{X}(X_i,d_s,d_x) + \lambda_s^\star \big( \textsf{d}_s(S_i,z_i)\nonumber \\ 
	&- \bar{\textsf{d}}_s(X_i,z_i) \big)  
	\geq \log M + \gamma| X^k \bigg]\bigg] - \exp(-\gamma)\bigg\}, \label{EFCF_converse_temp1}
	\end{align}
	where \eqref{EFCF_converse_temp1} is by \eqref{relation_of_tilted_information}. By \eqref{EFCF_distortion_measures}, the probability in \eqref{EFCF_converse_temp1} is the same for all $z^k \in \hat{\mathcal{S}}^k$, leading to \eqref{case_study_converse_bound}.
\end{proof}
\begin{figure}[htbp]
	\centering
	\includegraphics[width=1\linewidth]{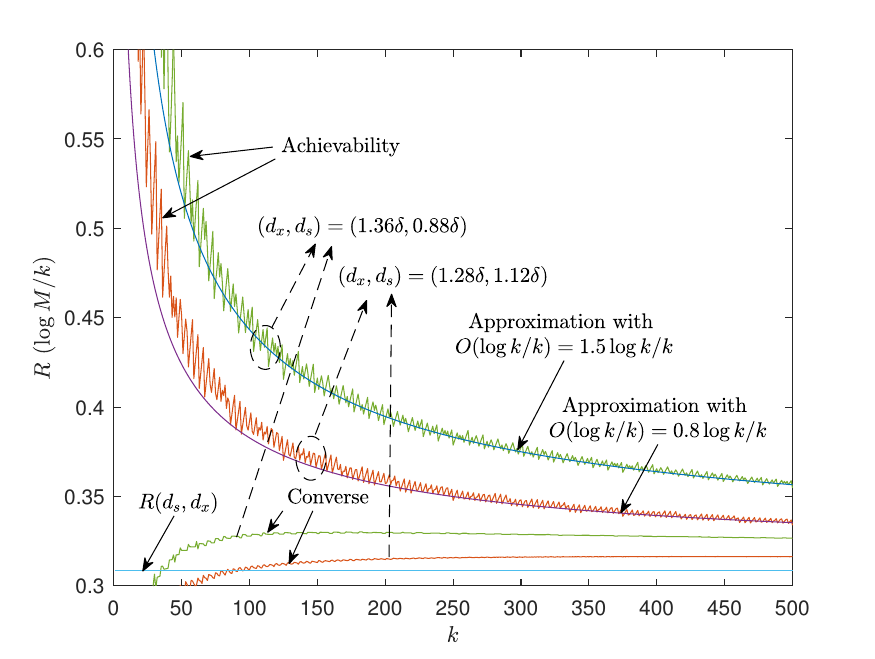}
	\caption{Rate-blocklength trade-off in the erased fair coin flips case with $\delta = 0.2$ and $\epsilon = 0.1$. Note that $(d_x, d_s) = (1.36\delta, 0.88\delta)$ and $(d_x, d_s) = (1.28\delta, 1.12\delta)$ share the same value of $R(d_s,d_x)$.}
	\label{convergence_result}
\end{figure}

\subsection{Numerical Results}

The second-order approximation in Theorem \ref{theorem_dispersion}, the nonasymptotic converse bound with $\gamma=\log k / 2$ in Theorem \ref{case_study_converse}, the nonasymptotic achievability bound in Theorem \ref{case_study_achievability}, and the asymptotically achievable rate \eqref{R_D_for_D4} are plotted in Fig. \ref{convergence_result} for both maximum admissible distortion settings $(d_x,d_s) = (1.36\delta, 0.88\delta)$ and $(d_x,d_s) = (1.28\delta, 1.12\delta)$ with $\delta = 0.2$ and $\epsilon = 0.1$. The nonasymptotic bounds in Fig. \ref{convergence_result} are plotted by using binary search to find the values of $M$ such that the ensemble excess-distortion probabilities computed by the right hand sides of \eqref{case_study_achievability_bound} and \eqref{case_study_converse_bound} equal $\epsilon$ with a given blocklength $k$. Note that the two pairs of maximum admissible distortions fall in $\mathcal{D}_4$ and share the same value of the asymptotically achievable rate $R(d_s,d_x)$. From Fig. \ref{convergence_result}, we observe that the second-order approximations well approximate their corresponding finite-blocklength rates, requiring only slightly adapting the remainder term. 
Furthermore, despite both the $(d_s,d_x)$ settings share the same asymptotically achievable rate $R(d_s,d_x)$, we observe a noticeable gap between their finite-blocklength rates. By Theorem \ref{theorem_dispersion}, we know this gap mainly arises from their distinct dispersions, technically, distinct values of the rate-dispersion function $\tilde{\mathcal{V}}(d_s,d_x)$.

\section{Conclusion}
\label{Conclusion}

In this paper, we derived general nonasymptotic achievability and converse bounds for indirect lossy source coding with observed source reconstruction. The dispersion of this problem was also obtained by second-order asymptotic analysis of the nonasymptotic bounds, which allows us to approximate the excess rate over the rate-distortion function incurred in the finite-blocklength regime. This also implies that the nonasymptotic bounds proposed in Sections \ref{Nonasymptotic_Achievability_Bounds} and \ref{Nonasymptotic_Converse_Bounds} are asymptotically tight up to the second order. Finally, we obtained nonasymptotic achievability and converse bounds for the case of erased fair coin flips. Based on these results, we numerically illustrated the accuracy of the second-order approximation.

An interesting direction for future research is to extend our second-order asymptotic results to the quadratic Gaussian case. Extending the techniques used in this paper to the Gaussian case is not straightforward. Specifically, our achievability proof relies heavily on Theorem \ref{weakening_random_coding_achievability}, which connects the random coding union bound in Theorem \ref{random_coding_achievability} with the KL divergence $D(P_{ZY} \Vert P_{\bar{Z}\bar{Y}})$ through a change of measure argument. This connection allows us to construct the summation $\sum_{i=1}^k \jmath_{X}(X_i,d_s,d_x)$ of independent random variables via \eqref{achi_step_1}-\eqref{achi_step_6}. However, Theorem \ref{weakening_random_coding_achievability} becomes particularly challenging to apply in the Gaussian case, as a proper $P_{ZY}$ that satisfies the constraint of \eqref{achi_inner_problem} is difficult to construct. At the very least, type-based construction methods, such as the one below \eqref{general_rate_distortion_function}, are no longer applicable. Therefore, in the Gaussian case, we may have to forgo Theorem \ref{weakening_random_coding_achievability} and instead directly connect Theorem \ref{random_coding_achievability} with $\sum_{i=1}^k \jmath_{X}(X_i,d_s,d_x)$ by leveraging properties of joint Gaussian distributions.

The converse part of the second-order analysis seems even harder for Gaussian sources. The core challenge lies in handling the outermost optimization in the min-max bound in Theorem \ref{general_converse_theorem}. For finite-alphabet sources considered in this paper, the outermost optimization can be tackled by first shifting it to the innermost layer to relax the min-max bound into the max-min form in Corollary \ref{converse_bound_corollary}, and then tackling the obtained innermost minimization using the approximate optimization results developed in \cite[Appendix A.3]{Kostina2013Lossy}. However, while the penalty introduced by the relaxation from Theorem \ref{general_converse_theorem} to Corollary \ref{converse_bound_corollary} can be properly bounded for finite-alphabet sources as implied by \eqref{low_1}-\eqref{lower_bound_Wi}, it appears that this relaxation may be too loose for Gaussian sources due to their unbounded dense alphabets. For this converse proof, we may need to explore approaches beyond Corollary \ref{converse_bound_corollary} to handle the outermost optimization in Theorem \ref{general_converse_theorem} or try to employ more basic tools of probability theory.

\appendices

\section{Proof of Property \ref{tilted_information_property}} \label{proof_tilted_information_property}

In this proof, we extend the methodology from \cite[Appendix B.1]{Kostina2013Lossy} to accommodate our two-constraint case.

Since Slater's condition is satisfied by convex problem \eqref{rate_distortion_function} when $(d_s,d_x) \in \mathcal{D}_{\mathrm{in}}$, the Karush-Kuhn-Tucker (KKT) coefficients for problem \eqref{rate_distortion_function} exists.  Furthermore, following the methodology in the proof of \cite[Theorem 2.5.1]{Berger1971Rate}, we can compute the KKT coefficients using \eqref{lambda_s} and \eqref{lambda_x}. As a consequence, \eqref{property1} holds. We note that any optimal solution of problem \eqref{rate_distortion_function} is also optimal for the problem in \eqref{property1}, but the converse is not necessarily true.

We now show \eqref{d_tilted_information_of_surrogate}, which automatically show \eqref{property3}. Consider function
\begin{align}
&F(P_{ZY|X}, P_{\bar{Z} \bar{Y}})\\
\triangleq& D(P_{ZY|X} \Vert P_{\bar{Z}\bar{Y}}| P_X) + \lambda_s^\star \mathbb{E}\left[\bar{\textsf{d}}_s(X,Z)\right]\nonumber\\ 
&+ \lambda_x^\star \mathbb{E}\left[\textsf{d}_x(X,Y)\right] - \lambda_s^\star d_s - \lambda_x^\star d_x\\
=& I(X;Z,Y) + D(P_{ZY}\Vert P_{\bar{Z}\bar{Y}})+ \lambda_s^\star \mathbb{E}\left[\bar{\textsf{d}}_s(X,Z)\right]\nonumber\\ 
&+ \lambda_x^\star \mathbb{E}\left[\textsf{d}_x(X,Y)\right] - \lambda_s^\star d_s - \lambda_x^\star d_x\\
\geq& I(X;Z,Y) + \lambda_s^\star \mathbb{E}\left[\bar{\textsf{d}}_s(X,Z)\right] + \lambda_x^\star \mathbb{E}\left[\textsf{d}_x(X,Y)\right]\nonumber\\ 
&- \lambda_s^\star d_s - \lambda_x^\star d_x. \label{B11}
\end{align}
Since the equality in \eqref{B11} holds if and only if $P_{\bar{Z}\bar{Y}} = P_{ZY}$, $R_{S,X}(d_s,d_x)$ can be expressed as
\begin{align}
R_{S,X}(d_s,d_x) =& F(P_{Z^\star Y^\star| X}, P_{Z^\star Y^\star}) \label{opt_R} \\ =&\min_{P_{\bar{Z}\bar{Y}}}\min_{P_{ZY|X}}F(P_{ZY|X}, P_{\bar{Z} \bar{Y}}). \label{double_minimizer}
\end{align}


For an arbitrary $P_{\bar{Z} \bar{Y}}$, define the conditional distribution $P_{\bar{Z}^\star \bar{Y}^\star| X}$ through
\begin{equation}
\label{optimal_PZYX}
\begin{aligned}
&\mathrm{d} P_{\bar{Z}^\star \bar{Y}^\star | X=x}(z,y)\\ \triangleq& \frac{\mathrm{d} P_{\bar{Z} \bar{Y}}(z,y) \exp\left(- \lambda_s^\star \bar{\textsf{d}}_s(x,z) - \lambda_x^\star \textsf{d}_x(x,y)\right)}{\mathbb{E}\left[\exp\left(- \lambda_s^\star \bar{\textsf{d}}_s(x,\bar{Z}) - \lambda_x^\star \textsf{d}_x(x,\bar{Y})\right)\right]},
\end{aligned}
\end{equation}
where the expectation is with respect to unconditional distribution of $(\bar{Z}, \bar{Y})$. Further, define
\begin{equation}
J_{\bar{Z}\bar{Y}} (x, \lambda_s, \lambda_x) \triangleq \log \frac{1}{\mathbb{E}\left[\exp\left(- \lambda_s \bar{\textsf{d}}_s(x,\bar{Z}) - \lambda_x \textsf{d}_x(x,\bar{Y})\right)\right]},
\end{equation}
where the expectation is also with respect to unconditional distribution of $(\bar{Z}, \bar{Y})$. Then, by \cite[Lemma A.2]{Kostina2013Lossy}, only $P_{\bar{Z}^\star \bar{Y}^\star| X}$ that satisfies \eqref{optimal_PZYX} achieves the equality in
\begin{equation}
\label{inner_problem_solve}
\begin{aligned}
&D(P_{ZY|X=x} \Vert P_{\bar{Z} \bar{Y}}) + \lambda_s^\star \mathbb{E}\left[\bar{\textsf{d}}_s(X,Z) | X=x\right]\\ 
&+ \lambda_x^\star \mathbb{E}\left[\textsf{d}_x(X,Y) | X=x\right] \geq J_{\bar{Z}\bar{Y}} (x, \lambda_s^\star, \lambda_x^\star).
\end{aligned}
\end{equation}
Applying \eqref{inner_problem_solve} to solve the inner minimizer in \eqref{double_minimizer}, we have
\begin{align}
R_{S,X}(d_s,d_x) =&
F(P_{Z^\star Y^\star| X}, P_{Z^\star Y^\star})  \label{B18} \\
=&\min_{P_{\bar{Z}\bar{Y}}}F(P_{\bar{Z}^\star \bar{Y}^\star| X}, P_{\bar{Z} \bar{Y}}). \label{B19}
\end{align}
Clearly, $(P_{Z^\star Y^\star| X}, P_{Z^\star Y^\star})$ must satisfy \eqref{optimal_PZYX} since otherwise, we have $R_{S,X}(d_s,d_x) \leq \min_{P_{ZY|X}}F(P_{ZY|X}, P_{Z^\star Y^\star}) < F(P_{Z^\star Y^\star| X}, P_{Z^\star Y^\star})$, where the first inequality follows from \eqref{double_minimizer}, and the last inequality holds since the equality in \eqref{inner_problem_solve} holds only for $(P_{\bar{Z}^\star \bar{Y}^\star| X}, P_{\bar{Z} \bar{Y}})$ that satisfies \eqref{optimal_PZYX}. Thus the pair $(P_{Z^\star Y^\star| X}, P_{Z^\star Y^\star})$ that achieves the rate-distortion function must satisfy \eqref{d_tilted_information_of_surrogate}, which is a particularization of \eqref{optimal_PZYX}.


We next show \eqref{property4}. Although for a fixed $P_{\bar{Z} \bar{Y}}$ we can always define $P_{\bar{Z}^\star \bar{Y}^\star | X}$ via \eqref{optimal_PZYX}, in general we cannot claim that $P_{\bar{Z}^\star \bar{Y}^\star} = P_{\bar{Z} \bar{Y}}$, where $P_X \to P_{\bar{Z}^\star \bar{Y}^\star | X} \to P_{\bar{Z}^\star \bar{Y}^\star}$, unless $P_{\bar{Z} \bar{Y}}$ is such that for $P_{\bar{Z} \bar{Y}}$-a.e. $(z,y)$,
\begin{equation}
\label{B16}
\mathbb{E}\left[\exp\left(J_{\bar{Z}\bar{Y}} (X, \lambda_s, \lambda_x) \!-\! \lambda_s^\star \bar{\textsf{d}}_s(X,z) \!-\! \lambda_x^\star \textsf{d}_x(X,y)\right)\right] \!=\! 1.
\end{equation}
Since $P_X \to P_{Z^\star Y^\star | X} \to P_{Z^\star Y^\star}$, equality in \eqref{B16} particularized to $P_{Z^\star Y^\star}$ holds for $P_{Z^\star Y^\star}$-a.e. $(z,y)$, which corresponds to the equality in \eqref{property4}. To show \eqref{property4} for all $(z,y)$, write, using \eqref{d_tilted_information_of_surrogate} and \eqref{property3},
\begin{align}
R_{S,X}(d_s,d_x) 
=& \mathbb{E}\left[J_{Z^\star Y^\star} (X, \lambda_s^\star, \lambda_x^\star)\right] -\lambda_s^\star d_s - \lambda_x^\star d_x  \\
\leq&\min_{P_{Z Y| X}}  F(P_{Z Y| X}, P_{\bar{Z} \bar{Y}})\\
=&\mathbb{E}\left[J_{\bar{Z} \bar{Y}} (X, \lambda_s^\star, \lambda_x^\star)\right] -\lambda_s^\star d_s - \lambda_x^\star d_x. \label{B22}
\end{align}
For arbitrary $z \in \widehat{\mathcal{M}}$, $y \in \widehat{\mathcal{X}}$ and $0 \leq \alpha \leq 1$, let
\begin{equation}
P_{\bar{Z}\bar{Y}} = (1 - \alpha) P_{Z^\star Y^\star} + \alpha \delta_{\bar{z}\bar{y}},
\end{equation}
for which 
\begin{align}
\label{B24}
J_{\bar{Z} \bar{Y}} (x, \lambda_s^\star, \lambda_x^\star) \!=\! &- \log \big[(1 - \alpha) \exp(- J_{Z^\star Y^\star} (x, \lambda_s^\star, \lambda_x^\star))\nonumber\\ 
&+\! \alpha \exp(-\! \lambda_s^\star \bar{\textsf{d}}_s(x,\bar{z}) \!-\! \lambda_x^\star \textsf{d}_x(x,\bar{y})) \big]\!.
\end{align}
Substituting \eqref{B24} in \eqref{B22}, we obtain
\begin{align}
0 \geq& \mathbb{E}\big[J_{Z^\star Y^\star} (X, \lambda_s^\star, \lambda_x^\star) - J_{\bar{Z} \bar{Y}} (X, \lambda_s^\star, \lambda_x^\star)\big] \label{B25} \\
=& \mathbb{E}\big[\log \big[1 - \alpha + \alpha \exp(J_{Z^\star Y^\star} (X, \lambda_s^\star, \lambda_x^\star)\nonumber\\ 
&- \lambda_s^\star \bar{\textsf{d}}_s(X,\bar{z}) - \lambda_x^\star \textsf{d}_x(X,\bar{y})) \big] \big]. \label{B26}
\end{align}
The derivative of the right side of \eqref{B26} with respect to $\alpha$ evaluated at $\alpha = 0$ is
\begin{equation}
\begin{aligned}
&\mathbb{E}\big[-1 + \exp(J_{Z^\star Y^\star} (X, \lambda_s^\star, \lambda_x^\star)\\ 
&- \lambda_s^\star \bar{\textsf{d}}_s(X,\bar{z}) - \lambda_x^\star \textsf{d}_x(X,\bar{y}))\big] \log e \leq 0,
\end{aligned}
\end{equation}
where the inequality holds since otherwise \eqref{B25} would be violated for sufficiently small $\alpha$. This concludes the proof of \eqref{property4}.

\section{The Berry–Ess\'een Theorem}

The Berry–Ess\'een Theorem introduced in the following lemma plays a central role in our proof.
\begin{lemma} 
	\label{Berry_Esseen}
	(Berry–Ess\'een CLT, e.g., \cite[Theorem 13]{Kostina2012Fixed}, \cite[Chapter XVI.5, Theorem 2]{Feller1971Introduction}):
	Fix an integer $k > 0$. Let random variables $\{W_i \in \mathbb{R}\}_{i=1}^k $ be independent. Denote
	\begin{align}
	\mu_k =& \frac{1}{k}\sum_{i=1}^k \mathbb{E}[W_i],\\
	V_k =& \frac{1}{k}\sum_{i=1}^k \textrm{Var}[W_i],\\
	T_k =&\frac{1}{k}\sum_{i=1}^k \mathbb{E}\left[|W_i - \mathbb{E}[W_i]|^3\right],\\
	B_k =&6 \frac{T_k}{V_k^{3/2}}.
	\end{align}
	Then, for any real $t$,
	\begin{align}
	\label{P_Berry_Esseen}
	\left|\mathbb{P}\left[\sum_{i=1}^k W_i > k \left(\mu_k + t\sqrt{\frac{V_k}{k}}\right)\right] - Q(t)\right| \leq \frac{B_k}{\sqrt{k}},
	\end{align}
	where $Q(t)$ denotes the complementary standard Gaussian cumulative distribution function.
\end{lemma}

\section{Proof of the Achievability Part of Theorem \ref{theorem_dispersion}} \label{proof_achievability_theorem_Gaussian_approximation}

%

We first give a brief overview of the proof. Our proof is an asymptotic analysis of Theorem \ref{weakening_random_coding_achievability}. As implied in Theorem \ref{weakening_random_coding_achievability}, we first tackled the randomness introduced by $S^k$ for each $X^k = x^k$ and then handled the randomness of $X^k$. Tackling the randomness of the hidden source $S^k$ benefits from the following observation: given $X^k = x^k$, the newly introduced error event $\frac{1}{k}\sum_{i=1}^k \textsf{d}_x(x_i, y_i) > d_x$ exhibits no uncertainty. Consequently, given $X^k = x^k$, the newly introduced distortion constraint does not affect tackling the randomness of the remaining error event $\frac{1}{k}\sum_{i=1}^k \textsf{d}_s(S_i, z_i) > d_s$, which allows us to postpone the analysis of the constraint interaction and enables us to extend the approach in [41] to our case through appropriate modifications, as specified in Appendix \ref{proof_achievability_theorem_Gaussian_approximation}.
	
However, when tackling the randomness of $X^k$, the analysis of the interaction between the constraints can no longer be postponed and, to the best of the authors' knowledge, has not been addressed in the literature. In this step, our core target is to bound the probability that a random codeword simultaneously falls into the two distortion balls\footnote{Given $X^k = x^k$, we say a codeword $(z^k, y^k)$ simultaneously falls into the two distortion balls if $\frac{1}{k}\sum_{i=1}^k \bar{\textsf{d}}_s(x_i, z_i) \leq d_s$ and $\frac{1}{k}\sum_{i=1}^k \textsf{d}_x(x_i, y_i) \leq d_x$.} using some summation of independent random variables $\sum_{i=1}^k f_i(X_i)$, so that we can tackle the randomness of $X^k$ by the Berry-Ess\'een theorem. By selecting an appropriate auxiliary distribution $P_{Z^k Y^k}$ and employing the type counting argument, the aforementioned probability can be bounded using the summation $\sum_{i=1}^k \Lambda_{Z^\star Y^\star}(X_i, \boldsymbol{\lambda}(X^k))$, which, however, is a summation of dependent random variables. To break the dependence among them, which is brought about by the shared argument $\boldsymbol{\lambda}(X^k)$, we discuss separately according to the different constraint interaction modes and give Lemmas \ref{Lemma_confrontation} and \ref{Lemma_confrontation_02} in Appendix \ref{GA_Preliminaries} to bound $\boldsymbol{\lambda}(X^k)$. With the aid of Lemmas \ref{Lemma_confrontation} and \ref{Lemma_confrontation_02}, we derive Lemma \ref{Lemma_Lambda_logn} in Appendix \ref{GA_Preliminaries} to break the dependence among $\{ \Lambda_{Z^\star Y^\star}(X_i, \boldsymbol{\lambda}(X^k)) \}_{i=1}^k$ with an acceptable relaxation, as reflected by \eqref{achi_step_5}.




We now give the proof details. This proof consists of an analysis of the random code described in Theorem \ref{random_coding_achievability} with $M$ codewords drawn from the distribution $P_{Z^{k \star} Y^{k \star}} = P_{Z^{\star} Y^{\star}} \times \dots \times P_{Z^{\star} Y^{\star}}$, where $(Z^{\star}, Y^{\star})$ achieves the rate-distortion function $R_{S,X}(d_s,d_x)$. Theorem \ref{weakening_random_coding_achievability} gives a way to analyze the performance of the code ensemble. Specifically, let
\begin{align}
\log M =& \log \gamma + \log \log_e \sqrt{k}, \label{log_M_achi} \\
\log \gamma =& k R(d_s, d_x) + \sqrt{k \tilde{\mathcal{V}}(d_s,d_x)} Q^{-1}\left(\epsilon\right) + O(\log k), \label{log_gamma_achi}
\end{align}
for a properly chosen $O(\log k)$. We note that both $M$ and $\gamma$ are functions of $\epsilon$. By Theorem \ref{weakening_random_coding_achievability}, there exists an $(M, d_s,d_x, \epsilon')$ code with
\begin{align}
\epsilon' \leq& \mathbb{P}\left[g_{Z^{k \star} Y^{k \star}}(X^k,U) \geq \log\gamma \right] + e^{-\frac{M}{\gamma}}\\
\leq& \mathbb{P}\left[g_{Z^{k \star} Y^{k \star}}(X^k,U) \geq \log\gamma, U \in I_k, X^k \in \mathcal{T}_k \right] \nonumber \\ 
&+ \frac{2 + 2 |\mathcal{A}|}{\sqrt{k}} + e^{-\frac{M}{\gamma}}, \label{weakening_Theorem_3}
\end{align}
where $I_k \triangleq [1/\sqrt{k}, 1 - 1/\sqrt{k}]$, and $\mathcal{T}_k$ is the typical set of $X^k$ defined as
\begin{equation}
\label{definition_typical_set}
\mathcal{T}_k \triangleq \left\{x^k \in \mathcal{A}^k: \|\textrm{type}(x^k) - P_X\|^2 \leq \left|\mathcal{A}\right|\frac{\log k}{k}\right\}
\end{equation}
so that
\begin{equation}
\label{Xk_doesnot_typical}
\mathbb{P}\left[X^k \notin \mathcal{T}_k \right] \leq \frac{2 |\mathcal{A}|}{\sqrt{k}}
\end{equation}
holds by \cite[Lemma 1]{Kostina2016Nonasymptotic},\footnote{Denote by $\|\cdot\|$ the Euclidean norm.} 
and \eqref{weakening_Theorem_3} is an obvious weakening. 
The existence of an $(M, d_s,d_x, \epsilon)$ code can thus be justified by showing that $\epsilon'\leq \epsilon$. In the following, we establish $\epsilon'\leq \epsilon$ by showing that the right side of \eqref{weakening_Theorem_3} is upper-bounded by $\epsilon$.

Instead of computing the infimum in \eqref{achi_inner_problem}, we compute an upper bound to $g_{Z^{k \star} Y^{k \star}}(x^k,t)$ by choosing $P_{Z^k Y^k}$ individually for each $x^k \in \mathcal{T}_k$ and $t$. 
Write the following optimization problem:
\begin{subequations}
	\label{general_rate_distortion_function}
	\begin{align}
	R_{X; Z^\star Y^\star}(d_s,d_x) \triangleq \min_{P_{ZY|X}}\ &D(P_{ZY|X} \Vert P_{Z^\star Y^\star}| P_X) \\
	\mathrm{s.t.}\
	&\mathbb{E}\left[\bar{\textsf{d}}_s(X,Z)\right] \leq d_s, \\
	&\mathbb{E}\left[\textsf{d}_x(X,Y)\right] \leq d_x.
	\end{align}
\end{subequations}
Given $t$ and $x^k$, we let $P_{Z^k Y^k}$ be equiprobable on set $\{(z^k, y^k)\ |\ \textrm{type}(z^k,y^k|x^k) = P_{\bar{Z}^\star \bar{Y}^\star| X}\}$, where $P_{\bar{Z}^\star \bar{Y}^\star| X}$ achieves $R_{\bar{X}; Z^\star Y^\star}(d_s-\delta_s, d_x)$ with $\bar{X} \sim P_{\bar{X}} \triangleq \textrm{type}\left(x^k\right)$,
\begin{align}
	\delta_s =& \sqrt{\frac{V_{\bar{X}}}{k}} Q^{-1}\left(t - \frac{B_{\bar{X}}}{\sqrt{k}}\right),
\end{align}
$V_{\bar{X}}$ and $B_{\bar{X}}$ is the normalized variance and the Berry-Ess\'een coefficient of the sum of independent random variables $\sum_{i=1}^k \textsf{d}_s(S_i, z_i)$, and $S_i$ follows the distribution $P_{S|X=x_i}$. We next show that the $P_{Z^k Y^k}$ we choose is feasible to problem \eqref{achi_inner_problem}. By the feasibility of $P_{\bar{Z}^\star \bar{Y}^\star| X}$, we have for $P_{Z^k Y^k}$-a.e. $(z^k, y^k)$, $\sum_{i=1}^k \textsf{d}_x(x_i, y_i) = k\mathbb{E}[\textsf{d}_x(\bar{X}, \bar{Y}^{\star})] \leq k d_x$ 
and $\sum_{i=1}^k \mathbb{E}[\textsf{d}_s(S_i, z_i)| X_i=x_i] = k\mathbb{E}\left[\bar{\textsf{d}}_s(\bar{X},\bar{Z}^{\star})\right] \leq k d_s - k\delta_s$. Therefore, if $V_{\bar{X}} > 0$, by the Berry-Ess\'een theorem,
\begin{align}
&\pi(x^k,z^k,y^k)\\ 
=&\mathbb{P}\Bigg[\sum_{i=1}^k \textsf{d}_s(S_i, z_i) \!>\! k d_s\cup \sum_{i=1}^k \textsf{d}_x(X_i, y_i) \!>\! k d_x \Bigg| X^k\!=\!x^k\Bigg]\\
=&\mathbb{P}\left[\sum_{i=1}^k \textsf{d}_s(S_i, z_i) > k d_s\Bigg| X^k=x^k\right]\\
\leq& \mathbb{P}\left[\sum_{i=1}^k \textsf{d}_s(S_i, z_i) \!>\! k(\mathbb{E}\left[\bar{\textsf{d}}_s(\bar{X},\bar{Z})\right] \!+\! \delta_s)\Bigg| X^k\!=\!x^k\right]\\
\leq& t;\label{S_half}
\end{align}
if $V_{\bar{X}} = 0$, \eqref{S_half} holds trivially since then $\sum_{i=1}^k \textsf{d}_s(S_i, z_i) = kd_s -k\delta_s$, almost surely.

Observe that, for all $x^k \in \mathcal{T}_k$,
\begin{align}
	&D(P_{Z^k Y^k} \Vert P_{Z^{\star} Y^{\star}} \times \dots \times P_{Z^{\star} Y^{\star}}) \label{achi_step_1} \\
	=& k D(P_{\bar{Z}^\star \bar{Y}^\star| X} \Vert P_{Z^{\star} Y^{\star}}| P_{\bar{X}}) + O(\log k) \label{achi_step_2} \\
	=& k R_{\bar{X}; Z^\star Y^\star}(d_s-\delta_s, d_x) + O(\log k) \label{achi_step_3} \\
	=& k R_{\bar{X}; Z^\star Y^\star}(d_s, d_x) + \lambda_{\bar{X},Z^\star Y^\star,s} k \delta_s + O(\log k) \label{achi_step_4} \\
	=& \sum_{i=1}^k \jmath_{X}(x_i,d_s,d_x) + \lambda_{s, X} k \delta_s + O(\log k) \label{achi_step_5} \\
	=& \sum_{i=1}^k \jmath_{X}(x_i,d_s,d_x) \!+\! \lambda_{s, X} \sqrt{k V_X}Q^{-1}\left(t\right) \!+\! O(\log k), \label{achi_step_6}
\end{align}
where
\begin{itemize}
	\item \eqref{achi_step_2} is by type counting;
	\item \eqref{achi_step_4} is by Taylor's theorem, where $\lambda_{\bar{X},Z^\star Y^\star,s} = -\partial R_{\bar{X}; Z^\star Y^\star}(d_s, d_x)/\partial d_s$;
	\item \eqref{achi_step_5} follows from \eqref{generalized_tilted_information_property}, \eqref{J_equals_Lambda}, Lemma \ref{Lemma_Lambda_logn}, and Remark \ref{remark_core_lemma} in Appendix \ref{GA_Preliminaries};
	\item \eqref{achi_step_6} is by applying a Taylor expansion to $V_{\bar{X}}$ in some neighborhood of $P_X$ and to $Q^{-1}(\cdot)$ in some neighborhood of $t$.
\end{itemize}

Finally, following \cite{Kostina2016Nonasymptotic}, note that for scalars $\mu$, $\gamma$, and $v > 0$, we have
\begin{align}
	\int_{0}^1 1\left\{\mu + v Q^{-1}\left(t\right) > \gamma \right\} \mathrm{d}t
	= \mathbb{P} \left[\mu + v G > \gamma \right], \label{math_insight}
\end{align}
where $G \sim \mathcal{N}(0, 1)$. Combining \eqref{achi_inner_problem}, \eqref{achi_step_6} and \eqref{math_insight}, we have
\begin{align}
	&\mathbb{P}\left[g_{Z^{k \star} Y^{k \star}}(X^k,U) \geq \log\gamma, U \in I_k, X^k \in \mathcal{T}_k \right]\\
	\leq& \int_{0}^1 \mathrm{d}t \cdot \mathbb{P}\Bigg[\sum_{i=1}^k \jmath_{X}(X_i,d_s,d_x) + \lambda_{s, X} \sqrt{k V_X}Q^{-1}\left(t\right)\nonumber\\ 
	&+ O(\log k) \geq \log \gamma \Bigg]\\
	=& \mathbb{P}\left[\sum_{i=1}^k \jmath_{X}(X_i,d_s,d_x) + \lambda_{s, X} \sqrt{k V_X}G + O(\log k) \geq \log \gamma \right]. \label{JX_sum}
\end{align}
Utilizing the Berry-Ess\'een theorem and \eqref{log_gamma_achi}, we can always choose the term $O(\log k)$ such that the right side of \eqref{JX_sum} is upper-bounded by $\epsilon - (3 + 2|\mathcal{A}|)/\sqrt{k}$. Juxtaposing this with \eqref{weakening_Theorem_3} and \eqref{JX_sum}, we finally obtain $\epsilon' \leq \epsilon$.

\section{Proof of the Converse Part of Theorem \ref{theorem_dispersion}} \label{proof_converse_theorem_Gaussian_approximation}

If $\lambda_s^\star = 0$, by Proposition \ref{converse_bound_corollary2}, the converse part of Theorem \ref{theorem_dispersion} can be derived directly through the proof methodology employed for the converse part of \cite[Theorem 12]{Kostina2012Fixed}. In the following, we consider the case $\lambda_s^\star > 0$.

Our subsequent proof is adapted from the converse proof of \cite[Theorem 5]{Kostina2016Nonasymptotic}. Let
\begin{equation}
\label{log_M}
\begin{aligned}
\log M = &k R_X(d_s, d_x) + \sqrt{k \tilde{\mathcal{V}}(d_s,d_x)} Q^{-1}(\epsilon_k)\\ 
&- \frac{1}{2} \log k - \log |\mathcal{P}_{[k]}| - c |\mathcal{A}|\log k,
\end{aligned}
\end{equation}
where $\epsilon_k = \epsilon + O\left(\frac{\log k}{\sqrt{k}}\right)$, $c$ is a constant that will be specified in the sequel, and $\mathcal{P}_{[k]}$ denotes the set of all conditional $k$-types $\hat{\mathcal{S}} \times \hat{\mathcal{A}} \to \mathcal{A}$. By type counting, we have $|\mathcal{P}_{[k]}| \leq (k+1)^{|\mathcal{A}| |\hat{\mathcal{S}}|  |\hat{\mathcal{A}}|}$.

We weaken the bound in Corollary \ref{converse_bound_corollary} by choosing
\begin{gather}
P_{\!\bar{X}^k|\bar{Z}^k=z^k,\bar{Y}^k=y^k}\!(x^k)\!=\!\frac{1}{|\mathcal{P}_{[k]}|}\!\!\sum_{P_{X\!|\!ZY} \in \mathcal{P}_{[k]}}\! \prod_{i=1}^k\! P_{X|Z=z_i, Y\!=y_i} \!(x_i), \label{sum_P_X_ZY} \\
\lambda_s= k \lambda_s(x^k) = -k\frac{\partial R_{\textrm{type}(x^k)}(d_s,d_x)}{\partial d_s}\\
\lambda_x= k \lambda_x(x^k) = -k\frac{\partial R_{\textrm{type}(x^k)}(d_s,d_x)}{\partial d_x}\\
\gamma=\frac{1}{2} \log k. \label{gamma_value}
\end{gather}
By Corollary \ref{converse_bound_corollary}, any $(k, M, d_s, d_x, \epsilon')$ code with $M$ given in \eqref{log_M} must satisfy
\begin{align} 
\label{block_converse_bound}
\epsilon' \geq& \mathbb{E}\Big[\min_{z^k \in \hat{\mathcal{S}}^k, y^k \in \hat{\mathcal{A}}^k}
\mathbb{P}\Big[\imath_{\bar{X}^k|\bar{Z}^k \bar{Y}^k\Vert X^k}(X^k;z^k,y^k)\nonumber\\
&+ k \lambda_s(X^k)(\textsf{d}_s(S^k,z^k) - d_s)\nonumber\\
&+ k \lambda_x(X^k) (\textsf{d}_x(X^k,y^k) - d_x) \geq \log M + \gamma \Big| X^k \Big]\Big]\nonumber\\ 
&- \exp(-\gamma).
\end{align}

Below, we prove that the right side of \eqref{block_converse_bound} is lower-bounded by $\epsilon$ for $M$ in \eqref{log_M}, implying that the logarithm of $M$ in any $(k, M, d_s, d_x, \epsilon)$ code must be no less than the right side of \eqref{log_M}. For each triple $(x^k, z^k, y^k)$, abbreviate
\begin{align}
\textrm{type}\left(x^k\right)=&P_{\bar{X}}, \\
\textrm{type}\left(z^k, y^k|x^k\right)=&P_{\bar{Z}\bar{Y}|\bar{X}}, \\
\lambda_s(x^k)=& \lambda_{s,\bar{X}}, \\
\lambda_x(x^k)=& \lambda_{x,\bar{X}},
\end{align}
and define independent random variables
\begin{equation}
\begin{aligned}
W_i \triangleq& I(\bar{X}; \bar{Z}, \bar{Y}) + \lambda_{s, \bar{X}}(\textsf{d}_s(S_i, z_i) - d_s)\\ 
&+ \lambda_{x, \bar{X}}(\textsf{d}_x(x_i, y_i) - d_x),\ i=1,\dots,k,
\end{aligned}
\end{equation}
where $S_i \sim P_{S|X=x_i}$, $i=1,\dots,k$. Since 
\begin{gather}
\mathbb{E}\!\Big[\sum_{i=1}^k \textsf{d}_s(S_i, z_i)\Big]=k \mathbb{E}\!\left[\mathbb{E}\left[\textsf{d}_s\!(S, \bar{Z})|\bar{X}, \bar{Z}\right]\right]\!=\!k \mathbb{E}\left[\textsf{d}_s\!(S, \bar{Z})\right]\!, \\
\sum_{i=1}^k \textsf{d}_x(x_i, y_i) =k\mathbb{E}\left[\textsf{d}_x(\bar{X}, \bar{Y})\right], \\
\begin{aligned}
&\text{Var}\Big[\sum_{i=1}^k \textsf{d}_s(S_i, z_i)\Big]\\
=& k\mathbb{E}\big[\mathbb{E}\big[\big(\textsf{d}_s(S, \bar{Z})- \mathbb{E}\big[\textsf{d}_s(S, \bar{Z})|\bar{X}, \bar{Z}\big]\big)^2| \bar{X}, \bar{Z}\big]\big]\\
=& k\textrm{Var}\left[\textsf{d}_s(S, \bar{Z})|\bar{X}, \bar{Z}\right],
\end{aligned}
\end{gather}
where $P_{S\bar{X}\bar{Z}\bar{Y}} = P_{S|X}P_{\bar{X}\bar{Z}\bar{Y}}$,
in the notation of \cite[Theorem 11]{Kostina2016Nonasymptotic}, we have
\begin{align}
\mu_k(P_{\bar{Z}\bar{Y}|\bar{X}})=&I(\bar{X};\bar{Z}, \bar{Y}) + \lambda_{s, \bar{X}} \left(\mathbb{E}\left[\textsf{d}_s(S,\bar{Z})\right] - d_s\right)\nonumber \\ 
&+ \lambda_{x, \bar{X}} \left(\mathbb{E}\left[\textsf{d}_x(\bar{X},\bar{Y})\right] - d_x\right), \\
V_k(P_{\bar{Z}\bar{Y}|\bar{X}}) =&\lambda_{s,\bar{X}}^2 \textrm{Var}\left[\textsf{d}_s(S, \bar{Z})|\bar{X}, \bar{Z}\right], \\
T_k(P_{\bar{Z}\bar{Y}|\bar{X}})=& \lambda_{s,\bar{X}}^3 \mathbb{E}\!\left[\!\left|\textsf{d}_s(S, \bar{Z}) \!-\! \mathbb{E}\!\left[\textsf{d}_s(S,\bar{Z})|\bar{X}, \bar{Z}\right]\!\right|^3\!\right]\!\!.
\end{align}

Besides, we can write
\begin{align}
&\imath_{\bar{X}^k|\bar{Z}^k \bar{Y}^k\Vert X^k}(x^k;z^k,y^k)
+ k \lambda_s(x^k)(\textsf{d}_s(S^k,z^k) - d_s)\nonumber\\
&+ k \lambda_x(x^k) (\textsf{d}_x(x^k,y^k) - d_x) \label{low_1} \\
\geq&kI(\bar{X};\bar{Z}, \bar{Y}) + kD(\bar{X}\Vert X) + \lambda_{s, \bar{X}} \Big(\sum_{i=1}^k\textsf{d}_s(S_i,z_i) - k d_s\Big)\nonumber\\ 
&+ \lambda_{x, \bar{X}} \Big(\sum_{i=1}^k\textsf{d}_x(x_i,y_i) - k d_x\Big) - \log |\mathcal{P}_{[k]}| \label{lower_bound_i1} \\
=&\sum_{i=1}^k W_i + kD(\bar{X}\Vert X) - \log |\mathcal{P}_{[k]}|\\
\geq&\sum_{i=1}^k W_i - \log |\mathcal{P}_{[k]}|, \label{lower_bound_Wi}
\end{align}
where \eqref{lower_bound_i1} is by lower-bounding the sum in \eqref{sum_P_X_ZY} using the term containing $P_{\bar{X}|\bar{Z}\bar{Y}}$, which is defined through $P_{\bar{X}} P_{\bar{Z}\bar{Y}|\bar{X}}$.
Weakening \eqref{block_converse_bound} using \eqref{lower_bound_Wi}, we have
\begin{align} 
\label{block_converse_bound2}
\epsilon' \!\geq& \mathbb{E}\Big[\!\min_{z^k \in \hat{\mathcal{S}}^k, y^k \in \hat{\mathcal{A}}^k}\!
\mathbb{P}\Big[\sum_{i=1}^k W_i \geq \log M + \gamma + \log |\mathcal{P}_{[k]}| \Big| X^k  \Big]\Big]\nonumber\\ 
&- \exp(-\gamma).
\end{align}

Recall that the typical set of $x^k$, $\mathcal{T}_k$, is defined as
\begin{equation}
\label{definition_typical_set_conv}
\mathcal{T}_k \triangleq \left\{x^k \in \mathcal{A}^k: \|\textrm{type}(x^k) - P_X\|^2 \leq \left|\mathcal{A}\right|\frac{\log k}{k}\right\},
\end{equation}
where $\|\cdot\|$ denotes the Euclidean norm, which satisfies
\begin{equation}
\label{Xk_doesnot_typical_conv}
\mathbb{P}\left[X^k \notin \mathcal{T}_k \right] \leq \frac{2 |\mathcal{A}|}{\sqrt{k}}.
\end{equation}
Next, we evaluate the minimum in \eqref{block_converse_bound2} for $x^k \in \mathcal{T}_k$.

If $V_k(P_{Z^\star Y^\star|X}) = 0$, we have $\textsf{d}_s(S,z) - \bar{\textsf{d}}_s(x,z) \overset{a.s.}{=} 0$ for all $z \in \textrm{supp}(P_{Z^\star})$. Therefore, the considered problem simplifies to a noiseless two-constraint source coding problem \cite[Section VI]{Blahut1972Computation}, \cite[Problem 10.19]{Cover2006Elements}, allowing us to establish our second-order result in a manner analogous to the proof of \cite[Theorem 12]{Kostina2012Fixed}. In the following, we assume that $V_k(P_{Z^\star Y^\star|X}) > 0$. Similar to the argument in the converse proof of \cite[Theorem 5]{Kostina2016Nonasymptotic},
the conditions of \cite[Theorem 11]{Kostina2016Nonasymptotic} are satisfied by $\{W_i\}_{i=1}^k$, with $\mu_k^\star = \mu_k(P_{\bar{Z}^\star\bar{Y}^\star| \bar{X}})$ and $V_k^\star = V_k(P_{\bar{Z}^\star\bar{Y}^\star| \bar{X}})$. Therefore, denoting
\begin{equation}
\Delta_k(P_{\bar{X}}) = \log M + \gamma + \log |\mathcal{P}_{[k]}| - k R_{\bar{X}}(d_s,d_x),
\end{equation}
where $M$ and $\gamma$ are chosen in \eqref{log_M} and \eqref{gamma_value}, by \cite[Theorem 11]{Kostina2016Nonasymptotic}, we have
\begin{align}
\label{Q_lower_bound_1}
&\min_{P_{\bar{Z}\bar{Y}|\bar{X}}} \mathbb{P}\Big[\sum_{i=1}^k W_i \geq \log M + \gamma + \log |\mathcal{P}_{[k]}| \Big| \textrm{type}(X^k) = P_{\bar{X}} \Big]\nonumber\\ 
&\geq Q\Bigg(\frac{\Delta_k(P_{\bar{X}})}{ \lambda_{s,\bar{X}}\sqrt{k \textrm{Var}\left[\textsf{d}_s(S, \bar{Z}^{\star})|\bar{X}, \bar{Z}^{\star}\right]}}\Bigg) - \frac{K}{\sqrt{k}},
\end{align}
where $K > 0$ is that in \cite[Theorem 11]{Kostina2016Nonasymptotic}.

By assumption \ref{differentiability} and \eqref{Prop2_c2}, we can apply a Taylor series expansion in a neighborhood of $P_X$ to $\frac{1}{ \lambda_{s,\bar{X}}\sqrt{\textrm{Var}\left[\textsf{d}_s(S, \bar{Z}^{\star})|\bar{X}, \bar{Z}^{\star}\right]}}$. Consequently, for some scalars $a$ and $K_1 \geq 0$, we have
\begin{align}
&Q\left(\frac{\Delta_k(P_{\bar{X}})}{ \lambda_{s,\bar{X}}\sqrt{k \textrm{Var}\left[\textsf{d}_s(S, \bar{Z}^{\star})|\bar{X}, \bar{Z}^{\star}\right]}}\right)\\
\geq& Q\left(\frac{\Delta_k(P_{\bar{X}})}{ \lambda_{s,X}\sqrt{k \textrm{Var}\left[\textsf{d}_s(S, Z^\star)|X, Z^\star \right]}}\left(1 + a \sqrt{\frac{\log k}{k}}\right)\right)\\
\geq& Q\left(\frac{\Delta_k(P_{\bar{X}})}{ \lambda_{s,X}\sqrt{k \textrm{Var}\left[\textsf{d}_s(S, Z^\star)|X, Z^\star \right]}}\right) - K_1 \frac{\log k}{\sqrt{k}}, \label{Q_lower_bound1}
\end{align}
where \eqref{Q_lower_bound1} holds since
$Q(x + \xi) \geq Q(x) - \frac{|\xi|^+}{\sqrt{2 \pi}}$
and $\Delta_k(P_{\bar{X}}) = O(\sqrt{k \log k})$ for $x^k \in \mathcal{T}_k$.

Besides, for $x^k \in \mathcal{T}_k$, there exists $c > 0$ such that
\begin{align}
&R_{\bar{X}}(d_s,d_x) \nonumber \\ 
\geq& R_{X}(d_s,d_x) + \sum_{a \in \mathcal{A}}(P_{\bar{X}}(a) - P_{X}(a)) \dot{R}_X(a, d_s,d_x)\nonumber\\ 
&- c \|P_{\bar{X}} - P_X\|^2 \label{Taylor_apply} \\
=& R_{X}(d_s,d_x) + \frac{1}{k} \sum_{i=1}^k \dot{R}_X(x_i, d_s, d_x)\nonumber\\ 
&- \mathbb{E}\left[\dot{R}_X(X, d_s, d_x)\right] - c \|P_{\bar{X}} - P_X\|^2\\
=& \mathbb{E}\left[\jmath_{X}(\bar{X},d_s,d_x)\right] - c \|P_{\bar{X}} - P_X\|^2 \label{app_Lemma1} \\
\geq& \mathbb{E}\left[\jmath_{X}(\bar{X},d_s,d_x)\right] - c |\mathcal{A}| \frac{\log k}{k}, \label{R_X_bar_lower_bound}
\end{align}
where \eqref{Taylor_apply} is by Taylor's theorem, applicable because assumption \ref{differentiability} holds, \eqref{app_Lemma1} is by Lemma \ref{finite_alphabet_derivative}, and \eqref{R_X_bar_lower_bound} is by the definition \eqref{definition_typical_set_conv}. Then, by introducing independent random variable $G \sim \mathcal{N}(0, 1)$, we have
\begin{align}
&Q\left(\frac{\Delta_k(P_{\bar{X}})}{ \lambda_{s,X}\sqrt{k \textrm{Var}\left[\textsf{d}_s(S, Z^\star)|X, Z^\star \right]}}\right)\\
=& \mathbb{P}\Big[k R_{\bar{X}}(d_s,d_x) + \lambda_{s,X}\sqrt{k \textrm{Var}\left[\textsf{d}_s(S, Z^\star)|X, Z^\star \right]} G\nonumber\\ 
&\geq \log M + \gamma + \log |\mathcal{P}_{[k]}| \Big]\\
\geq& \mathbb{P}\Big[k\mathbb{E}\left[\jmath_{X}(\bar{X},d_s,d_x)\right] + \lambda_{s,X}\sqrt{k \textrm{Var}\left[\textsf{d}_s(S, Z^\star)|X, Z^\star \right]} G\nonumber\\ 
&\geq \log M + a_k\Big], \label{lower_bound_Q_2}
\end{align}
where $a_k \triangleq \gamma + \log |\mathcal{P}_{[k]}| + c |\mathcal{A}| \log k$.

Finally, we have \eqref{step1}-\eqref{step7},
\begin{figure*}[!t]
	\begin{align}
	\epsilon'
	\geq& \mathbb{E}\left[\min_{z^k \in \hat{\mathcal{S}}^k, y^k \in \hat{\mathcal{A}}^k}
	\mathbb{P}\left[\sum_{i=1}^k W_i \geq \log M + \gamma + \log |\mathcal{P}_{[k]}| \bigg| X^k  \right]\right] - \exp(-\gamma) \label{step1}\\
	=&\mathbb{E}\left[\min_{P_{\bar{Z}\bar{Y}|\bar{X}}} \mathbb{P}\left[\sum_{i=1}^k W_i \geq \log M + \gamma + \log |\mathcal{P}_{[k]}| \bigg| \textrm{type}(X^k) = P_{\bar{X}} \right]\right] - \exp(-\gamma) \label{step2}\\
	\geq&\mathbb{E}\left[\min_{P_{\bar{Z}\bar{Y}|\bar{X}}} \mathbb{P}\left[\sum_{i=1}^k W_i \geq \log M + \gamma + \log |\mathcal{P}_{[k]}| \bigg| \textrm{type}(X^k) = P_{\bar{X}} \right] \textbf{1}\left\{X^k \in \mathcal{T}_k\right\}\right] - \exp(-\gamma) \label{step3}\\
	\geq&\mathbb{E}\!\left[\mathbb{P}\!\left[k\mathbb{E}\left[\jmath_{X}\!(\bar{X},d_s,d_x)\right] \!+\! \lambda_{s,X}\!\sqrt{k \textrm{Var}\left[\textsf{d}_s(S, Z^\star)|X, Z^\star \right]} G \!\geq\! \log M \!+\! a_k\right] \!\textbf{1}\!\left\{X^k \!\in\! \mathcal{T}_k\right\}\!\right] \!-\! \frac{K}{\sqrt{k}} \!-\! K_1 \frac{\log k}{\sqrt{k}} \!-\! \exp(-\gamma) \label{step4}\\
	\geq&\mathbb{P}\!\left[\sum_{i=1}^k\jmath_{X}(X_i,d_s,d_x) \!+\! \lambda_{s,X}\sqrt{k \textrm{Var}\left[\textsf{d}_s(S, Z^\star)|X, Z^\star \right]} G \!\geq\! \log M \!+\! a_k\right] \!-\! \mathbb{P}\left[X^k \notin \mathcal{T}_k\right]\!-\! \frac{K}{\sqrt{k}} \!-\! K_1 \frac{\log k}{\sqrt{k}} \!-\! \exp(-\gamma) \label{step5}\\
	\geq&\epsilon_k - \frac{B}{\sqrt{k+1}} - \mathbb{P}\left[X^k \notin \mathcal{T}_k\right] 
	- \frac{K}{\sqrt{k}} - K_1 \frac{\log k}{\sqrt{k}} - \exp(-\gamma) \label{step6}\\
	\geq& \epsilon_k - \frac{B + 2|\mathcal{A}| + K + K_1 \log k + 1}{\sqrt{k}} \label{step7}
	\end{align}
	\hrulefill
\end{figure*}
where 
\begin{itemize}
	\item \eqref{step1} is by \eqref{block_converse_bound2};
	\item \eqref{step2} is by the observation that $\sum_{i=1}^k W_i$ remains unchanged for fixed $P_{\bar{X}\bar{Z}\bar{Y}}$;
	\item \eqref{step4} is by \eqref{Q_lower_bound_1}, \eqref{Q_lower_bound1} and \eqref{lower_bound_Q_2};
	\item \eqref{step5} is by the union bound and the definition of $\mathbb{E}\left[\jmath_{X}(\bar{X},d_s,d_x)\right]$;
	\item \eqref{step6} is by the Berry-Ess\'een theorem (Lemma \ref{Berry_Esseen}), the choice of $M$ in \eqref{log_M}, and Proposition \ref{Prop_relationship_V_tildeV}, where $B$ is that in Lemma \ref{Berry_Esseen};
	\item \eqref{step7} is by the choice of $\gamma$ in \eqref{gamma_value} and property \eqref{Xk_doesnot_typical_conv}.
\end{itemize}
Our proof completes by letting
\begin{equation}
\epsilon_k = \epsilon + \frac{B + 2|\mathcal{A}| + K + K_1 \log k + 1}{\sqrt{k}}.
\end{equation}

\section{Auxiliary Results for the Proof of the Achievability Part of Theorem \ref{theorem_dispersion}}
\label{GA_Preliminaries}

In this appendix, we derive Lemma \ref{Lemma_Lambda_logn} (and thereby Remark \ref{remark_core_lemma}) to assist the proof of the achievability part of Theorem \ref{theorem_dispersion}. Prior to proving Lemma \ref{Lemma_Lambda_logn}, we present some preliminary results that will be employed.

\begin{lemma} 
	\label{lemma_concave}
	If a concave function $f: \mathbb{R}^r \to \mathbb{R}$ and a convex subset $\mathcal{E}$ of the domain of $f$ satisfy
	\begin{itemize}
		\item [(a)] $f$ is differentiable at the boundary points of $\mathcal{E}$;
		\item [(b)] for each boundary point $\boldsymbol{x}$ of $\mathcal{E}$, there exists $t > 0$ such that $\boldsymbol{x} + t \nabla f (\boldsymbol{x}) \in \textsf{int}(\mathcal{E})$, where $\textsf{int}(\mathcal{E})$ denotes the interior of set $\mathcal{E}$,
	\end{itemize}
	we have $\arg\sup_{\boldsymbol{x}}f(\boldsymbol{x}) \subseteq \textsf{int}(\mathcal{E})$.
\end{lemma}
\begin{proof}[Proof]
	Suppose by contradiction that there exists $\boldsymbol{x}^{\star} \in \arg\sup_{\boldsymbol{x}}f(\boldsymbol{x})$ such that $\boldsymbol{x}^{\star} \notin \textsf{int}(\mathcal{E})$. Clearly, $\boldsymbol{x}^{\star}$ is not a boundary point of $\mathcal{E}$ since otherwise we have $\nabla f (\boldsymbol{x}^{\star}) = \boldsymbol{0}$, leading to $\boldsymbol{x}^{\star} + t \nabla f (\boldsymbol{x}^{\star}) \notin \textsf{int}(\mathcal{E})$, $\forall t > 0$. Therefore, $\boldsymbol{x}^{\star} \notin \textsf{cl}(\mathcal{E})$, where $\textsf{cl}(\mathcal{E})$ denotes the closure set $\mathcal{E}$. Let $\boldsymbol{x}_0 \in \arg\min_{\boldsymbol{x}\in \textsf{cl}(\mathcal{E})}\| \boldsymbol{x} - \boldsymbol{x}^{\star} \|$, which must be a boundary point of $\mathcal{E}$. Next, we prove that $\langle \nabla f(\boldsymbol{x}_0), \boldsymbol{x}^{\star} - \boldsymbol{x}_0\rangle > 0$.
	
	Since $\boldsymbol{x}_0$ is a boundary point, $\nabla f(\boldsymbol{x}_0) \neq \boldsymbol{0}$. Thus $\boldsymbol{x}_0 \notin \arg\sup_{\boldsymbol{x}}f(\boldsymbol{x})$, leading to $f(\boldsymbol{x}_0) < f(\boldsymbol{x}^{\star})$. Juxtaposing this with the concavity of $f$, defining $g(\theta) = f(\boldsymbol{x}_0 + \theta (\boldsymbol{x}^{\star} - \boldsymbol{x}_0))$, we have $g_{+}^{\prime}(0) > 0$. $\langle \nabla f(\boldsymbol{x}_0), \boldsymbol{x}^{\star} - \boldsymbol{x}_0\rangle > 0$ is obtained by noticing that $\|\boldsymbol{x}^{\star} - \boldsymbol{x}_0\|>0$ and $\langle \nabla f(\boldsymbol{x}_0), (\boldsymbol{x}^{\star} - \boldsymbol{x}_0)/\|\boldsymbol{x}^{\star} - \boldsymbol{x}_0\|\rangle = g_{+}^{\prime}(0)$.
	
	Observe that, for all $\boldsymbol{x} \in \textsf{cl}(\mathcal{E})$, $\langle \boldsymbol{x} - \boldsymbol{x}_0, \boldsymbol{x}^{\star} - \boldsymbol{x}_0 \rangle \leq 0$, since otherwise we can always find a point on the line segment from $\boldsymbol{x}_0$ to $\boldsymbol{x}$ that belongs to $\textsf{cl}(\mathcal{E})$ (since $\textsf{cl}(\mathcal{E})$ is convex) and is closer to $\boldsymbol{x}^{\star}$ than $\boldsymbol{x}_0$, which contradicts $\boldsymbol{x}_0 \in \arg\min_{\boldsymbol{x}\in \textsf{cl}(\mathcal{E})}\| \boldsymbol{x}^{\star} - \boldsymbol{x} \|$. Juxtaposing this observation with the result $\langle \nabla f(\boldsymbol{x}_0), \boldsymbol{x}^{\star} - \boldsymbol{x}_0\rangle > 0$ obtained earlier, we have, for any $\boldsymbol{x} \in \textsf{cl}(\mathcal{E})$ and $t > 0$, $t \nabla  f(\boldsymbol{x}_0) \neq \boldsymbol{x} - \boldsymbol{x}_0$, which contradicts condition (b). Thus $\boldsymbol{x}^{\star} \in \textsf{int}(\mathcal{E})$ if $\boldsymbol{x}^{\star} \in \arg\sup_{\boldsymbol{x}}f(\boldsymbol{x})$, which completes the proof.
\end{proof}

To establish Lemmas \ref{Lemma_confrontation}, \ref{Lemma_confrontation_02}, and \ref{Lemma_Lambda_logn}, we begin by providing the definition of generalized $(\bar{\textsf{d}}_s, \textsf{d}_x)$-tilted information, along with some of its properties.

For $\lambda_s, \lambda_x \geq 0$, the generalized $(\bar{\textsf{d}}_s,\textsf{d}_x)$-tilted information is defined as 
\begin{equation}
\label{Lambda_ZY}
	\begin{aligned}
	&\Lambda_{ZY}(x, \lambda_s, \lambda_x)\\ 
	\triangleq &\log\! \frac{1}{\mathbb{E}[\exp(\lambda_s d_s \! + \! \lambda_x d_x \!-\! \lambda_s \bar{\textsf{d}}_s(x,Z) \!-\! \lambda_x \textsf{d}_x(x,Y))]}.
	\end{aligned}
\end{equation}
Similar to \cite[Section II-C]{Kostina2012Fixed}, for $(d_s,d_x) \in \textsf{int}(\mathcal{D}_{\mathrm{adm}})$, the minimum in \eqref{general_rate_distortion_function} is always achieved by a $P_{\bar{Z}^\star \bar{Y}^\star|X}$ satisfying
\begin{align}
\label{generalized_tilted_information_property}
	&\Lambda_{Z^\star Y^\star}(x, \lambda^\star_{X,Z^\star Y^\star,s}, \lambda^\star_{X,Z^\star Y^\star,x}) \nonumber \\
	=&\log \frac{\mathrm{d} P_{\bar{Z}^\star \bar{Y}^\star| X=x}}{\mathrm{d} P_{Z^\star Y^\star}}(z,y) + \lambda^\star_{X,Z^\star Y^\star,s} \bar{\textsf{d}}_s(x,z)\nonumber \\ 
	&+ \lambda^\star_{X,Z^\star Y^\star,x} \textsf{d}_x(x,y) - \lambda^\star_{X,Z^\star Y^\star,s} d_s - \lambda^\star_{X,Z^\star Y^\star,x} d_x,
\end{align}
where $\lambda^\star_{X,Z^\star Y^\star,s} = - \partial R_{X; Z^\star Y^\star}(d_s,d_x)/\partial d_s$, $\lambda^\star_{X,Z^\star Y^\star,x} = - \partial R_{X; Z^\star Y^\star}(d_s,d_x)/\partial d_x$.
Clearly, 
\begin{equation}
\label{J_equals_Lambda}
	\jmath_{X}(x,d_s,d_x) = \Lambda_{Z^{\star}Y^{\star}}(x, \lambda_s^{\star}, \lambda_x^{\star}).
\end{equation}
For nonnegative integers $n_s$ and $n_x$, define
\begin{equation}
\label{d_ZY_sx}
\begin{aligned}
&\tilde{d}_{ZY,n_s,n_x}(x, \lambda_s, \lambda_x)\\ 
\triangleq & \frac{\mathbb{E}[\bar{\textsf{d}}_s^{n_s}(x,Z)\textsf{d}_x^{n_x}(x,Y)\exp(- \lambda_s \bar{\textsf{d}}_s(x,Z) - \lambda_x \textsf{d}_x(x,Y))]}{\mathbb{E}[\exp(- \lambda_s \bar{\textsf{d}}_s(x,Z) - \lambda_x \textsf{d}_x(x,Y))]}.
\end{aligned}
\end{equation}
The expectations in \eqref{Lambda_ZY} and \eqref{d_ZY_sx} are both with respect to the unconditional distribution of $(Z,Y)$.
We have the following properties of $\Lambda_{ZY}(x, \cdot, \cdot)$:
\begin{enumerate}
	\item
For $(d_s^0, d_x^0) \in \textsf{int}(\mathcal{D}_{sx})$, we have
	\begin{align}
	\label{pro_A}
		\nabla \mathbb{E}[\Lambda_{ZY}(X, \lambda^\star_{X,ZY,s}, \lambda^\star_{X,ZY,x})] = \boldsymbol{0},
	\end{align}
	where $\lambda^\star_{X,ZY,s} = - \partial R_{X; ZY}(d_s^0,d_x^0)/\partial d_s$, $\lambda^\star_{X,ZY,x} = - \partial R_{X; ZY}(d_s^0,d_x^0)/\partial d_x$, and $\nabla(\cdot)$ denotes the gradient with respect to $\boldsymbol{\lambda} = [\lambda_s, \lambda_x]^\top$.
	
	
	\item
	\begin{align}
	\label{pro_B}
	\nabla \Lambda_{ZY}(x, \lambda_{s}, \lambda_{x})  = \tilde{\boldsymbol{d}}_{ZY,1}(x, \lambda_s, \lambda_x)
	- \boldsymbol{d},
	\end{align}
	where 
	\begin{align}
	\tilde{\boldsymbol{d}}_{ZY,1}(x, \lambda_s, \lambda_x)  
	&\triangleq \begin{bmatrix}
	\tilde{d}_{ZY,1,0}(x, \lambda_s, \lambda_x) \\
	\tilde{d}_{ZY,0,1}(x, \lambda_s, \lambda_x)
	\end{bmatrix},
	\boldsymbol{d} 
	&= \begin{bmatrix}
	d_s \\
	d_x
	\end{bmatrix}.
	\end{align}
	
	\item
	\label{Hess}
	$\Lambda_{ZY}^{\prime_{\!s}\prime_{\!s}}(x, \lambda_{s}, \lambda_{x}) = (\log e)^{-1} [(\tilde{d}_{ZY,1,0}(x, \lambda_s, \lambda_x))^2 - \tilde{d}_{ZY,2,0}(x, \lambda_s, \lambda_x)]$, $\Lambda_{ZY}^{\prime_{\!x}\prime_{\!x}}(x, \lambda_{s}, \lambda_{x}) = (\log e)^{-1}[(\tilde{d}_{ZY,0,1}(x, \lambda_s, \lambda_x))^2 - \tilde{d}_{ZY,0,2}(x, \lambda_s, \lambda_x)]$, and $\Lambda_{ZY}^{\prime_{\!s}\prime_{\!x}}(x, \lambda_{s}, \lambda_{x}) = \Lambda_{ZY}^{\prime_{\!x}\prime_{\!s}}(x, \lambda_{s}, \lambda_{x}) = (\log e)^{-1}[\tilde{d}_{ZY,1,0}(x, \lambda_s, \lambda_x) \tilde{d}_{ZY,0,1}(x, \lambda_s, \lambda_x) - \tilde{d}_{ZY,1,1}(x, \lambda_s, \lambda_x)]$, where $(\cdot)^{\prime_{\!s}}$ and $(\cdot)^{\prime_{\!x}}$ represent taking partial derivatives with respect to $\lambda_s$ and $\lambda_x$, respectively. Clearly, 
	\begin{align}
	\label{pro_C}
		H ( \Lambda_{ZY}(x, \lambda_{s}, \lambda_{x})) \preceq 0,
	\end{align}
	where $H(\cdot)$ denotes the Hessian matrix with respect to $\boldsymbol{\lambda} = [\lambda_s, \lambda_x]^\top$.
\end{enumerate}

\begin{proof}[Proof]
	We only derive \eqref{pro_A}. Other properties can be obtained by direct computation. 
	
	
	Fix $(d_s^0, d_x^0) \in \textsf{int}(\mathcal{D}_{sx})$. By analysing the KKT conditions of problems \eqref{rate_distortion_function} and \eqref{general_rate_distortion_function}, we find that functions $d_s(\lambda_s, \lambda_x)$ and $d_x(\lambda_s, \lambda_x)$, which are defined by equation system $\lambda_{s} = - \partial R_{X; ZY}(d_s,d_x)/\partial d_s$ and $\lambda_{x} = - \partial R_{X; ZY}(d_s,d_x)/\partial d_x$, are both continuous differentiable in some neighborhood of point $(\lambda^\star_{X,ZY,s}, \lambda^\star_{X,ZY,x})$.
	Then we have
	\begin{align}
		&\frac{\partial \mathbb{E}[\Lambda_{ZY}(X, \lambda_{s}, \lambda_{x})]}{\partial \lambda_s}\bigg|_{\substack{\lambda_s = \lambda^\star_{X,ZY,s}\\ \lambda_x = \lambda^\star_{X,ZY,x}}}\\
		=&\frac{\partial \substack{R_{X; Z Y}(d_s(\lambda_s, \lambda_x),d_x(\lambda_s, \lambda_x))\\+\lambda_s d_s(\lambda_s, \lambda_x) + \lambda_x d_x(\lambda_s, \lambda_x) - \lambda_s d_s^0 - \lambda_x d_x^0}}{\partial \lambda_s}\bigg|_{\substack{\lambda_s = \lambda^\star_{X,ZY,s}\\ \lambda_x = \lambda^\star_{X,ZY,x}}} \label{pro1_eq1} \\
		=&(d_s(\lambda_s, \lambda_x) - d_s^0) \bigg|_{\substack{\lambda_s = \lambda^\star_{X,ZY,s}\\ \lambda_x = \lambda^\star_{X,ZY,x}}} \label{pro1_eq2} \\ 
		=&0,
	\end{align}
	where \eqref{pro1_eq1} holds since by \eqref{generalized_tilted_information_property}, $R_{X; Z Y}(d_s(\lambda_s, \lambda_x),d_x(\lambda_s, \lambda_x)) = \mathbb{E}[\Lambda_{ZY}(X, \lambda_{s}, \lambda_{x})] + \lambda_s d_s^0 + \lambda_x d_x^0 -\lambda_s d_s(\lambda_s, \lambda_x) - \lambda_x d_x(\lambda_s, \lambda_x)$, and \eqref{pro1_eq2} holds since $\partial R_{X; Z Y}(d_s(\lambda_s, \lambda_x),d_x(\lambda_s, \lambda_x))/\partial d_s = -\lambda_s$ and $\partial R_{X; Z Y}(d_s(\lambda_s, \lambda_x),d_x(\lambda_s, \lambda_x))/\partial d_x = -\lambda_x$.
	
	Similarly, 
	\begin{align}
	\frac{\partial \mathbb{E}[\Lambda_{ZY}(X, \lambda_{s}, \lambda_{x})]}{\partial \lambda_x}\bigg|_{\substack{\lambda_s = \lambda^\star_{X,ZY,s}\\ \lambda_x = \lambda^\star_{X,ZY,x}}}
	=0.
	\end{align}

\end{proof}

\begin{remark} 
	\label{remark_equal_lambda}
	Note that problem \eqref{general_rate_distortion_function} with $P_{ZY} = P_{Z^\star Y^\star}$ shares the same KKT coefficients with problem \eqref{rate_distortion_function} when $\boldsymbol{d} \in \mathcal{D}_{\mathrm{in}}$, where $P_{Z^\star Y^\star}$ achieves $ R_{X}\left(\boldsymbol{d}\right)$. Since the KKT coefficients of the two problems can be computed by $- \nabla R_{X; Z^\star Y^\star}\left(\boldsymbol{d}\right)$ and $- \nabla R_{X}\left(\boldsymbol{d}\right)$, respectively, we have
	\begin{equation}
		\boldsymbol{\lambda}^\star = \boldsymbol{\lambda}^\star_{X;Z^\star Y^\star},
	\end{equation}
	where $\boldsymbol{\lambda}^\star = [\lambda^\star_s, \lambda^\star_x]^\top = - \nabla R_{X}\left(\boldsymbol{d}\right)$ and $\boldsymbol{\lambda}^\star_{X;Z^\star Y^\star} = [\lambda^\star_{X,Z^\star Y^\star,s}, \lambda^\star_{X,Z^\star Y^\star,x}]^\top = - \nabla R_{X; Z^\star Y^\star}\left(\boldsymbol{d}\right)$.
\end{remark}

%

Now, we are ready to give Lemma \ref{Lemma_confrontation}.
\begin{lemma} 
	\label{Lemma_confrontation}
	Assuming that assumptions \ref{memoryless_sources}-\ref{Assumption_positive_definite} hold and fixing $\boldsymbol{d} \in \textsf{int}(\mathcal{D}_{sx})$, there exist $\delta_0$, $k_0 > 0$ such that for all $0 < \delta \leq \delta_0$, $ k \geq k_0$, there exist a set $F_k \subseteq \mathcal{A}^k$ and a constant $K_1$ such that
	\begin{equation}
	\label{1_sqrtn}
		\mathbb{P}\left[X^k \notin F_k\right] \leq \frac{K_1}{\sqrt{k}},
	\end{equation}
	and for all $x^k \in F_k$,
	\begin{equation}
	\label{close_lambda}
		\|\boldsymbol{\lambda}(x^k) - \boldsymbol{\lambda}^\star\| < \delta,
	\end{equation}
	where $\boldsymbol{\lambda}^\star = - \nabla R_{X}\left(\boldsymbol{d}\right)$ and $\boldsymbol{\lambda}(x^k) = - \nabla R_{\bar{X}; Z^\star Y^\star}\left(\boldsymbol{d}\right)$ with $\bar{X} \sim \textrm{type}\left(x^k\right)$.
\end{lemma}
\begin{proof}[Proof]
	
We first give the proof idea. Since $\boldsymbol{\lambda}(x^k)$ and $\boldsymbol{\lambda}^\star$ are defined through the optimization problem \eqref{general_rate_distortion_function}, directly studying them is rather challenging.  Lemma \ref{lemma_concave} offers an alternative approach, allowing us to indirectly characterize the distance between $\boldsymbol{\lambda}(x^k)$ and $\boldsymbol{\lambda}^\star$ by studying their functions. As shown below, function $\frac{1}{k} \sum_{i=1}^k \Lambda_{Z^{\star}Y^{\star}}(x_i, \boldsymbol{\lambda})$ meets our needs. We construct the set $F_k$ such that for all $x^k$ in $F_k$, function $\frac{1}{k} \sum_{i=1}^k \Lambda_{Z^{\star}Y^{\star}}(x_i, \boldsymbol{\lambda})$ satisfies the two conditions in Lemma \ref{lemma_concave} on an elliptical convex region with $\boldsymbol{\lambda}^\star$ as the center point. Since it can be shown that $\boldsymbol{\lambda}(x^k) \in \underset{\boldsymbol{\lambda}}{\arg\sup}\frac{1}{k} \sum_{i=1}^k \Lambda_{Z^{\star}Y^{\star}}(x_i, \boldsymbol{\lambda})$, by leveraging Lemma \ref{lemma_concave}, $\boldsymbol{\lambda}(x^k)$ must be within the elliptical region, thus bounding the distance between $\boldsymbol{\lambda}(x^k)$ and $\boldsymbol{\lambda}^\star$. The probability that $X^k$ falls in $F_k$ can be bounded by the Berry–Ess\'een theorem.

%
%

Prior to presenting a detailed proof, we provide some definitions.
Let $\boldsymbol{d} \in \textsf{int}(\mathcal{D}_{sx})$ and denote
\begin{equation}
\boldsymbol{H}_{\boldsymbol{d}} = H\left(R_{X; Z^\star Y^\star}\left(\boldsymbol{d}\right)\right).
\end{equation}
By analysing the KKT conditions of problems \eqref{rate_distortion_function} and \eqref{general_rate_distortion_function} and assuming that assumption \ref{Assumption_positive_definite} holds, we have $\boldsymbol{H}_{\boldsymbol{d}} \succ 0$ and $R_{X; Z^\star Y^\star}\left(\cdot\right)$ possesses continuous second-order partial derivatives in some neighborhood of $\boldsymbol{d}$.
Denote the eigen-decomposition of $\boldsymbol{H}_{\boldsymbol{d}}$ as
\begin{equation}
\label{eigen_decomposition}
\boldsymbol{H}_{\boldsymbol{d}} = \boldsymbol{Q} \boldsymbol{D} \boldsymbol{Q}^{\top},
\end{equation}
where the orthogonal matrix $\boldsymbol{Q} = [\boldsymbol{q}_1, \boldsymbol{q}_2]$ and the diagonal matrix $\boldsymbol{D} = \mathrm{diag}([\beta_1, \beta_2]^{\top})$ with $\beta_1 \geq \beta_2 > 0$. Fix a sufficiently small $\Delta > 0$ and a sufficiently large integer $N$ (which will be specified in the sequel), and denote
\begin{align}
\label{lambda_e}
\boldsymbol{\lambda}(\boldsymbol{v}) = \boldsymbol{\lambda}^\star - \frac{3 \Delta}{2} \boldsymbol{H}_{\boldsymbol{d}} \boldsymbol{v},\ \boldsymbol{v} \in \mathbb{R}^2,
\end{align}
\begin{align}
\label{e_j}
\boldsymbol{e}_j =& \left[\cos\left(\angle \boldsymbol{q}_1 +  \frac{2\pi j + \pi}{4N}\right), \sin\left(\angle \boldsymbol{q}_1 + \frac{2\pi j + \pi}{4N}\right)\right]^\top, \nonumber \\  
&j = 0,\dots, 4N-1,
\end{align}
\begin{align}
	\delta = \frac{3\Delta}{2}\beta_1,
\end{align}
where $\angle \boldsymbol{q}_1 = \arccos([\boldsymbol{q}_1]_1/\|\boldsymbol{q}_1\|)$.
Finally, $F_k$ is defined as
\begin{align}
	\label{F_k}
	F_k \triangleq \bigcap_{j=0}^{4N - 1} H_{k,j},
\end{align}
where
\begin{align}
\label{main_condition}
H_{k,j}\triangleq\bigg\{&x^k \in \mathcal{A}^k: \ \Delta \cdot \textsf{abs}(\boldsymbol{Q}^{\top} \boldsymbol{e}_j) \nonumber \\ 
&\preceq
\textsf{sgn}(\boldsymbol{Q}^{\top} \boldsymbol{e}_j) \circ \left(\frac{1}{k}  \sum_{i=1}^k \boldsymbol{Q}^{\top} \nabla\Lambda_{Z^\star Y^\star}(x_i, \boldsymbol{\lambda}(\boldsymbol{e}_j)) \right) \nonumber \\ 
&\preceq
2 \Delta \cdot \textsf{abs}(\boldsymbol{Q}^{\top}\! \boldsymbol{e}_j)\!\bigg\}, \forall j \!\in\! \{0, \dots, 4N\!-\!1\}, 
\end{align}
%
and $\textsf{abs}(\cdot)$ represents the vector of absolute values of the elements in the input vector, and $\textsf{sgn}(\cdot)$ represents the vector of signs of the elements in the input vector.\footnote{\eqref{e_j} ensures that $\boldsymbol{Q}^{\top} \boldsymbol{e}_j$ has non-zero elements, $j \in \{0, \dots, 4N - 1\}$.}


We now show that \eqref{close_lambda} holds for $x^k \in F_k$. 
Define set
\begin{equation}
	\mathcal{E}_{\boldsymbol{\lambda}^{\star}} \triangleq \{\boldsymbol{\lambda}(\boldsymbol{v}): \|\boldsymbol{v}\| \leq 1\},
\end{equation}
which is an elliptical closed region centered at $\boldsymbol{\lambda}^\star$ with boundary $\mathcal{B}_{\boldsymbol{\lambda}^\star} \triangleq \{\boldsymbol{\lambda}(\boldsymbol{v}): \|\boldsymbol{v}\| = 1\}$.
Next, we prove that $\boldsymbol{\lambda}(x^k) \in \textsf{Int}(\mathcal{E}_{\boldsymbol{\lambda}^{\star}})$, $x^k \in F_k$, which naturally implies $\|\boldsymbol{\lambda}(x^k) - \boldsymbol{\lambda}^\star\| < \delta$, $x^k \in F_k$.

Based on the geometric properties of ellipse $\mathcal{B}_{\boldsymbol{\lambda}^\star}$, it can be deduced that vector $\boldsymbol{Q}\boldsymbol{P}\boldsymbol{D}\boldsymbol{P}^{\top} \boldsymbol{Q}^{\top} \boldsymbol{e}$ points in the direction of the inward-pointing normal of ellipse $\mathcal{B}_{\boldsymbol{\lambda}^\star}$ at point $\boldsymbol{e}$ satisfying $\|\boldsymbol{e}\|=1$, where the permutation matrix $\boldsymbol{P} = \left[ \begin{smallmatrix} 0 & 1 \\ 1 & 0 \end{smallmatrix} \right]$. As a consequence, given $\langle \boldsymbol{P}\boldsymbol{D}\boldsymbol{P}^{\top} \boldsymbol{Q}^{\top} \boldsymbol{e}, \boldsymbol{Q}^{\top} \boldsymbol{w} \rangle = \langle \boldsymbol{Q}\boldsymbol{P}\boldsymbol{D}\boldsymbol{P}^{\top} \boldsymbol{Q}^{\top} \boldsymbol{e}, \boldsymbol{w} \rangle$, we have for all vectors $\boldsymbol{e}$ and $\boldsymbol{w}$ that satisfy  $\|\boldsymbol{e}\| = 1$ and $\langle \boldsymbol{P}\boldsymbol{D}\boldsymbol{P}^{\top} \boldsymbol{Q}^{\top} \boldsymbol{e}, \boldsymbol{Q}^{\top} \boldsymbol{w} \rangle > 0$, there exists $t > 0$ such that
\begin{equation}
\label{hhh}
\boldsymbol{\lambda}(\boldsymbol{e}) + t\boldsymbol{w} \in \textsf{Int}(\mathcal{E}_{\boldsymbol{\lambda}^\star}).
\end{equation}

Let
\begin{align}
\label{inner_product_func}
f(\boldsymbol{v}, k) = \Big \langle \boldsymbol{P}\boldsymbol{D}\boldsymbol{P}^{\top} \boldsymbol{Q}^{\top} \boldsymbol{v}, \boldsymbol{Q}^{\top} \Big(\frac{1}{k} \sum_{i=1}^k& \nabla\Lambda_{Z^\star Y^\star}(x_i, \boldsymbol{\lambda}(\boldsymbol{v})) \Big) \Big\rangle.
\end{align}
For further analysis, we now give some properties of function $f(\boldsymbol{v}, k)$. Since the distortion measures $\textsf{d}_s$ and $\textsf{d}_x$ are both finite, the second-order partial derivatives of functions $\{\Lambda_{Z^\star Y^\star}(x, \cdot, \cdot)\}_{x \in \mathcal{A}}$, as computed in property \ref{Hess}, are bounded. Moreover, they have shared upper and lower bounds since alphabet $\mathcal{A}$ is finite.
As a consequence, treating $\frac{1}{k}  \sum_{i=1}^k \nabla\Lambda_{Z^\star Y^\star}(x_i, \boldsymbol{\lambda})$ as a function of $(\boldsymbol{\lambda}, k)$, we have the partial derivative of $\frac{1}{k}  \sum_{i=1}^k \nabla\Lambda_{Z^\star Y^\star}(x_i, \boldsymbol{\lambda})$ with respect to $\boldsymbol{\lambda}$ is bounded. Juxtaposing this with definition \eqref{lambda_e}, function $f(\boldsymbol{v}, k)$ is differentiable with respect to $\boldsymbol{v}$, and for any compact set $\mathcal{V} \subseteq \{\boldsymbol{v} \in \mathbb{R}^2: \boldsymbol{\lambda}(\boldsymbol{v}) \succeq \boldsymbol{0}\}$, the partial derivative of $f(\boldsymbol{v}, k)$ with respect to $\boldsymbol{v}$, which is also a function of $(\boldsymbol{v}, k)$, is bounded on $\mathcal{V} \times \{1,2,\dots\}$.

By checking the four vertices of the rectangular region $\{\boldsymbol{v}:\Delta \cdot \textsf{abs}(\boldsymbol{Q}^{\top} \boldsymbol{e}_j)
\preceq
\textsf{sgn}(\boldsymbol{Q}^{\top} \boldsymbol{e}_j) \circ (\boldsymbol{Q}^{\top}\boldsymbol{v})
\preceq
2 \Delta \cdot \textsf{abs}(\boldsymbol{Q}^{\top} \boldsymbol{e}_j)\}$, for all $x^k \in H_{k,j}$, we have
\begin{align}
\label{PDP_angle}
	f(\boldsymbol{e}_j, k) \geq \Delta \beta_2, \ \forall j \in \{0, \dots, 4N-1\}.
\end{align}
Note that $\mathcal{V}_{ub} \triangleq \cup_{j=0}^{4N-1}\mathcal{B}^{l_{2}}_{\sqrt{2 - 2 \cos\frac{\pi}{4N}}}(\boldsymbol{e}_j)$ with $\mathcal{B}^{l_2}_{d}(\boldsymbol{a})\triangleq \{\boldsymbol{v}: \|\boldsymbol{v} - \boldsymbol{a}\| \leq d\}$ is a compact subset of $\{\boldsymbol{v} \in \mathbb{R}^2: \boldsymbol{\lambda}(\boldsymbol{v}) \succeq \boldsymbol{0}\}$ for sufficiently small $\Delta$. Consequently, given that \eqref{PDP_angle} holds and $\Delta \beta_2$ is strictly positive, by the properties of function $f(\boldsymbol{v}, k)$ introduced before, there exists a sufficiently large $N$ such that for $j \in \{0, \dots, 4N-1\}$,
%
%
\begin{align}
\label{PDP_T}
f(\boldsymbol{v}, k) > 0, \ \forall \boldsymbol{v} \in \mathcal{B}^{l_{2}}_{\sqrt{2 - 2 \cos\frac{\pi}{4N}}}(\boldsymbol{e}_j), \ k\in \{1,2,\dots\}.
\end{align}
Note that \eqref{PDP_T} implies that the choice of $N$ is independent of $k$.
Since $\{\boldsymbol{v}: \|\boldsymbol{v}\| = 1\}\subseteq \cup_{j=0}^{4N-1}\mathcal{B}^{l_{2}}_{\sqrt{2 - 2 \cos\frac{\pi}{4N}}}(\boldsymbol{e}_j)$, considering both \eqref{PDP_angle} and \eqref{PDP_T} as well as the discussion related to them, for a sufficiently large $N$, it holds that for all $x^k \in F_k$, 
\begin{align}
\label{PDP_T_new}
f(\boldsymbol{e}, k) > 0,\ \forall \boldsymbol{e}: \| \boldsymbol{e}\|=1, \ k\in \{1,2,\dots\}.
\end{align}
Immediately, by \eqref{hhh} and its related discussion, for a sufficiently large $N$, it holds that for all $x^k \in F_k$, there exists $t > 0$ such that
\begin{equation}
\label{hhh_new}
\boldsymbol{\lambda}(\boldsymbol{e}) + t \cdot \frac{1}{k} \sum_{i=1}^k \nabla\Lambda_{Z^\star Y^\star}(x_i, \boldsymbol{\lambda}(\boldsymbol{e})) \in \textsf{Int}(\mathcal{E}_{\boldsymbol{\lambda}^\star}), \ \boldsymbol{e}: \|\boldsymbol{e}\|=1.
\end{equation}

Additionally, by property \eqref{pro_C}, $\frac{1}{k} \sum_{i=1}^k \Lambda_{Z^{\star}Y^{\star}}(x_i, \boldsymbol{\lambda})$ is concave.
By property \eqref{pro_A},
\begin{equation}
\nabla \left(\frac{1}{k} \sum_{i=1}^k \Lambda_{Z^{\star}Y^{\star}}(x_i, \boldsymbol{\lambda}(x^k))\right) = \boldsymbol{0}.
\end{equation}
As a consequence, we have
\begin{equation}
\label{lambda_argsup}
\boldsymbol{\lambda}(x^k) \in \underset{\boldsymbol{\lambda}}{\arg\sup}\frac{1}{k} \sum_{i=1}^k \Lambda_{Z^{\star}Y^{\star}}(x_i, \boldsymbol{\lambda}).
\end{equation}
Given \eqref{hhh_new}, \eqref{lambda_argsup}, the concavity of $\frac{1}{k} \sum_{i=1}^k \Lambda_{Z^{\star}Y^{\star}}(x_i, \boldsymbol{\lambda})$ and the convexity of $\mathcal{E}_{\boldsymbol{\lambda}^{\star}}$, we have $\boldsymbol{\lambda}(x^k) \in \textsf{Int}(\mathcal{E}_{\boldsymbol{\lambda}^{\star}})$, $x^k \in F_k$ using Lemma \ref{lemma_concave}.

We next show that there exits a constant $K_1$ such that \eqref{1_sqrtn} holds for $F_k$ defined by \eqref{F_k}.

Since $R_{X; Z^\star Y^\star}\left(\cdot\right)$ possesses continuous second-order partial derivatives at point $\boldsymbol{d_0}$,
by Taylor's theorem, we have as $\frac{3\Delta}{2} \boldsymbol{v} \to \boldsymbol{0}$,
\begin{align}
\label{lambda_e_o}
	\boldsymbol{\lambda}(\boldsymbol{v}) = \boldsymbol{\lambda}_{X}(\boldsymbol{v}) + o\left(\frac{3\Delta}{2} \boldsymbol{v}\right),
\end{align}
where $o\left(\frac{3\Delta}{2} \boldsymbol{v}\right)$ is a two-dimensional vector satisfying $\|o\left(\frac{3\Delta}{2} \boldsymbol{v}\right)\| = o\left(\|\frac{3\Delta}{2} \boldsymbol{v}\|\right)$, and
\begin{equation}
\boldsymbol{\lambda}_{X}(\boldsymbol{v}) = - \nabla R_{X; Z^\star Y^\star}\left(\boldsymbol{d} + \frac{3\Delta}{2} \boldsymbol{v}\right).
\end{equation}
Note that
\begin{align}
	&\mathbb{E}[\nabla\Lambda_{Z^\star Y^\star}(X, \boldsymbol{\lambda}_X(\boldsymbol{v}))]  \\
	=&\mathbb{E}[\tilde{\boldsymbol{d}}_{Z^\star Y^\star,1}(X, \boldsymbol{\lambda}_X(\boldsymbol{v}))]
	- \boldsymbol{d} \label{mean_1}\\
	=&\frac{3\Delta}{2}\boldsymbol{v},\label{mean_2}
\end{align}
where \eqref{mean_1} is by property \eqref{pro_B}, and \eqref{mean_2} is by properties \eqref{pro_A} and \eqref{pro_B}. By \eqref{lambda_e_o}, \eqref{mean_2} and the differentiability of $\mathbb{E}[\nabla\Lambda_{Z^\star Y^\star}(X, \cdot)]$, we have as $\frac{3\Delta}{2} \boldsymbol{v} \to \boldsymbol{0}$,
\begin{align}
	\mathbb{E}[\nabla\Lambda_{Z^\star Y^\star}(X, \boldsymbol{\lambda}(\boldsymbol{v}))] = \frac{3\Delta}{2}\boldsymbol{v} + o\left(\frac{3\Delta}{2} \boldsymbol{v}\right).
\end{align}
Thus, as long as $\Delta$ is small enough, we have
\begin{align}
\label{mean_in_set}
&\Delta \cdot \textsf{abs}(\boldsymbol{Q}^{\top} \boldsymbol{e}_j) \nonumber \\ 
&\prec
\textsf{sgn}(\boldsymbol{Q}^{\top} \boldsymbol{e}_j) \circ \left(\boldsymbol{Q}^{\top} \mathbb{E}[\nabla\Lambda_{Z^\star Y^\star}(X, \boldsymbol{\lambda}(\boldsymbol{e}_j))]\right) \nonumber \\ 
&\prec
2 \Delta \cdot \textsf{abs}(\boldsymbol{Q}^{\top}\! \boldsymbol{e}_j),\ \forall j \in \{0, \dots, 4N-1\}.
\end{align}
Given \eqref{mean_in_set}, by first applying the Berry–Ess\'een theorem separately to the two elements of $\frac{1}{k}  \sum_{i=1}^k \boldsymbol{Q}^{\top} \nabla\Lambda_{Z^\star Y^\star}(X_i, \boldsymbol{\lambda}(\boldsymbol{e}_j))$, and then using the union bound, we have, as $k$ is large enough, there exists a constant $L_1$ such that
\begin{align}
\label{H_k_j}
	\mathbb{P}[X^k \notin H_{k,j}] \leq \frac{L_1}{\sqrt{k}}, \ \forall j \!\in\! \{0, \dots, 4N\!-\!1\}.
\end{align}
Finally, \eqref{1_sqrtn} holds by \eqref{H_k_j} and the union bound.

\end{proof}

The counterpart of Lemma \ref{Lemma_confrontation} for $\boldsymbol{d} \in \textsf{int}(\mathcal{D}_{\bar{s}x}) \cup \textsf{int}(\mathcal{D}_{s\bar{x}}) \cup \textsf{int}(\mathcal{D}_{\bar{s}\bar{x}})$ is as follows.

\begin{lemma} 
	\label{Lemma_confrontation_02}
	Assuming that assumptions \ref{memoryless_sources}-\ref{Assumption_positive_definite} hold and fixing $\boldsymbol{d} \in \textsf{int}(\mathcal{D}_{\bar{s}x}) \cup \textsf{int}(\mathcal{D}_{s\bar{x}}) \cup \textsf{int}(\mathcal{D}_{\bar{s}\bar{x}})$, there exist $\delta_0$, $k_0 > 0$ such that for all $0 < \delta \leq \delta_0$, $ k \geq k_0$, there exist a set $F_k \subseteq \mathcal{A}^k$ and a constant $K_1$ such that
	\begin{equation}
	\label{1_sqrtn_02}
	\mathbb{P}\left[X^k \notin F_k\right] \leq \frac{K_1}{\sqrt{k}},
	\end{equation}
	and for all $x^k \in F_k$, we have
	\begin{equation}
	\label{close_lambda_02}
	\left\{
	\begin{array}{r@{\;}l}
	|\lambda_x(x^k) - \lambda_x^\star|<\delta,\, \lambda_s(x^k) = 0 & \quad \text{if} \ {\boldsymbol{d} \in \textsf{int}(\mathcal{D}_{\bar{s}x}),} \\[1em]
	|\lambda_s(x^k) - \lambda_s^\star|<\delta,\ \lambda_x(x^k) = 0 & \quad \text{if} \ {\boldsymbol{d} \in \textsf{int}(\mathcal{D}_{s\bar{x}}),} \\[1em]
	\lambda_s(x^k) = 0,\, \lambda_x(x^k) = 0 & \quad \text{if} \ {\boldsymbol{d} \in \textsf{int}(\mathcal{D}_{\bar{s}\bar{x}}),}
	\end{array}
	\right.
	\end{equation}
	where $\lambda_{s}^\star = - \partial R_{X}(d_s,d_x)/\partial d_s$, $\lambda_{x}^\star = - \partial R_{X}(d_s,d_x)/\partial d_x$, $\lambda_s(x^k) = - \partial R_{\bar{X}; Z^\star Y^\star}(d_s,d_x)/\partial d_s$, $\lambda_x(x^k) = - \partial R_{\bar{X}; Z^\star Y^\star}(d_s,d_x)/\partial d_x$, and $\bar{X} \sim \textrm{type}\left(x^k\right)$.
\end{lemma}

\begin{proof}[Proof]
	This proof is based on the one-constraint result \cite[Lemma 4]{Kostina2012Fixed} and several properties of problem \eqref{general_rate_distortion_function} with at least one loose constraint.
	We first consider $\boldsymbol{d} \in \textsf{int}(\mathcal{D}_{s\bar{x}})$. For $\boldsymbol{d} \in \textsf{int}(\mathcal{D}_{s\bar{x}})$, we have
	\begin{equation}
	\label{lambda_s_bar_x}
	\lambda_{x}^\star = 0.
	\end{equation}
	By the definition of $\mathcal{D}_{s\bar{x}}$ and the discussion in Remark \ref{remark_equal_lambda}, we have
	\begin{equation}
	\label{lambda_generalized_s_bar_x}
	\boldsymbol{\lambda}^\star_{X;Z^\star Y^\star} = [\lambda_s^{\star},\, 0]^{\top},
	\end{equation}
	and there exists an optimal solution $P_{\ddot{Z}^\star \ddot{Y}^\star|X}$ of problem \eqref{general_rate_distortion_function} such that
	\begin{align}
	\mathbb{E}\left[\bar{\textsf{d}}_s(X,\ddot{Z}^\star)\right] &= d_s, \label{equal_ds} \\
	\mathbb{E}\left[\textsf{d}_x(X,\ddot{Y}^\star)\right] &< d_x. \label{less_than_dx}
	\end{align}
	Additionally, the objective of problem \eqref{general_rate_distortion_function} can be rewritten as
	\begin{align}
	\label{general_rate_distortion_function_rewrite}
	&D(P_{ZY|X} \Vert P_{Z^\star Y^\star}| P_X) \nonumber \\
	=&D(P_{Z|X} \Vert P_{Z^\star}| P_X) + D(P_{Y|XZ} \Vert P_{Y^\star|Z^\star}| P_XP_{Z|X}).
	\end{align}
	Note that the first term on the right-hand side of \eqref{general_rate_distortion_function_rewrite} is only related to $P_{Z|X}$ and the second term is greater than or equal to $0$ with the equality holds if and only if $P_{Y|XZ} = P_{Y^\star|Z^\star}$. Combining this with \eqref{equal_ds} and \eqref{less_than_dx}, the optimal solution $P_{\ddot{Z}^\star \ddot{Y}^\star|X}$ must satisfy
	\begin{equation}
	\label{Markov}
	P_{\ddot{Z}^\star \ddot{Y}^\star|X} = P_{Y^\star|Z^\star} P_{\ddot{Z}^\star|X},
	\end{equation}
	where $P_{\ddot{Z}^\star|X}$ is the optimal solution of problem
	\begin{subequations}
		\label{general_rate_distortion_function_one_distortion}
		\begin{align}
		R_{X; Z^\star}(d_s) \triangleq \min_{P_{Z|X}}\ &D(P_{Z|X} \Vert P_{Z^\star}| P_X) \\
		\mathrm{s.t.}\
		&\mathbb{E}\left[\bar{\textsf{d}}_s(X,Z)\right] \leq d_s
		\end{align}
	\end{subequations}
	that satisfies \cite{Kostina2012Fixed}
	\begin{equation}
	\label{optimal_solution_gen}
	P_{\ddot{Z}^\star|X}(z|x) = \frac{\exp(-\lambda_{X;Z^\star Y^\star, s}^\star\bar{\textsf{d}}_s(x,z))P_{Z^\star}(z)}{\mathbb{E}\left[\exp(-\lambda_{X;Z^\star Y^\star, s}^\star\bar{\textsf{d}}_s(x,Z^\star))\right]}.
	\end{equation}
	Juxtaposing \eqref{lambda_generalized_s_bar_x}, \eqref{less_than_dx}, \eqref{Markov}, and \eqref{optimal_solution_gen}, we have
	\begin{align}
	\label{sum}
	\sum_{x,y,z} &\bigg(\frac{\exp(-\lambda_s^\star\bar{\textsf{d}}_s(x,z))P_{Z^\star}(z)}{\mathbb{E}\left[\exp(-\lambda_s^\star\bar{\textsf{d}}_s(x,Z^\star))\right]} \nonumber \\
	&\cdot P_X(x) P_{Y^\star|Z^\star}(y|z)\textsf{d}_x(x,y)\bigg) =\mathbb{E}\left[\textsf{d}_x(X,\ddot{Y}^\star)\right]< d_x.
	\end{align}
	
	By \cite[Lemma 4]{Kostina2012Fixed}, there exist $\delta_1$, $k_1 > 0$ such that for all $0 < \delta_2 \leq \delta_1$, $ k \geq k_1$, there exist a set $W_k \subseteq \mathcal{A}^k$ and a constant $Z_1$ such that
	\begin{equation}
	\label{2_sqrtn}
	\mathbb{P}\left[X^k \notin W_k\right] \leq \frac{Z_1}{\sqrt{k}},
	\end{equation}
	and for all $x^k \in W_k$,
	\begin{equation}
	\label{close_lambda_2}
	|\lambda_s(x^k) - \lambda_s^\star| < \delta_2,
	\end{equation}
	where $\lambda_s(x^k) = -  R_{\bar{X}; Z^\star}'(d_s) = - \partial R_{\bar{X}; Z^\star Y^\star}(d_s,d_x)/\partial d_s$ with $\bar{X} \sim \textrm{type}\left(x^k\right)$. By \eqref{sum}, we can always find a small enough $\delta_2 > 0$ such that for all $\lambda_s(x^k)$ satisfying $|\lambda_s(x^k) - \lambda_s^\star| < \delta_2$, we have
	\begin{align}
	\label{sum1}
	\sum_{x,y,z} &\bigg(\frac{\exp(-\lambda_s(x^k)\bar{\textsf{d}}_s(x,z))P_{Z^\star}(z)}{\mathbb{E}\left[\exp(-\lambda_s(x^k)\bar{\textsf{d}}_s(x,Z^\star))\right]} \nonumber \\
	&\cdot P_X(x) P_{Y^\star|Z^\star}(y|z)\textsf{d}_x(x,y)\bigg) < d_x.
	\end{align}
	Then, by \eqref{definition_typical_set} and \eqref{sum1}, for all $x^k \in \mathcal{T}_k \cap W_k$ with $W_k$ corresponding to $\delta_2$ and large enough $k$'s, we have
	\begin{align}
	\label{sum2}
	\sum_{x,y,z} &\bigg(\frac{\exp(-\lambda_s(x^k)\bar{\textsf{d}}_s(x,z))P_{Z^\star}(z)}{\mathbb{E}\left[\exp(-\lambda_s(x^k)\bar{\textsf{d}}_s(x,Z^\star))\right]} \nonumber \\
	&\cdot P_{\bar{X}}(x) P_{Y^\star|Z^\star}(y|z)\textsf{d}_x(x,y)\bigg) < d_x.
	\end{align}
	By \eqref{sum2} and the argument in \eqref{general_rate_distortion_function_rewrite}-\eqref{sum}, we have the optimal solution $P_{\dddot{Z}^\star \dddot{Y}^\star|X}$ of problem \eqref{general_rate_distortion_function} with $X = \bar{X}$ satisfies
	\begin{align}
	\label{loose_constraint}
	\mathbb{E}\left[\textsf{d}_x(\bar{X},\dddot{Y}^\star)\right] < d_x,
	\end{align}
	implying $\lambda_{x}(x^k)=0$. Note that, by \eqref{Xk_doesnot_typical}, \eqref{2_sqrtn}, and the union bound,
	\begin{equation}
	\mathbb{P}[x^k \notin \mathcal{T}_k \cap W_k] \leq \frac{Z_1 + 2 |\mathcal{A}|}{\sqrt{k}}.
	\end{equation}
	By setting $\delta_0 = \delta_2$, $F_k = \mathcal{T}_k \cap W_k$, and $K_1 = Z_1 + 2 |\mathcal{A}|$, we have Lemma \ref{Lemma_confrontation_02} holds for $\boldsymbol{d} \in \textsf{int}(\mathcal{D}_{s\bar{x}})$. The proof for $\boldsymbol{d} \in \textsf{int}(\mathcal{D}_{\bar{s}x})$ is the same as that for $\boldsymbol{d} \in \textsf{int}(\mathcal{D}_{s\bar{x}})$ and is omitted here.
	
	We now consider the case $\boldsymbol{d} \in \textsf{int}(\mathcal{D}_{\bar{s}\bar{x}})$. 
	By the definition of $\mathcal{D}_{\bar{s}\bar{x}}$ and the discussion in Remark \ref{remark_equal_lambda}, there exists an optimal solution $P_{\ddot{Z}^\star \ddot{Y}^\star|X}$ of problem \eqref{general_rate_distortion_function} such that
	\begin{align}
	\mathbb{E}\left[\bar{\textsf{d}}_s(X,\ddot{Z}^\star)\right] &< d_s, \label{less_than_ds_01} \\
	\mathbb{E}\left[\textsf{d}_x(X,\ddot{Y}^\star)\right] &< d_x. \label{less_than_dx_01}
	\end{align}
	Since the objective of problem \eqref{general_rate_distortion_function} $D(P_{ZY|X} \Vert P_{Z^\star Y^\star}| P_X) \geq 0$ with equality holds if and only if $P_{ZY|X}=P_{Z^\star Y^\star}$, we have $P_{\ddot{Z}^\star \ddot{Y}^\star|X} = P_{Z^\star Y^\star}$ and
	\begin{align}
	\sum_{x,z} P_X(x) P_{Z^\star}(z) \bar{\textsf{d}}_s(x,z) &< d_s,\\
	\sum_{x,y} P_X(x) P_{Y^\star}(y) \textsf{d}_x(x,y) &< d_x.
	\end{align}
	We can always find a large enough $k_0$ such that, for all $k > k_0$, we have for all $x^k \in \mathcal{T}_k$, 
	\begin{align}
	\sum_{x,z} P_{\bar{X}}(x) P_{Z^\star}(z) \bar{\textsf{d}}_s(x,z) &< d_s,\\
	\sum_{x,y} P_{\bar{X}}(x) P_{Y^\star}(y) \textsf{d}_x(x,y) &< d_x,
	\end{align}
	implying $	\lambda_s(x^k) = 0$ and $\lambda_x(x^k) = 0$. Finally, since $\mathbb{P}\left[X^k \notin \mathcal{T}_k \right] \leq \frac{2 |\mathcal{A}|}{\sqrt{k}}$ by \eqref{Xk_doesnot_typical}, Lemma \ref{Lemma_confrontation_02} holds for $\boldsymbol{d} \in \textsf{int}(\mathcal{D}_{\bar{s}\bar{x}})$.
\end{proof}

In light of Lemmas \ref{Lemma_confrontation} and \ref{Lemma_confrontation_02}, we have Lemma \ref{Lemma_Lambda_logn} as follows.
\begin{lemma} 
	\label{Lemma_Lambda_logn}
	Assuming that assumptions \ref{memoryless_sources}-\ref{Assumption_positive_definite} hold and fixing $\boldsymbol{d} \in \mathcal{D}_{\mathrm{in}}$, there exist constants $k_0$, $K_2$, $C_2 > 0$ such that for $k \geq k_0$,
	\begin{align}
	\label{eq_Lemma_Lambda_logn}
	&\mathbb{P}\left[\sum_{i=1}^k\! \Lambda_{Z^\star Y^\star}(X_i, \boldsymbol{\lambda}(X^k)) \!\leq\! \sum_{i=1}^k\! \Lambda_{Z^\star Y^\star}(X_i, \boldsymbol{\lambda}^{\star}) \!+\! C_2\log k \right]\nonumber \\ 
	>& 1 - \frac{K_2}{\sqrt{k}}.
	\end{align}
\end{lemma}

\begin{proof}[Proof]
	
	Recall that $\mathcal{D}_{\mathrm{in}} \triangleq \textsf{int}(\mathcal{D}_{sx}) \cup \textsf{int}(\mathcal{D}_{\bar{s}x}) \cup \textsf{int}(\mathcal{D}_{s\bar{x}}) \cup \textsf{int}(\mathcal{D}_{\bar{s}\bar{x}})$. For $\boldsymbol{d} \in \textsf{int}(\mathcal{D}_{\bar{s}\bar{x}})$, \eqref{eq_Lemma_Lambda_logn} is a direct consequence of Lemma \ref{Lemma_confrontation_02} since $\Lambda_{Z^\star Y^\star}(X, \boldsymbol{0}) = 0$ almost surely. For $\boldsymbol{d} \in \textsf{int}(\mathcal{D}_{\bar{s}x}) \cup \textsf{int}(\mathcal{D}_{s\bar{x}})$, by only considering $x^k \in F_k$, we have $\lambda_s(x^k) = 0$ or $\lambda_x(x^k)=0$ by Lemma \ref{Lemma_confrontation_02}, and thus, the case degenerates to that considered in \cite[Lemma 5]{Kostina2012Fixed} and can be proved similarly. In the following, we prove Lemma \ref{Lemma_Lambda_logn} for $\boldsymbol{d} \in \textsf{int}(\mathcal{D}_{sx})$ based on Lemma \ref{Lemma_confrontation}.
	
	This proof is a two-dimensional generalization of that of \cite[Lemma 5]{Kostina2012Fixed}, following similar reasoning but employing a more general treatment to yield a more general result applicable to higher dimensions. For all $x^k \in F_k$ (which is defined in \eqref{F_k}), we have
	\begin{align}
	\label{Lemma_core}
		&\sum_{i=1}^k \left[\Lambda_{Z^\star Y^\star}(x_i, \boldsymbol{\lambda}(x^k)) - \Lambda_{Z^\star Y^\star}(x_i, \boldsymbol{\lambda}^{\star})\right]\\
		=&\sup_{\|\boldsymbol{\theta}\| < \delta} \sum_{i=1}^k \left[\Lambda_{Z^\star Y^\star}(x_i, \boldsymbol{\lambda}^{\star} + \boldsymbol{\theta}) - \Lambda_{Z^\star Y^\star}(x_i, \boldsymbol{\lambda}^{\star})\right] \label{eq1} \\
		=&\sup_{\|\boldsymbol{\theta}\| < \delta} \boldsymbol{\theta}^\top \boldsymbol{b}(x^k)
		- \frac{1}{2}\boldsymbol{\theta}^\top \boldsymbol{H}(x^k, \boldsymbol{\xi}_k) \boldsymbol{\theta} \label{eq2}\\
		\leq&\sup_{\boldsymbol{\theta} \in \mathbb{R}^2} \boldsymbol{\theta}^\top \boldsymbol{b}(x^k)
		- \frac{1}{2}\boldsymbol{\theta}^\top \boldsymbol{H}(x^k, \boldsymbol{\xi}_k) \boldsymbol{\theta} \label{eq3}\\		
		=& \frac{1}{2} \boldsymbol{b}(x^k)^{\top} \left(\boldsymbol{H}(x^k, \boldsymbol{\xi}_k)\right)^{-1} \boldsymbol{b}(x^k) \label{eq4} \\
		\leq& \frac{1}{2} \big\|\boldsymbol{b}(x^k)\big\| \cdot \big\|\left(\boldsymbol{H}(x^k, \boldsymbol{\xi}_k)\right)^{-1} \boldsymbol{b}(x^k)\big\| \label{eq5} \\
		\leq& \frac{1}{2} \big\|\boldsymbol{b}(x^k)\big\|^2 \cdot \lambda_{\max}\big(\left(\boldsymbol{H}(x^k, \boldsymbol{\xi}_k)\right)^{-1}\big) \label{eq6} \\
		\leq&\frac{\big\|\boldsymbol{b}(x^k)\big\|^2 \cdot \mathrm{tr}(\boldsymbol{H}(x^k, \boldsymbol{\xi}_k))}{2\det(\boldsymbol{H}(x^k, \boldsymbol{\xi}_k))} \label{eq7}\\
		\leq&\frac{\big\|\boldsymbol{b}(x^k)\big\|^2 \cdot \mathrm{tr}(\boldsymbol{H}(x^k, \boldsymbol{\xi}_k))}{2\left(\sum_{i=1}^k\sqrt{\det(-H(\Lambda_{Z^\star Y^\star}(x_i, \boldsymbol{\lambda}^{\star} + \boldsymbol{\xi}_k)))}\right)^2} \label{eq8}\\
		\leq&\frac{\big\|\boldsymbol{b}(x^k)\big\|^2 \cdot T(x^k)}{2\left(D(x^k)\right)^2} \label{eq9}
	\end{align}
	where
	\begin{itemize}
		\item \eqref{eq1} is by \eqref{close_lambda} and \eqref{lambda_argsup};
		\item \eqref{eq2} holds for some $\|\boldsymbol{\xi}_k\| < \delta$ by Taylor's theorem, where we denoted
		\begin{align}
		\boldsymbol{b}(x^k) =& \sum_{i=1}^k \nabla\Lambda_{Z^\star Y^\star}(x_i, \boldsymbol{\lambda}^{\star}),\\
		\boldsymbol{H}(x^k, \boldsymbol{\xi}_k) =& -\sum_{i=1}^k H(\Lambda_{Z^\star Y^\star}(x_i, \boldsymbol{\lambda}^{\star} + \boldsymbol{\xi}_k)).
		\end{align}
		\item \eqref{eq4} is by the positive definiteness of $\boldsymbol{H}(x^k, \boldsymbol{\xi}_k)$;
		\item \eqref{eq5} is by the Cauchy–Schwarz inequality;
		\item \eqref{eq6} is by the compatibility of vector $l_2$ norm and spectral norm, where $\lambda_{\max}(\cdot)$ returns the maximum eigenvalue of the input matrix;
		\item \eqref{eq8} is by the Minkowski inequality;
		\item in \eqref{eq9}, we denoted
		\begin{align}
		T(x^k) =& \sum_{i=1}^k \sup_{\|\boldsymbol{\theta}\| < \delta} \textrm{tr} (-H(\Lambda_{Z^\star Y^\star}(x_i, \boldsymbol{\lambda}^{\star} + \boldsymbol{\theta}))),\\
		D(x^k) =& \sum_{i=1}^k\inf_{\|\boldsymbol{\theta}\| < \delta}\sqrt{\det(-H(\Lambda_{Z^\star Y^\star}(x_i, \boldsymbol{\lambda}^{\star} + \boldsymbol{\theta})))}.
		\end{align}
	\end{itemize}
	
	By \eqref{pro_A}, we have $\mathbb{E}[\nabla\Lambda_{Z^\star Y^\star}(X, \boldsymbol{\lambda}^{\star})] = \boldsymbol{0}$. 
Denote $V_{b1} = \textrm{Var}\big[[\nabla\Lambda_{Z^\star Y^\star}(X, \boldsymbol{\lambda}^{\star})]_1\big]$ and $V_{b2} = \textrm{Var}\big[[\nabla\Lambda_{Z^\star Y^\star}(X, \boldsymbol{\lambda}^{\star})]_2\big]$, where $[\boldsymbol{v}]_i$ denotes the $i$-th element in vector $\boldsymbol{v}$. Then, there exists a constant $K_{21}$ such that
	\begin{align}
		&\mathbb{P}\left[\big\|\boldsymbol{b}(X^k)\big\|^2 > (V_{b1} + V_{b2}) k \log_{e}k \right] \\
		\leq&\mathbb{P}\left[\big|[\boldsymbol{b}(X^k)]_1\big| \!>\! \sqrt{V_{b1} k \log_{e}k} \cup \big|[\boldsymbol{b}(X^k)]_2\big| \!>\! \sqrt{V_{b2} k \log_{e}k} \right] \\
		\leq&\frac{K_{21}}{\sqrt{k}}, \label{K_21}
	\end{align}
	where \eqref{K_21} is by the union bound and the argument for \cite[(326)-(329)]{Kostina2012Fixed}.
	\begin{align}
		Z_t =& \sup_{\|\boldsymbol{\theta}\| < \delta} \textrm{tr} (-H(\Lambda_{Z^\star Y^\star}(X, \boldsymbol{\lambda}^{\star} + \boldsymbol{\theta}))), \\
		Z_d =& \inf_{\|\boldsymbol{\theta}\| < \delta}\sqrt{\det(-H(\Lambda_{Z^\star Y^\star}(X, \boldsymbol{\lambda}^{\star} + \boldsymbol{\theta})))},\\
		\mu_t =& \mathbb{E}\left[\textrm{tr} (-H(\Lambda_{Z^\star Y^\star}(X, \boldsymbol{\lambda}^{\star})))\right],\\
		\mu_d =& \mathbb{E}\left[\sqrt{\det(-H(\Lambda_{Z^\star Y^\star}(X, \boldsymbol{\lambda}^{\star})))}\right].
	\end{align}
	For small enough $\delta$, we have $\mathbb{E}[Z_t] \leq \frac{5\mu_t}{4}$ and $\mathbb{E}[Z_d] \geq \frac{3\mu_d}{4}$. As a consequence, by the Berry-Ess\'een theorem, there exist constants $K_{22}$ and $K_{23}$ such that
	\begin{align}
		\mathbb{P}\left[T(X^k) \!>\! k \frac{3 \mu_t}{2} \right] \leq \mathbb{P}\left[T(X^k) \!>\! k \left(\mathbb{E}[Z_t] + \frac{\mu_t}{4}\right) \right] \leq& \frac{K_{22}}{\sqrt{k}}, \label{K_22} \\
		\mathbb{P}\left[D(X^k) \!<\! k \frac{\mu_d}{2} \right] \leq \mathbb{P}\left[D(X^k) \!<\! k \left(\mathbb{E}[Z_d] - \frac{\mu_d}{4} \right)\right] \leq& \frac{K_{23}}{\sqrt{k}}. \label{K_23}
	\end{align}
	Finally, denote $G_k = \{x^k: \big\|\boldsymbol{b}(X^k)\big\|^2 \leq (V_{b1} + V_{b2}) k \log_{e}k\}$, $J_k = \{x^k: T(X^k) \leq k \frac{3 \mu_t}{2}\}$, $O_k = \{x^k: D(X^k) \geq k \frac{\mu_d}{2}\}$ and $g(x^k) = \sum_{i=1}^k \left(\Lambda_{Z^\star Y^\star}(x_i, \boldsymbol{\lambda}(x^k)) - \Lambda_{Z^\star Y^\star}(x_i, \boldsymbol{\lambda}^{\star})\right)$. Then, let $C_2 = \frac{3\mu_t (V_{b1} + V_{b2})}{\mu_d^2 \log e}$, we have
	\begin{align}
		&\mathbb{P}\left[g(X^k) > C_2\log k \right]\\
		=&1 - \mathbb{P}\left[g(X^k) \leq C_2\log k \right]\\
		\leq& 1 - \mathbb{P}\left[X^k \in F_k\cap G_k \cap J_k \cap Q_k \right] \label{final_lemma5_01} \\
		\leq& \frac{K_1 + K_{21} + K_{22} + K_{23}}{\sqrt{k}}, \label{final_lemma5}
	\end{align}
	where \eqref{final_lemma5_01} is by \eqref{eq9} and the definitions of $G_k$, $J_k$ and $O_k$, and \eqref{final_lemma5} is by \eqref{1_sqrtn}, \eqref{K_21}, \eqref{K_22}, \eqref{K_23} and the union bound. Finally, \eqref{eq_Lemma_Lambda_logn} is obtained by letting $K_2 = K_1 + K_{21} + K_{22} + K_{23}$.
\end{proof}

\begin{remark} 
	\label{remark_core_lemma}
	By making slight adaptations to the proof of Lemma 5, we can show that Lemma \ref{Lemma_Lambda_logn} still holds after replacing \eqref{eq_Lemma_Lambda_logn} with
		\begin{align}
	\label{eq_Lemma_Lambda_logn_pro}
	&\mathbb{P}\Bigg[\sum_{i=1}^k \! \Big(\Lambda_{Z^\star Y^\star}(X_i, \boldsymbol{\lambda}(X^k)) + \boldsymbol{\lambda}^{\!\top}\!(X^k)\boldsymbol{r}\Big)\nonumber \\ 
	&\ \ \ \leq \sum_{i=1}^k \Big(\Lambda_{Z^\star Y^\star}(X_i, \boldsymbol{\lambda}^{\star}) + (\boldsymbol{\lambda}^{\star})^{\!\top} \boldsymbol{r}\Big) + C_2\log k \Bigg]\nonumber \\ 
	>& 1 - \frac{K_2}{\sqrt{k}},
	\end{align}
	where $\boldsymbol{r}$ satisfies $\displaystyle{\|\boldsymbol{r}\| = O\left(\sqrt{\frac{\log k}{k}}\right)}$.
\end{remark}

\section{Auxiliary Results for the Proof of the Converse Part of Theorem \ref{theorem_dispersion}}
\label{Auxiliary_Results}

We consider $R_X(d_s,d_x) = R_{P_X}(d_s,d_x)$ as a function of the $|\mathcal{X}|$-dimensional probability vector $P_X$, and define
\begin{equation}
R_{Q_X}(d_s,d_x) \triangleq R_{P_{\bar{X}}}(d_s,d_x),
\end{equation}
where $Q_X$ is a nonnegative $|\mathcal{M}|$-dimensional vector which may not be a probability vector, and $P_{\bar{X}}(x) = Q_X(x)/\sum_{x' \in \mathcal{X}}Q_X(x')$. Then, for each $a \in \mathcal{X}$, define the partial derivatives of $R_{P_X}(d_s,d_x)$ with respect to $P_X(a)$ as
\begin{equation}
\dot{R}_X(a, d_s,d_x) \triangleq \frac{\partial}{\partial Q_{X}(a)}R_{Q_X}(d_s,d_x)\bigg|_{Q_X=P_{X}}.
\end{equation}


\begin{lemma} 
	\label{finite_alphabet_derivative}
	Fix $(d_s,d_x) \in \mathcal{D}_{\mathrm{in}}$. Assume that the alphabet $\mathcal{X}$ is finite and for all $P_{\bar{X}}$ in some neighborhood of $P_X$, $\textsf{supp}(P_{\bar{Z}^\star \bar{Y}^\star}) = \textsf{supp}(P_{Z^\star Y^\star})$, where $P_{\bar{Z}^\star \bar{Y}^\star}$ achieves $R_{\bar{X}}(d_s,d_x)$.
	Then
	\begin{align}
	\dot{R}_X(a, d_s,d_x) =& \jmath_{X}(a,d_s,d_x) - R_X(d_s,d_x), \label{Prop2_c1} \\
	\textrm{Var}\left[\dot{R}_X(X, d_s,d_x)\right] =& \textrm{Var}\left[\jmath_{X}(X,d_s,d_x)\right]. \label{Prop2_c2}
	\end{align}
	
\end{lemma}
\begin{proof}[Proof]
	Recall that $\mathcal{D}_{\mathrm{in}} \triangleq \textsf{int}(\mathcal{D}_{sx}) \cup \textsf{int}(\mathcal{D}_{\bar{s}x}) \cup \textsf{int}(\mathcal{D}_{s\bar{x}}) \cup \textsf{int}(\mathcal{D}_{\bar{s}\bar{x}})$.
	Since, in the cases when $(d_s,d_x) \in \textsf{int}(\mathcal{D}_{\bar{s}x}) \cup \textsf{int}(\mathcal{D}_{s\bar{x}})$, both $R_X(d_s,d_x)$ and $\jmath_{X}(x,d_s,d_x)$ degenerate into scenarios previously addressed in \cite{Kostina2016Nonasymptotic}, where \eqref{Prop2_c1} and \eqref{Prop2_c2} have already been established, we now turn our attention to the remaining two cases.
	
	We first consider the case when $(d_s,d_x) \in \textsf{int}(\mathcal{D}_{\bar{s}\bar{x}})$. Define $d_{s,\max} \triangleq \min_{z \in \widehat{\mathcal{M}}} \mathbb{E}[\bar{\textsf{d}}_s(X,z)]$ and $d_{x,\max} \triangleq \min_{y \in \widehat{\mathcal{X}}} \mathbb{E}[\textsf{d}_x(X,y)]$. We note that
	\begin{align}
	\textsf{int}(\mathcal{D}_{\bar{s}\bar{x}}) = \{(d_s,d_x): d_s > d_{s,\max}, d_x> d_{x, \max}\},
	\end{align}
	and for all $(d_s,d_x) \in \textsf{int}(\mathcal{D}_{\bar{s}\bar{x}})$, we have $\lambda_s^\star = 0$, $\lambda_x^\star=0$, and $R_X(d_s,d_x) = 0$. Consequently, the right-hand side of \eqref{Prop2_c1} equals to $0$.
	Since $d_{s,\max}$ and $d_{x,\max}$ are both continuous functions of the distribution of $X$, there exists a neighborhood of $P_X$ such that for all $P_{\bar{X}}$ in it, we have $d_s > d_{s,\max}(P_{\bar{X}})$ and $d_x> d_{x, \max}(P_{\bar{X}})$ still hold, implying $R_{\bar{X}}(d_s,d_x) = 0$ in this neighborhood. Accordingly, the left-hand side of \eqref{Prop2_c1} also equals to $0$. This leads to the conclusion that \eqref{Prop2_c1} and \eqref{Prop2_c2} hold for $(d_s,d_x) \in \textsf{int}(\mathcal{D}_{\bar{s}\bar{x}})$. 
	
	
	We now consider the case when $(d_s,d_x) \in \textsf{int}(\mathcal{D}_{sx})$. First, we have \eqref{Prop2_step0}-\eqref{Prop2_step5}, 
	\begin{figure*}[!t]
		\begin{align}
		&\frac{\partial}{\partial Q_X(a)}\mathbb{E}[\jmath_{\bar{X}}(X,d_s,d_x)]\Big|_{Q_X=P_X} \label{Prop2_step0} \\
		=&\frac{\partial}{\partial Q_X(a)}\mathbb{E}[\imath_{\bar{X};\bar{Z}^\star \bar{Y}^\star}(X;Z^\star,Y^\star)] + \lambda_{s, \bar{X}} \mathbb{E}[\bar{\textsf{d}}_s(X,Z^\star) - d_s] + \lambda_{x, \bar{X}} \mathbb{E}[\textsf{d}_x(X,Y^\star) - d_x]\Big|_{Q_X=P_X} \label{Prop2_step1} \\
		=&\frac{\partial}{\partial Q_X(a)}\mathbb{E}[\imath_{\bar{X};\bar{Z}^\star \bar{Y}^\star}(X;Z^\star,Y^\star)]\Big|_{Q_X=P_X} \label{Prop2_step2}\\
		=&\frac{\partial}{\partial Q_X(a)}\mathbb{E}[\log P_{\bar{X}| \bar{Z}^\star \bar{Y}^\star}(X; Z^\star, Y^\star)]\Big|_{Q_X=P_X} -
		\frac{\partial}{\partial Q_X(a)}\mathbb{E}[\log P_{\bar{X}}(X)]\Big|_{Q_X=P_X} \label{Prop2_step3}\\
		=&\log e \cdot \frac{\partial}{\partial Q_X(a)}\mathbb{E}\left[\frac{ P_{\bar{X}| \bar{Z}^\star \bar{Y}^\star}(X; Z^\star, Y^\star)}{P_{X| Z^\star Y^\star}(X; Z^\star, Y^\star)}  \right]\Big|_{Q_X=P_X} - \log e \cdot \frac{\partial}{\partial Q_X(a)}\mathbb{E}\left[ \frac{P_{\bar{X}}(X)}{P_{X}(X)} \right]\Big|_{Q_X=P_X} \label{Prop2_step4}\\
		=&0 \label{Prop2_step5}
		\end{align}
		\hrulefill
	\end{figure*}
	where \eqref{Prop2_step1} is by \eqref{d_tilted_information_of_surrogate} and the assumption $\textrm{supp}(P_{\bar{Z}^\star \bar{Y}^\star}) = \textrm{supp}(P_{Z^\star Y^\star})$, 
	\eqref{Prop2_step2} holds since $\mathbb{E}[\bar{\textsf{d}}_s(X,Z^\star) - d_s] =0$ and $\mathbb{E}[\textsf{d}_x(X,Y^\star) - d_x]=0$ for $(d_s,d_x) \in \textsf{int}(\mathcal{D}_{sx})$,
	and \eqref{Prop2_step5} holds since the two expectation terms in \eqref{Prop2_step4} are always equal to $0$. By \eqref{Prop2_step5}, we have
	\begin{align}
	&\dot{R}_X(a, d_s,d_x)\\ =&\frac{\partial}{\partial Q_{X}(a)}\mathbb{E}[\jmath_{\bar{X}}(\bar{X},d_s,d_x)]\bigg|_{Q_X=P_{X}}\\
	=&\jmath_{X}\!(a,d_s,d_x) \!-\! R_X\!(d_s,d_x) \!+\! \frac{\partial}{\partial Q_{X}(a)}\!\mathbb{E}[\jmath_{\bar{X}}\!(X,d_s,d_x)]\bigg|_{Q_X=P_{X}}\\
	=&\jmath_{X}(a,d_s,d_x) - R_X(d_s,d_x),
	\end{align}
	completing the proof of \eqref{Prop2_c1}. \eqref{Prop2_c2} is an immediate corollary to \eqref{Prop2_c1}.
\end{proof}

\begin{proposition} 
	\label{converse_bound_corollary2}
	If $\lambda_s^\star = 0$, any $(M, d_s, d_x, \epsilon)$ code must satisfy
	\begin{equation}
	\label{converse_bound3}
	\epsilon \geq \sup_{\gamma \geq 0} \left\{\mathbb{P}\left[\jmath_{X}(X,d_s,d_x) - \log M \geq \gamma\right] - \exp(-\gamma)\right\}.
	\end{equation}
\end{proposition}
\begin{proof}[Proof]
	By setting $P_{\bar{X}|\bar{Z}\bar{Y}} = P_{X|Z^\star Y^\star}$, $\lambda_s = \lambda_s^\star$, and $\lambda_x = \lambda_x^\star$ in \eqref{general_converse}, any $(M, d_s, d_x, \epsilon)$ code must satisfy
	\begin{align}
	\epsilon \geq &\inf_{\substack{P_{ZY|X}:\\ \mathcal{X} \to \widehat{\mathcal{S}}\times \widehat{\mathcal{X}}}} \sup_{\gamma \geq 0}\big\{\mathbb{P}\big[\imath_{X; Z^\star Y^\star}(x;z,y) + \lambda_s^\star(\textsf{d}_s(s,z) - d_s)\nonumber\\
	&+ \lambda_x^\star (\textsf{d}_x(x,y) - d_x) - \log M \geq \gamma\big]
	- \exp(-\gamma)\big\}. \label{converse_C2_1}
	\end{align}
	Since $\lambda_s^\star = 0$, we can replace the term $\lambda_s^\star(\textsf{d}_s(s,z) - d_s)$ in \eqref{converse_C2_1} by $\lambda_s^\star(\bar{\textsf{d}}_s(x,z) - d_s)$. Combining this with \eqref{d_tilted_information_of_surrogate}, we have
	\begin{align}
	\epsilon \geq& \inf_{\substack{P_{ZY|X}:\\ \mathcal{X} \to \widehat{\mathcal{S}}\times \widehat{\mathcal{X}}}}\!\! \sup_{\gamma \geq 0} \left\{\mathbb{P}\left[\jmath_{X}(X,d_s,d_x) \!-\! \log M \geq \gamma\right] \!-\! \exp(-\gamma)\right\}\\
	=& \sup_{\gamma \geq 0} \left\{\mathbb{P}\left[\jmath_{X}(X,d_s,d_x) - \log M \geq \gamma\right] - \exp(-\gamma)\right\}.
	\end{align}
\end{proof}




\bibliographystyle{IEEEtran}
\bibliography{ref}

\end{document}